%% file: WJH18_ieee_com.tex
\PassOptionsToPackage{}{graphicx}
\PassOptionsToPackage{usenames,dvipsnames,table}{xcolor}
\documentclass[onecolumn,12pt,draftcls]{IEEEtran}  

\usepackage{amsmath}
\usepackage{etex} 

\usepackage{amssymb}
\usepackage{amsxtra,amsthm}
\usepackage{wasysym}
\usepackage{bbold}
\usepackage{amsfonts} 
\usepackage{mathdots}
\usepackage{mathtools}
\usepackage{thmtools}
\usepackage{pdfcolmk}
\usepackage[pdftex]{graphicx}                                        
\usepackage{algorithm}
\usepackage{algorithmic} 
\usepackage{mathrsfs}
\usepackage{todonotes}

\usepackage[utf8]{inputenc}                                         
\usepackage[T1]{fontenc}
\usepackage[pdfborder={0 0 0},colorlinks,citecolor=blue,hyperindex=true,hypertexnames=false,extension=pdf]{hyperref} 
\usepackage{subcaption} 
\usepackage{verbatim}
\usepackage{float}
\usepackage{srcltx}
\usepackage[english]{babel}

\usepackage[shortlabels]{enumitem}
\usepackage[all]{xy}

\usepackage{rotating}
\usepackage{tocloft}
\usepackage{import}

\usepackage[style=ieee,    
            firstinits=true, 
            url=false,
            backend=bibtex,
            doi=false,
            maxbibnames=99,
            isbn=false,
            natbib=true,
            hyperref=true,
            bibencoding=utf8]{biblatex}
\AtEveryBibitem{\clearfield{note}}   
\DefineBibliographyStrings{english}{%
  references = {},
}

\def\pdfdir{}
\include{Header_Philipp}

\newcommand{\tdmin}{\ensuremath{\tilde{\dmin}}}

\newtheorem{thm}{Theorem}
\newtheorem{lemi}{Lemma}
\newtheorem{cori}{Corollary}

\ifx\definition\undefined
\newtheorem{definition}{Definition}         
\fi
\ifx\corrolary\undefined
\newtheorem{corrolary}{Corrolary} 
\fi
\ifx\conjecture\undefined
\newtheorem{conjecture}{Conjecture}         
\fi
\ifx\thm\undefined
\newtheorem{thm}{Theorem}         
\fi
\ifx\lemma\undefined
\newtheorem{lemma}{Lemma}
\fi
\ifx\question\undefined
\newtheorem{question}{\color{red}{Question}}         
\fi
\ifx\hypothesis\undefined
\newtheorem{hypothesis}{Hypothesis}
\fi

\ifx\remark\undefined
\newenvironment{remark}{\par\vspace{1.5ex}\noindent{\em Remark\/}.}{\par\vspace{1.5ex}}
\fi

\ifx\approach\undefined
\newcounter{theapproach}
\newenvironment{approach}[1][]{\par\vspace{1.5ex}\noindent{\em Approach\/ }\arabic{theapproach}. #1
\stepcounter{theapproach}}{\par\vspace{1.5ex}}
\fi

\newif\ifold\oldfalse
\newif\ifhuffmanall\huffmanallfalse
\newif\ifBivariate\Bivariatefalse
\newif\ifyonina\yoninafalse
\newif\ifjohn\johnfalse
\newif\ifmattia\mattiafalse
\newif\ifchernoff\chernofffalse
\newif\ifdzd\dzdfalse
\newif\ifHufMinDis\HufMinDisfalse

\newif\ifdetails\detailstrue 
\newif\ifalt\altfalse 
\newif\ifseefor\seeforfalse 
\newif\ifRealvsComplexProblem\RealvsComplexProblemfalse 
\newif\ifstabilitymasks\stabilitymasksfalse 
\newif\iflistofstuff\listofstufffalse 
\newif\iflongversion\longversionfalse
\newif\ifaminenot\aminenottrue 

\newif\ifWirtinger\Wirtingerfalse

\newif\ifnotforbabak\notforbabaktrue
\newif\ifShaCip\ShaCipfalse

\newcommand{\unsurelist}[1]{%
  \refstepcounter{unsures}%
  \addcontentsline{uns}{unsures}%
  {\protect\numberline{\thechapter.\theunsures\ }\hspace{1em} #1}
}

\newcommand{\unsure}[1]{%
  \unsurelist{#1}\color{brown}\!#1\color{black}}

\newcommand{\wrong}[1]{}

\newcommand{\unsurenolist}[1]{}

\newcommand{\unsurenote}[2]{}

\newcommand{\seefor}[1]{\!}
\newcommand{\seeintern}[1]{\!}

\usepackage{filecontents}
\makeatletter
\@ifpackageloaded{tikz}{
\usetikzlibrary{positioning,decorations.pathreplacing}

  \setcounter{MaxMatrixCols}{18}

}{}
\makeatother

\ifx\unsure\undefined
 \newcommand{\unsure}[1]{}
\else
 \renewcommand{\unsure}[1]{} 
\fi
\ifx\wrong\undefined
  \newcommand{\wrong}[1]{}
\else
  \renewcommand{\wrong}[1]{}
\fi
  
\ifx\unsurenolist\undefined
  \newcommand{\unsurenolist}[1]{}
\else
  \renewcommand{\unsurenolist}[1]{}
\fi

\ifx\unsurenote\undefined
  \newcommand{\unsurenote}[2]{}
\else
  \renewcommand{\unsurenote}[2]{}
\fi

\ifdetails
  \newcommand{\detail}[1]{\color{brown}{#1}\color{black}}
\else
\newcommand{\detail}[1]{}
\fi

\renewcommand{\alp}{\ensuremath{\alpha}}

\newcommand{\Ring}{\ensuremath{{\mathscr{R}}}}

\renewcommand{\CN}{\ensuremath{\mathcal{CN}}}

\newcommand{\zmin}{\ensuremath{z_{\text{min}}}}
\newcommand{\zmax}{\ensuremath{z_{\text{max}}}}
\newcommand{\zcenter}{\ensuremath{z^{\text{c}}}}
\newcommand{\zvertex}{\ensuremath{z^{\text{v}}}}
\newcommand{\phimax}{\ensuremath{\phi_{\text{max}}}}

\newcommand{\nmax}{\ensuremath{n_{\text{max}}}}

\newcommand{\cdmin}{\ensuremath{\tilde{d}_{\text{min}}}}

\altfalse
\aminenotfalse
\oldtrue

\newcommand*\xbar[1]{%
   \hbox{%
     \vbox{%
       \hrule height 0.5pt 
       \kern-0.1ex
       \hbox{%
         \kern-0.0em
         \ensuremath{ #1}%
         \kern-0.1em
       }%
     }%
   }%
} 
\newcommand*\sxbar[1]{%
   \hbox{%
     \vbox{%
       \hrule height 0.5pt 
       \kern-0.3ex
       \hbox{%
         \kern-0.0em
         \ensuremath{\scriptstyle #1}%
         \kern-0.1em
       }%
     }%
   }%
}

\newcommand{\vValp}{\ensuremath{\vV_{\!\valp}}}
\newcommand{\Ropt}{\ensuremath{R_{\text{uni}}}} 

\newcommand{\Nh}{{\ensuremath{L}}}
\newcommand{\Nx}{{\ensuremath{K}}}
\newcommand{\nx}{{\ensuremath{k}}}
\newcommand{\nh}{{\ensuremath{l}}}
\newcommand{\pd}{\ensuremath{p}} 

\usepackage{multirow}
\def\minus{%
  \setbox0=\hbox{-}%
  \vcenter{%
    \hrule width\wd0 height \the\fontdimen8\textfont3%
  }%
}
\newif\iflong\longfalse
\newif\ifextras\extrasfalse 
\newif\ifwrong\wrongfalse

\renewcommand{\ve}{\ensuremath{{\mathbf e}}}

\newcommand{\pwprod}{\ensuremath{\bullet}}
\newcommand{\ZeroSet}{\ensuremath{Z}}

\usepackage{scalerel} 

\newcommand{\tuY}{{\ensuremath{\tilde{\uY}}}}
\newcommand{\rSNR}{{\ensuremath{\text{rSNR}}}}

\newcommand{\todostart}{\color{brown}} \newcommand{\todoend}{\color{black}} 

\newcommand{\delmax}{\ensuremath{\del_{\text{max}}}}
\newcommand{\Peter}[1]{\textcolor{red}{#1}}
\renewcommand{\Peter}[1]{{}}

\newcommand{\xord}{\ensuremath{N}} 
\newcommand{\xind}{\ensuremath{n}} 
\newcommand{\alpmax}{\ensuremath{\alp_{\text{\rm max}}}} 
\newcommand{\const}{\ensuremath{\text{const}}} 

\newif\ifall\allfalse
\newif\ifarxiv\arxivfalse

\begin{document}
  %
  %
  \title{
  Noncoherent Short-Packet Communication via Modulation on Conjugated Zeros}

  \author{
    \IEEEauthorblockN{Philipp Walk\IEEEauthorrefmark{1}\IEEEauthorrefmark{3}, Peter Jung\IEEEauthorrefmark{2}, 
    and Babak Hassibi\IEEEauthorrefmark{3}\\}
   \IEEEauthorblockA{\IEEEauthorrefmark{1}Dept. of Electrical Engineering \& Computer Science, UCI, Irvine, CA 92697\\
    Email: pwalk@uci.edu}\\
  \IEEEauthorblockA{\IEEEauthorrefmark{2}Communications \& Information Theory, TU Berlin, 10587 Berlin\\
   Email: peter.jung@tu-berlin.de}\\
   \IEEEauthorblockA{\IEEEauthorrefmark{3}Dept. of Electrical Engineering, Caltech, Pasadena, CA 91125\\
    Email: hassibi@caltech.edu }
  } 

  \newcommand{\vSig}{\ensuremath{\boldsymbol{\Sig}}}

  \maketitle

  \begin{abstract}
    We introduce a novel blind (noncoherent) communication scheme, called modulation on conjugate-reciprocal zeros
    (MOCZ), to reliably transmit short binary packets over unknown finite impulse response systems as used, for example,
    to model underspread wireless multipath channels.  In MOCZ, the information is modulated onto the zeros of the
    transmitted signals $z-$transform.  In the absence of additive noise, the zero structure of the signal is perfectly
    preserved at the receiver, no matter what the channel impulse response (CIR) is. Furthermore, by a proper selection
    of the zeros, we show that MOCZ is not only invariant to the CIR, but also robust against additive noise. Starting
    with the maximum-likelihood estimator, we define a low complexity and reliable decoder and compare it to various
    state-of-the art noncoherent schemes. 
  \end{abstract}

  \section{introduction} The future generation of wireless networks faces a diversity of new challenges.  Trends on the
  horizon -- such as the emergence of the Internet of Things (IoT) and the tactile Internet -- have radically changed
  our thinking about how to scale the wireless infrastructure.  Among the main challenges new emerging technologies have
  to cope with is the support of a massive number (billions) of devices ranging from powerful smartphones and tablet
  computers to small and low-cost sensor nodes. These devices come with diverse and even contradicting types of traffic
  including high speed cellular links, device-to-device connections, and wireless links carrying short-packet sensor
  data.  Short messages of sporadic nature \cite{Wunder2015:sparse5G} will dominate in the future and the conventional
  cellular and centrally-managed wireless network infrastructure will not be flexible enough to keep pace with these
  demands. Although intensively discussed in the research community, the most fundamental question here on how we will
  communicate in the near future under such diverse requirements remains largely unresolved.  A key problem is how to
  acquire, communicate, and process channel information.  Conventional channel estimation procedures require a
  substantial amount of resources and overhead. This overhead can  dominate the intended information exchange when the
  message is short and the traffic sporadic.  \emph{Noncoherent and blind strategies}, provide a potential way out of
  this dilemma.  Classical approaches like blind equalization have been already investigated in the engineering
  literature \cite{Godard1980,For72,CP96}, but new noncoherent modulation ideas which explicitly account for the
  short-message and sporadic type of data are required \cite{Jung2014}.

  In many wireless communication scenarios the transmitted signals are affected by multipath propagation and the channel
  will therefore be frequency-selective. Additionally, in mobile and time-varying scenarios one encounters also
  time-selective fast fading. In both cases channel parameters typically have a random flavour and potentially cause
  various kinds of interference. From a signal processing perspective it is therefore necessary to take care of
  possible signal distortions, at the receiver and potentially also at the transmitter. A well know approach to deal
  with such channels is to modulate data on multiple parallel waveforms which are well-suited for the particular channel
  conditions. One of the most simple approaches for the frequency-selective case is orthogonal frequency division
  multiplexing (OFDM). When the maximal channel delay spread is known inter-symbol-interference (ISI) can be avoided by a
  suitable guard interval and an orthogonality of the subcarriers ensures that there is no
  interference-carrier-interference.  On the other hand, from an information-theoretic perspective, random
  channel parameters are helpful from a diversity view point. To exploit multipath diversity the data has to
  be spread over the subcarriers. To coherently demodulate the data at the receiver and to also make use of diversity
  the channel impulse response (CIR) has to be known at the receiver.  To gain knowledge of the CIR training symbols
  (pilots) are included in the transmitted signal, leading to a substantial overhead when the signal length is on the
  order of the channel length. Furthermore, the pilot density has to be adapted to the mobility and, in
  particular, OFDM is very sensitive to time-varying distortions due to Doppler shift and oscillator instabilities. Dense
  CIR updates are then required, which results in complex transceiver designs.

  There are only a few works on noncoherent OFDM schemes in the literature.  Some are known as self-heterodyne OFDM or
  self-coherent OFDM \cite{JHV15,FHV13}. Very recently a noncoherent method for OFDM with Index Modulation (IM) was
  proposed in \cite{Cho18}, which exploits a sparsity of $\Nh$ subcarriers out of $N$. The modulation can be seen as a
  generalized $N-ary$ frequency shift keying (FSK), which uses $\Nh$ tones (frequencies) and results in a codebook of
  $M=\binom{N}{\Nh}$ non-orthogonal constellations.  In this work we follow a completely different strategy. We propose
  to encode each bit of the data payload into one of a conjugate-reciprocal pairs of zeros (in the complex plane) and
  thereby construct a polynomial whose degree is the number of payload bits.  The complex-valued coefficients of the
  polynomial are in fact the transmit baseband signal samples. We introduced such a non-linear modulation on polynomial
  zeros first for Huffman sequences in \cite{WJH17a,WJH17b} and demonstrated to perform efficient and reliable convex
  and non-convex decoding algorithms.  However, such optimization algorithms are meant for blind deconvolution, i.e.,
  reconstruct channel and signal simultaneously, and are therefore not necessarily well-suited and efficient to retrieve
  the digital data from finite alphabets.

  In this work we will therefore extend our previous ideas and develop and analyze polynomial-factorization-based
  approaches more concretely from a communication-oriented perspective. We will extend the modulation and encoding
  principle to general codebooks based on polynomial zeros. We derive and analyse the maximum likelihood decoder,
  which depends only on the power delay profile of the channel and the noise power.  Then, we construct a low complexity
  decoder for Huffman sequences having a complexity which scales only linearly in the number of bits to transmit.  We
  will demonstrate by numerical experiments that our scheme is able to outperform noncoherent OFDM-IM and pilot based
  $M-$QAM schemes in terms of bit-error rate.

  \subsection{Notation} We will use small letters for complex numbers in $\C$.  Capital Latin letters denote natural
  numbers and refer to fixed dimensions, where small letters are used as indices.  Boldface small letters denote vectors
  and capitalized letters refer to matrices. Upright capital letters denote complex-valued polynomials in $\C[z]$. For a
  complex number $x=a+\im b$, given by its real part $\Re(x)=a\in\R$ and imaginary part $\Im(x)=b\in\R$ with imaginary
  unit $\im=\sqrt{-1}$, its complex-conjugation is given by $\cc{x}=a-\im b$ and its absolute value by
  $|x|=\sqrt{x\cc{x}}$.  For a vector $\vx\in\C^N$ we denote by $\cc{\vx^-}$ its complex-conjugated time-reversal or
  \emph{conjugated-reciprocal}, given as $\cc{x_k^-} = \cc{x_{N-k}}$ for $k=0,1,\dots,N-1$. We use $\vA^*=\cc{\vA}^T$
  for the complex-conjugated transpose of the matrix $\vA$. For the identity and all zero matrix in $N$ dimension we
  write $\id_N$ respectively $\vzero_N$. By $\vD_{\vx}$ we refer to the diagonal matrix generated by the vector
  $\vx\in\C^N$.  The $N\times N$ unitary Fourier matrix $\Fmatrix=\Fmatrix_N$ is given entry-wise by $f_{l,k}=e^{\im
  2\pi lk/N}/\sqrt{N}$ for $l,k=0,1,\dots,N-1$.  The all one respectively all zero vector in dimension $N$ will be
  denoted by $\eins_N$ resp. $\zero_N$.  The
  $\ell_p$-norm of a vector $\vx\in\C^N$ is given by $\Norm{\vx}_p=(\sum_{k=1}^N|x_k|^p)^{1/p}$ for $p\geq 1$.  If
  $p=\infty$ we write $\Norm{\vx}_{\infty}=\max_k |x_k|$.  The expectation of a random variable $x$ is denoted by
  $\Expect{x}$.  We will refer to $\vx\bullet\vy:=\diag(\vx)\vy$ as the Hadamard (point-wise) product of the vectors
  $\vx, \vy\in\C^N$.

  \input{dezet}

  \input{generalzerocodebooks}

\input{robustness}

\input{simulations}
  \ifarxiv \input{MaryMOCZ}\fi
  \input{training}

  \input{stabilitybound}

  \section{Conclusion}

We introduced a novel modulation scheme based on the zeros of the discrete-time signals to transmit reliable over unknown FIR
channels. For the demodulation we presented three different approaches. The first based on a zero observation of the
received zeros and deciding by a minimum distance decoder for each single bit (zero) independently. Secondly, a maximum
likelihood decoder, which obtains the best performance. If the CIR length is larger than the block length, the ML
decoder for Huffman BMOCZ outperforms all comparable known non-coherent signaling schemes.  However, this decoder relies
on the knowledge of the channels power delay profile and the SNR. Finally, we introduced a low-complexity
decoder which decodes the zeros independently by only evaluating the received signal on the zero-codebook, which leads
to linear complexity in the number of bits, instead of exponential complexity for the ML decoder.  The derivation of bit
error probabilities is mathematically hard to carry out, not only due to the overlap of the signals caused by multipath
delays, but also in terms of the non-linear encoding in the zero domain. For the RFMD decoder we obtained a local
stability analysis in the presence of additive noise, which suggest a proper zero separation of the codebook and
channel. The analysis of reliable bit data rates and error probability bounds, based on a careful root neighborhood
analysis, might be addressed in a future research.  

\section{Acknoledgements}

The authors would like to thank Richard Kueng and Urbashi Mitra for many helpful discussions. We like to thank the SURF
student Mattia Carrera for his support during a summer program at Caltech. Most of the work by Philipp Walk was done
during a two year postdoc fellowship at Caltech, which was sponsored by the DFG WA 3390/1.

   \section*{References}
  \printbibliography[heading=bibintoc]
  \appendices
  \input{mgon}

\end{document}

%% file: Header_Philipp.tex
\usepackage{tensor} 

\usepackage{cleveref} 

\newcommand{\mgeq}{\succeq}

\newcommand{\xone}{\ensuremath{x_1}}
\newcommand{\xtwo}{\ensuremath{x_2}}

\newcommand{\versine}{\ensuremath{\operatorname{versine}}}

\newcommand{\diag}{\operatorname{diag}}

\let\forallalt\forall
\renewcommand{\forall}{\;\forallalt\;}

\let\refalt\ref
\renewcommand{\ref}[1]{(\refalt{#1})}

\newcommand{\alp}{\ensuremath{\alpha}}
\newcommand{\bet}{\ensuremath{\beta}}
\newcommand{\Del}{\ensuremath{\Delta}}
\newcommand{\del}{\ensuremath{\delta}}
\newcommand{\eps}{\ensuremath{\epsilon}}
\newcommand{\Gam}{\ensuremath{\Gamma}}
\newcommand{\gam}{\ensuremath{\gamma}}

\newcommand{\lam}{\ensuremath{\lambda}}

\newcommand{\ome}{\ensuremath{\omega}}
\newcommand{\Sig}{\ensuremath{\Sigma}}
\newcommand{\sig}{\ensuremath{\sigma}}
\newcommand{\tht}{\ensuremath{\theta}}
\newcommand{\zet}{\ensuremath{\zeta}}

\newcommand{\vGam}{\ensuremath{\mathbf{\Gam}}}

\newcommand{\vm}{\ensuremath{\mathbf m}}

\newcommand{\va}{{\ensuremath{\mathbf a}}}
\newcommand{\vb}{{\ensuremath{\mathbf b}}}
\newcommand{\vc}{{\ensuremath{\mathbf c}}}
\newcommand{\vd}{{\ensuremath{\mathbf d}}}
\newcommand{\ve}{{\ensuremath{\boldsymbol{\del}}}} 
\newcommand{\vf}{{\ensuremath{\mathbf f}}}

\newcommand{\vh}{{\ensuremath{\mathbf h}}}                         
\newcommand{\vn}{{\ensuremath{\mathbf n}}}                         
\newcommand{\vp}{{\ensuremath{\mathbf p}}}

\newcommand{\vr}{{\ensuremath{\mathbf r}}}                         

\newcommand{\vu}{{\ensuremath{\mathbf u}}}

\newcommand{\vw}{{\ensuremath{\mathbf w}}}
\newcommand{\vx}{{\ensuremath{\mathbf x}}}
\newcommand{\vy}{{\ensuremath{\mathbf y}}}

\newcommand{\vz}{{\ensuremath{\mathbf z}}}

\newcommand{\vta}{\ensuremath{\tilde{\mathbf a}}}

\newcommand{\vth}{\ensuremath{\tilde{\mathbf h}}}

\newcommand{\vtx}{\ensuremath{\tilde{\mathbf x}}}

\newcommand{\hvh}{{\ensuremath{\hat{\mathbf h}}}}

\newcommand{\hvx}{{\ensuremath{\hat{\mathbf x}}}}

\newcommand{\hvz}{{\ensuremath{\hat{\mathbf z}}}}

\newcommand{\uA}{{\ensuremath{\mathrm A}}}

\newcommand{\uX}{{\ensuremath{\mathrm X}}}
\newcommand{\uY}{{\ensuremath{\mathrm Y}}}

\newcommand{\ux}{{\ensuremath{\mathrm x}}}

\newcommand{\uy}{{\ensuremath{\mathrm y}}}

\newcommand{\uH}{{\ensuremath{\mathrm H}}}

\newcommand{\uP}{{\ensuremath{\mathrm P}}}

\newcommand{\vA}{{\ensuremath{\mathbf A}}}
\newcommand{\vB}{{\ensuremath{\mathbf B}}}
\newcommand{\vC}{{\ensuremath{\mathbf C}}}
\newcommand{\vD}{{\ensuremath{\mathbf D}}}

\newcommand{\vM}{\ensuremath{\mathbf M}}
\newcommand{\vP}{\ensuremath{\mathbf P }}                         

\newcommand{\vR}{\ensuremath{\mathbf R}}
\newcommand{\vS}{\ensuremath{\mathbf S }}                         
\newcommand{\vT}{\ensuremath{\mathbf T}}

\newcommand{\vV}{{\ensuremath{\mathbf V}}}

\newcommand{\vX}{{\ensuremath{\mathbf X}}}

\newcommand{\vZ}{{\ensuremath{\mathbf Z}}}

\newcommand{\vtA}{{\ensuremath{\tilde{\mathbf A}}}}

\newcommand{\vtX}{\ensuremath{\tilde{\mathbf X}}}

\newcommand{\valp}{\ensuremath{\boldsymbol{ \alp}}}

\newcommand{\zero}{{\ensuremath{\mathbf 0}}}

\newcommand{\vzero}{{\ensuremath{\mathbf O}}}

\newcommand{\tx}{\ensuremath{\tilde{x}}}

\newcommand{\tN}{\ensuremath{\tilde{N}}}

\newcommand{\talp}{\ensuremath{\tilde{\alp}}}

\newcommand{\tdel}{\ensuremath{\tilde{\del}}}

\newcommand{\Calg}{\ensuremath{\mathscr C}}
\newcommand{\Code}{\Calg}

\newcommand{\Pset}{\ensuremath{\mathscr P}}
\newcommand{\Sset}{\ensuremath{\mathscr S}}

\newcommand{\Aset}{{\ensuremath{\mathcal{A}}}}

\newcommand{\SNR}{\ensuremath{{\text{SNR}}}}

\newcommand{\PAPR}{\ensuremath{{\text{PAPR}}}}

\newcommand{\tvp}{\ensuremath{\tilde{\vp}}}

\newcommand{\Alin}{{\ensuremath{\mathcal{A}}}}

\newcommand{\C}{{\ensuremath{\mathbb C}}}

\newcommand{\R}{{\ensuremath{\mathbb R}}}

\newcommand{\N}{{\ensuremath{\mathbb N}}}

\newcommand{\CN}{\ensuremath{\C^N}}

\newcommand{\Fmatrix}{{\ensuremath{\mathbf F}}}
\newcommand{\tFmatrix}{{\ensuremath{\tilde{\mathbf F}}}}
\newcommand{\Fmatrixa}{{\ensuremath{\mathbf F}^*}}

\newcommand{\Zform}{{\ensuremath{\mathcal Z}}}

\newcommand{\id}{{\ensuremath{\mathbf I}}}

\newcommand{\eins}{{\ensuremath{\mathbf 1}}}

\newcommand{\RA}{\ensuremath{\Rightarrow} }
\newcommand{\LA}{\ensuremath{\Leftarrow} }
\newcommand{\ra}{\rightarrow}

\newcommand{\LRA}{\ensuremath{\Leftrightarrow} }

\newcommand{\dmin}{{\ensuremath{d_{\text{\rm min}} } }}

\newcommand{\Pro}{\prod}

\newcommand{\skprod}[2]{\ensuremath{ \left\langle #1,#2 \right\rangle }}

\definecolor{gray}{rgb}{0.3,0.3,0.3}
{\color{black}}                      
{\color{black}}

{\color{black}}                      
{\color{black}}

\newcommand{\thmref}[1]{Theorem~\ref{#1}}     
\newcommand{\lemref}[1]{Lemma~\ref{#1}}       
\newcommand{\conref}[1]{Conjecture~\ref{#1}}  

\newcommand{\appref}[1]{Appendix~\ref{#1}}
\newcommand{\secref}[1]{Section~\ref{#1}}

\newcommand{\figref}[1]{Figure~\ref{#1}}

\newcommand{\noi}{\noindent}

\newcommand\rank{\operatorname{rank}}

\newcommand{\argmin}[1]{\underset{#1}{\operatorname{argmin}}}

\newcommand{\Norm}[1]{\ensuremath{ \left\|#1\right\| }}

\newcommand{\Expect}[1]{{\ensuremath{\mathbb E}[#1]}}

\newcommand{\cc}[1]{{\ensuremath{\overline{#1}}}} 

\makeatletter
  \newcommand{\set}[2]{\ensuremath{%
  \setbox0=\hbox{\ensuremath{#2}}
  \dimen@\ht0
  \advance\dimen@ by \dp0
  \left\{\left.#1\rule[-\dp0]{0pt}{\dimen@}\;\right|\;#2\right\} }}
\makeatother

\renewcommand{\argmin}{\operatornamewithlimits{argmin}}

\ifx \@paragraph \@empty
\makeatletter
\renewcommand\paragraph{\@startsection
{paragraph}{4}{\z@}{-3.5ex plus-1ex minus-.2ex}%
{1.3ex plus.2ex}{\normalfont\itshape}}
\makeatother
\fi

\renewcommand{\Re}{\ensuremath{\operatorname{Re}}}
\renewcommand{\Im}{\ensuremath{\operatorname{Im}}}

\newcommand{\one}{\ensuremath{1}} 

\DeclareFontFamily{U}{matha}{\hyphenchar\font45}
\DeclareFontShape{U}{matha}{m}{n}{
      <5> <6> <7> <8> <9> <10> gen * matha
      <10.95> matha10 <12> <14.4> <17.28> <20.74> <24.88> matha12
      }{}
\DeclareSymbolFont{matha}{U}{matha}{m}{n}
\DeclareFontSubstitution{U}{matha}{m}{n}

\DeclareFontFamily{U}{mathx}{\hyphenchar\font45}
\DeclareFontShape{U}{mathx}{m}{n}{
      <5> <6> <7> <8> <9> <10>
      <10.95> <12> <14.4> <17.28> <20.74> <24.88>
      mathx10
      }{}
\DeclareSymbolFont{mathx}{U}{mathx}{m}{n}
\DeclareFontSubstitution{U}{mathx}{m}{n}

\DeclareMathDelimiter{\vvvert}{0}{matha}{"7E}{mathx}{"17}

\ifx\oast\undefined 
\newcommand{\oast}{\ensuremath{\circledast}}
\fi
\newcommand{\mleq}{\ensuremath{\preceq}}

\newcommand{\etaE}{\ensuremath{R_b}} 
\newcommand{\uW}{{\rm{W}}}
\newcommand{\im}{\ensuremath{j}} 

\ifdefined\Zero
\renewcommand{\Zero}{\ensuremath{\mathscr Z}}
\else
\newcommand{\Zero}{\ensuremath{\mathscr Z}}
\fi
\newcommand{\Td}{\ensuremath{T_d}} 
\newcommand{\vtXp}{\ensuremath{\tilde{\mathbf X}_{\vp}}} 
\newcommand{\vXp}{\ensuremath{{\mathbf X}_{\vp}}} 
\newcommand{\vtD}{\ensuremath{\tilde{\mathbf D}_{\vp}}} 
\newcommand{\vtB}{\ensuremath{\tilde{\mathbf B}}} 

\newcommand{\vtXtp}{\ensuremath{\tilde{\mathbf X}_{\tvp}}} 
\newcommand{\uw}{\ensuremath{{\rm w}}} 
\ifdefined\Dset
\renewcommand{\Dset}{\ensuremath{\mathfrak D}} 
\else
\newcommand{\Dset}{\ensuremath{\mathfrak D}} 
\fi
\ifdefined\Sset
\renewcommand{\Sset}{\ensuremath{\mathfrak S}} 
\else
\newcommand{\Sset}{\ensuremath{\mathfrak S}} 
\fi
\newcommand{\Alp}{\ensuremath{\mathscr Z}} 

\newcommand{\veta}{\ensuremath{{\boldsymbol \eta}}}

\newcommand{\radl}{\ensuremath{R_\nx^{-1}}}
\newcommand{\radlinv}{\ensuremath{R_\nx}}
\newcommand{\phil}{\ensuremath{\phi_\nx}}
\newcommand{\Dsetoneone}{\ensuremath{\Dset_1^{(1)}}} 
\newcommand{\Dsetonezero}{\ensuremath{\Dset_1^{(0)}}} 
\newcommand{\Dsettwoone}{\ensuremath{\Dset_2^{(1)}}} 
\newcommand{\Dsettwozero}{\ensuremath{\Dset_2^{(0)}}} 
\DeclareMathOperator{\nullity}{nullity}
\renewcommand{\ux}{\ensuremath{\uX}} 
\renewcommand{\uy}{\ensuremath{\uY}} 
\renewcommand{\uw}{\ensuremath{\uW}} 
\renewcommand{\vn}{\ensuremath{\vw}} 
\newcommand{\yzero}{\ensuremath{y_0}}
\newcommand{\yK}{\ensuremath{y_{\Nh\!-\!1}}}
\newcommand{\yKone}{\ensuremath{y_{\Nh}}}
\newcommand{\yKtwo}{\ensuremath{y_{\Nh\!+\!1}}}
\newcommand{\yN}{\ensuremath{y_{N\!-\!1}}}
\newcommand{\xL}{\ensuremath{x_\Nx}}
\newcommand{\xzero}{\ensuremath{x_0}}
\newcommand{\gamtwo}{\ensuremath{\zeta_2}}
\newcommand{\gamone}{\ensuremath{\zeta_1}}
\newcommand{\gamN}{\ensuremath{\zeta_{N\!-\!1}}}
\newcommand{\alpone}{\ensuremath{\alpha_1}}
\newcommand{\alptwo}{\ensuremath{\alpha_2}}
\newcommand{\alpthree}{\ensuremath{\alpha_3}}

\newcommand{\alpnine}{\ensuremath{\alpha_9}}
\newcommand{\alpten}{\ensuremath{\alpha_{10}}}
\newcommand{\alpeleven}{\ensuremath{\alpha_{11}}}
\newcommand{\alpfiften}{\ensuremath{\alpha_{15}}}
\newcommand{\alpsixten}{\ensuremath{\alpha_{16}}}
\newcommand{\alpseventen}{\ensuremath{\alpha_{17}}}

\newcommand{\alpL}{\ensuremath{\alpha_\Nx}}
\newcommand{\alpl}{\ensuremath{\alpha_\nx}}
\newcommand{\halptwo}{\ensuremath{\hat{\alpha}_2}}
\newcommand{\halpone}{\ensuremath{\hat{\alpha}_1}}
\newcommand{\halpL}{\ensuremath{\hat{\alpha}_\Nx}}
\newcommand{\rtwo}{\ensuremath{r_2}}
\newcommand{\rthree}{\ensuremath{r_3}}
\newcommand{\mtwo}{\ensuremath{m_2}}
\newcommand{\rone}{\ensuremath{r_1}}
\newcommand{\mone}{\ensuremath{m_1}}
\newcommand{\mL}{\ensuremath{m_\Nx}}
\newcommand{\hmtwo}{\ensuremath{\hat{m}_2}}
\newcommand{\hmone}{\ensuremath{\hat{m}_1}}
\newcommand{\hmL}{\ensuremath{\hat{m}_\Nx}}
\newcommand{\Vitae}{\ensuremath{\mathscr{V}}}

%% file: dezet.tex
\section{Channel Model}

In this work we will consider communication over frequency-selective block-fading channels used for indoor and outdoor
scenarios, where the channel delay spread $\Td$ is in the order of the signal duration $T_s=NT$, given by the symbol
duration $T$ and overall block length $N$. We assume that the channel is time-invariant in each block, but changes
arbitrary from block to block, which models a time-varying channel \cite{VHHK01}. Conventional coherent communication
strategies, e.g., most based on OFDM, are expected to be inefficient in this regime.  We will therefore propose (in the
next section) a novel modulation scheme for noncoherent communication, which keeps the relevant information in the
transmitted signal invariant under multipath propagation and therefore completely avoids channel estimation and signal
equalization at the receiver. Assuming that the CIR remains constant over the one-shot (block) communication period, the
discrete-time baseband model for this channel is given as a linear \emph{convolution}:
\begin{align}
  y_n = \sum_{\nh=0}^{\Nh-1} h_\nh x_{n-\nh} + w_n \quad\text{for}\quad n\in\{0,1,2,\dots,N\},
\end{align}
of the transmitted time symbols $\{x_n\}_{n=0}^{K}$ with the complex-valued channel coefficients (taps)
$\{h_\nh\}_{\nh=0}^{L-1}\in\C$ resulting in a block of $N=L+K$ received symbols.  Additionally, the convolution is
disturbed by additive noise $w_n$. We denote the block (packet) of $\Nx+1$ transmitted time symbols as the vector
$\vx=(x_0,x_1,\dots,x_\Nx)^T\in\C^{\Nx+1}$ and assume wlog a normalization $\Norm{\vx}_2^2=\sum_\nx |x_\nx|^2=1$.
%
%
In this form, we obtain at the receiver the vector:
\begin{align}
  \vy= \vx * \vh + \vw\in \mathbb{C}^{N}\label{eq:receivedsignal}.
\end{align}
Contrary to usual assumptions, we assume that only one packet $\vx$ is transmitted, which is called a ``one-shot''
communication. Here, the next transmission will be at an indefinite time point such that it is not possible to predict
the CIR. Such a \emph{sporadic} transmission scheme can therefore be seen as a prototype problem relevant for
machine-to-machine communications, car-to-car/infrastructure and wireless sensor networks where status updates and
control messages determine the typical traffic type.
%
%
\subsection{Channel and Noise Statistics}

The channel and noise taps are modeled as independent circularly symmetric Gaussian random variables
\begin{align}
  \vh &\in \mathbb{C}^{\Nh} \quad, \quad h_\nh  \sim \CN(0,\pd^{\nh}) \label{eq:channel}\\
  \vw &\in \mathbb{C}^{N}\quad,\quad w_n \sim \CN(0,\sigma^2)\label{eq:noise}
\end{align}
where we assume with $p\leq 1$ an exponential decaying average power delay profile $\Expect{|h_\nh|^2}=\pd^{\nh}$ for
the $\nh$th path, see for example \cite{JSP96}.  The average noise power is denoted by $\sig^2=N_0>0$ and is constant
for each tap. 
Due to the independence of the channel taps we can derive for the \emph{average received signal-to-noise ratio}:
%
%
\begin{align}
  \rSNR = \Expect{\left(\frac{ \Norm{\vx*\vh}^2_2}{\Norm{\vn}^2_2}\right)}
  =\frac{\Expect{\Norm{\vx}_2^2} \Expect{\Norm{\vh}_2^2}}{\Expect{\Norm{\vn}^2_2}}
  =\frac{\Expect{\Norm{\vh}_2^2}}{N\cdot N_0}.\label{eq:rSNR}
\end{align}
The average energy of the multipath Rayleigh fading
channel $\vh$ is then given by 
\begin{align}
  \Expect{\Norm{\vh}_2^2}=\sum_{\nh=0}^{\Nh-1} \pd^{\nh}= \frac{1-\pd^{\Nh}}{1-\pd}.
\end{align}
Hence we obtain
\begin{align}
  \rSNR 
  = \frac{1}{N\cdot N_0} \frac{1-\pd^{\Nh}}{1-\pd}.
\end{align}

\section{Transmission Scheme via Modulation On Zeros}
The convolution in \eqref{eq:receivedsignal} can be also represented by a polynomial multiplication. Let
$\vx\in\C^{\Nx+1}$, then its $z$-transform is the polynomial
\begin{align}
  \uX(z)=\sum_{\nx=0}^{\Nx} x_\nx z^\nx\quad,\quad z\in\C,
\end{align}
which has order $\Nx$ if and only if $x_\Nx\not=0$. The received signal \eqref{eq:receivedsignal} is in the $z-$domain given by
a polynomial of order $\Nx+\Nh-1$
\begin{align}
   \uY(z)=\uX(z)\uH(z)+\uW(z),
\end{align}
where $\uX(z), \uH(z)$ and $\uW(z)$ are the polynomials of order $\Nx$, $\Nh-1$ and $\Nx+\Nh-1$ generated by $\vx,\vh$
respectively $\vw$.  Any polynomial $\uX(z)$ of order $\Nx$, can also be represented by its $\Nx$ zeros $\alpl$ and its
leading coefficient $x_\Nx$ as
\begin{align}
  \uX(z)=x_\Nx \prod_{\nx=1}^\Nx (z-\alp_\nx).
\end{align}
If we assume that $\vx$ is normalized, then $x_\Nx$ is fully determined by its $\Nx$ zeros, which leaves us with $\Nx$
degrees of freedom for our signals, given by $K$ \emph{zero-symbols} $\alp_k$.  Let us note, that the notation $\uX(z)$
is commonly used for the $z-$transform. However, since each polynomial of order $\Nx$, with non-vanishing zeros,
corresponds to a unilateral (one-sided) $z-$transform with the same zeros and an additional pole at $z=0$, both ``zero''
representations above are equivalent. In this work we will exclusively use the polynomial notation, since it will be more convenient
for our purpose.

The multiplication by the channel polynomial $\uH(z)$ adds at most $\Nh-1$ zeros $\bet_\nh$, which may be arbitrary
distributed over the complex plane depending by the actual channel coefficients. However, for typical random channel
models, it holds with probability one that the channel and signal polynomials, generated by a finite codebook set
$\Code\subset \C^{\Nx+1}$, do not share a common zero. The \emph{no common zero} property is a necessary condition for
blind deconvolution, see \cite{XLTK95,LXTK96,WJPH17}.  We will later investigate in more detail the distribution of the
zeros and their dependence on the coefficients to derive robustness results against additive noise.

Contrary to time or frequency modulations, where each time-symbol, resp. frequency-symbol, uses the whole complex plane
as its constellation domain, the $K$ zero-symbols have to share their constellation domains.  Hence, we need to
partition the complex plane in $M \Nx$ disjoint (connected) sets $\{\Dset_{\nx}^{(m)}\}_{\nx=1,m=0}^{\Nx,M-1}$ and cluster
them to $\Nx$ sectors (constellation domains) $\Sset_\nx:=\bigcup_{m=0}^{M-1} \Dset_\nx^{(m)}$ for $\nx=1,2,\dots, \Nx$ of
size $M$ each.  For each set $\Dset_{\nx}^{(m)}$ we associate exactly one zero $\alp_\nx^{(m)}$. This will define
$\Nx$ \emph{zero constellation sets} $\Alp_\nx=\{\alp_\nx^{(0)},\dots,\alp_\nx^{(M-1)}\}$ for $k=1,2,\dots,K$ of $M$ zeros
each. If we select from each $ \Alp_\nx$ exactly one zero-symbol $\alp_\nx$, then we can construct $M^\Nx$ different zero
vectors
\begin{align}
   \valp=\begin{pmatrix}\alp_1\\ \vdots\\ \alp_\Nx\end{pmatrix} \in \Zero=\Alp_1\times \dots \times \Alp_\Nx\subset \C^{\Nx}.
\end{align}
The \emph{zero-codebook}  $\Zero$ has cardinality $M^\Nx$ and allows therefore to encode $\Nx\log M$
bits. Hence, the  message stream of an $M$-ary alphabet is partitioned in words $\vm=(m_1,\dots,m_\Nx)^T$ of length
$\Nx$ and each letter $m_\nx$ is assigned to the $\nx$th zero-symbol $\alp_\nx\in\Alp_\nx$, see \figref{fig:mozscheme}. Note,
that the zero constellation sets $\Alp_\nx$ have to be ordered in the zero-codebook, otherwise a unique letter assignment would not be
possible.  The zero vector $\valp$ generates then by the Vitae formula $\Vitae$, see for example \cite{MMR94}, the
coefficients of the corresponding polynomial\footnote{The Vitae formula can also be seen as an explicit formula for the
  inverse $z-$transform  $z^{-\Nx} x_0\Pro(z-\alp_\nx)$. }
\begin{align}
  \vx=\Vitae(\valp)= x_\Nx\begin{pmatrix}
    (-1)^\Nx \Pro_k \alp_\nx \\ \vdots\\ -\sum_\nx \alp_\nx\\ 1 
  \end{pmatrix}\label{eq:vitae},
\end{align}
where $x_\Nx=x_\Nx(\valp)$ is chosen, such that $\vx$ has unit $\ell_2-$norm. These signal constellations therefore
define an $(\Nx+1)-$block codebook $\Code$ of signals (sequences) in the time--domain. To avoid a signal overlap between
blocks we use a guard interval of $\Nh-1$ resulting in a received block length of $N=\Nx+\Nh$. We will call this channel
encoding scheme a \emph{Modulation On Zeros} (MOZ), see \figref{fig:mozscheme} and \figref{fig:BMOCZscheme} for $M=2$.
Let us note, that the digital data, modulated on the zero-symbols, results in perfect interleaved time and frequency
symbols.  Hence, the transmitter exploits the full multipath diversity in time and frequency. This is in contrast to
most modulation schemes, which either interleave the data in time (OFDM) or frequency domain (PPM, PAM).   
\begin{figure}[t]
  \centering
  \def\svgwidth{\textwidth}
  \footnotesize{
  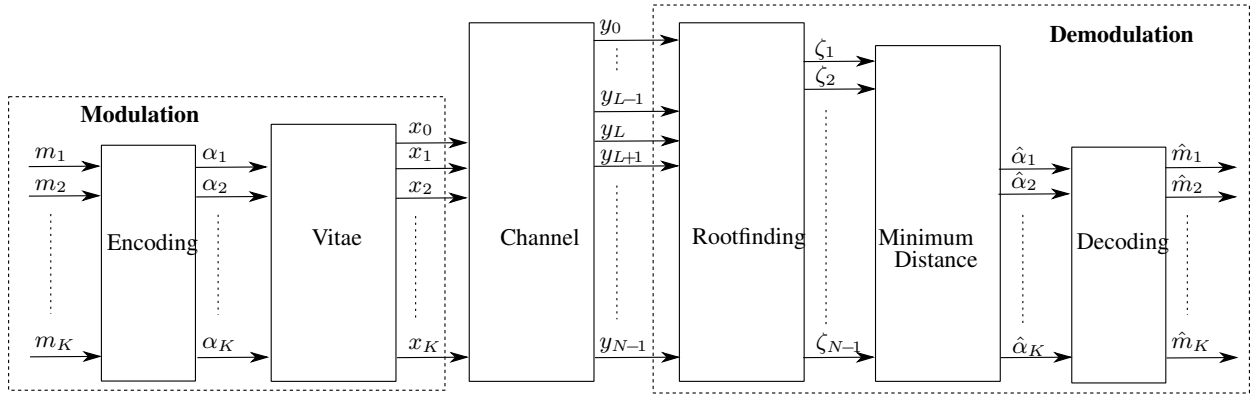
}
\caption{MOZ scheme}\label{fig:mozscheme}
\end{figure}

\subsection{Modulation On Conjugate-reciprocal Zeros}

One such partition structure is given for even $M$ by $M\Nx/2$ conjugate-reciprocal zero pairs with distinct phases.  By
ordering the pairs by their phases in increasing order, we can generate $\Nx$ sectors with $M/2$ possible
conjugate-reciprocal zero pairs 
\begin{align}
  \Alp_\nx= \Big\{
    \{(\alp_\nx^{(0)},\alp_\nx^{(1)})\},\{(\alp_\nx^{(2)},\alp_\nx^{(3)})\},\dots,\{(\alp_\nx^{(M-2)},\alp_\nx^{(M-1)})\}\Big\},
\end{align}
where for all $m=0,2,4,\dots,M-2$ we have $\alp_\nx^{(m+1)}=1/\cc{\alp_\nx^{(m)}}$. We will additionally order $\alp_\nx^{(m)}$ by
increasing phase or radius respectively.  This allows to encode $\log M$ bits per transmitted zero and we call this scheme an $M-$ary
\emph{Modulation On Conjugate-reciprocal Zeros} (MOCZ), pronounced as ``Moxie``.

If we set $M=2$ we can encode exactly  $\Nx$ bits in the signal $\vx$. The  $2\Nx$ zeros $\bigcup \Alp_\nx$ of
the $\Nx$ pairs define an autocorrelation $\va\in\C^{2\Nx+1}$ where we set the leading coefficient $a_{2\Nx}$ such that
$a_K=1$. Then each normalized signal $\vx$ is  generated by
\eqref{eq:vitae} from the zero codeword
\begin{align}
  \valp=\begin{pmatrix}\alp_1\\ \vdots\\ \alp_\Nx\end{pmatrix} \in
  \Zero:=\{\alp_1^{(0)},\alp_1^{(1)}\}\times \dots \times \{\alp_\Nx^{(0)},\alp_\Nx^{(1)}\}\subset \C^{\Nx},
\end{align}
and will have the same autocorrelation $\va=\vx*\cc{\vx^-}$,  see \figref{fig:BMOCZscheme}.
%
Hence, the codebook $\Code$ can be seen as an autocorrelation codebook, where the $\Nx$ bits of information are encoded in
the $2^\Nx$ non-trivial ambiguities\footnote{The trivial scaling ambiguity, is not seen by the zeros and is in the MOZ
scheme not used for information. Hence we loose one degree of freedom of the signal dimension $\Nx+1$. However, this
scheme is therefore independent to global phase of the signals. However, the absolute scaling effects the transmitted
and received power which governors the SNR and hence the robustness against noise. }   of the autocorrelation.
Let us set $\alp_\nx^{(0)}=\radl e^{\im\phil}$ and $\alp_\nx^{(1)}=\radlinv e^{\im \phil}$ for $\phi_1<\phi_2<\dots
<\phi_\Nx$ and $\radlinv>1$ for every $\nx\in[\Nx]$. We can then encode a block $\vm\in\{0,1\}^\Nx$ of $\Nx$ bits $m_\nx$  in
$\vx\in\C^{\Nx+1}$ by assigning the zeros to
\begin{equation}
  \alp_\nx:=\begin{cases} \alp_\nx^{(1)} = \radlinv e^{\im\phil}&, m_\nx=1\\ \alp_\nx^{(0)}=\radl e^{\im\phi_\nx}&, m_\nx=0 \end{cases}
\quad,\quad \nx\in[\Nx],
\end{equation}
see \figref{fig:BMOCZscheme}. We call this scheme a \emph{Binary Modulation On Conjugate-reciprocal Zeros} (BMOCZ).
The blue circles denote the conjugate-reciprocal zero pairs, which define the zero-codebook $\Zero$. The solid blue
circles are the actual transmitted zeros and the red square zeros are the received zeros, given by the disturbed  channel and
data zeros.

\begin{figure}[H]
  \centering
\begin{subfigure}{0.47\textwidth}
  \vspace{0.04\textwidth}
  \def\svgwidth{0.9\textwidth}
  \footnotesize{
  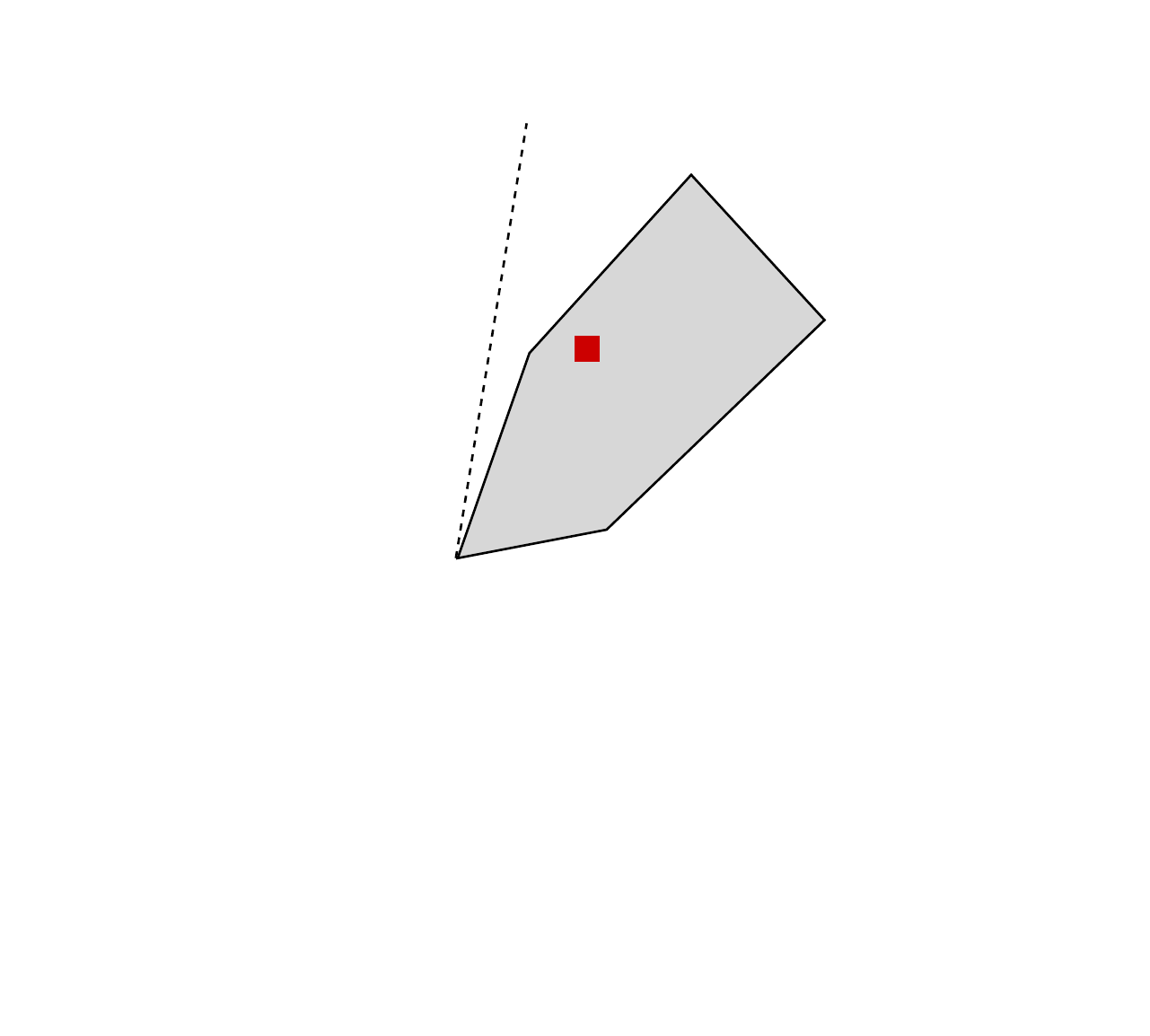
}
  \vspace{0.117\textwidth}
  \caption{Arbitrary BMOCZ scheme. }
  \label{fig:BMOCZscheme}
  \end{subfigure}
\begin{subfigure}{0.5\textwidth}
  \def\svgwidth{\textwidth}
  \footnotesize{
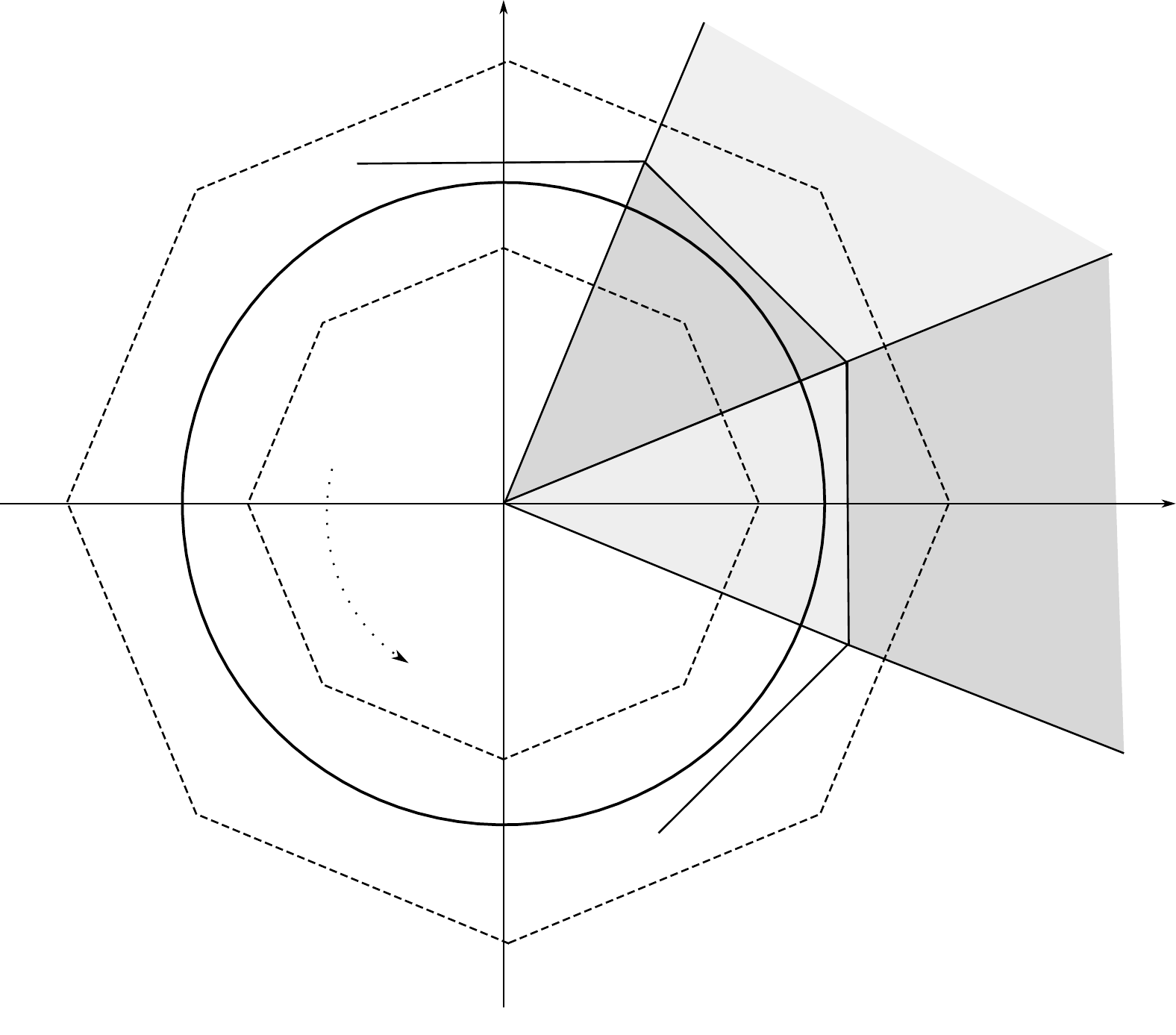
}
\vspace{-0.07\textwidth}
  \caption{Huffman BMOCZ scheme.}
  \label{fig:HuffmanBMOCZscheme}
  \end{subfigure}
  \caption{The zero-codebook $\Alp$ and their decoding sets (Voronoi cells). Red squares denote received zeros.}
  \label{fig:BMOCZ}
\end{figure}

\subsection{Demodulation and Decoding via Root finding and Minimum Distance} 
Let us first explain how one could in principle demodulate the data.
The following exposition is meant mainly for illustration and
analysis. More efficient implementations will be discussed later on.
Thus, at the receiver we will observe by \eqref{eq:receivedsignal} a
disturbed version of the transmitted polynomial
\begin{align}
 \uY(z)= \uX(z)\cdot \uH(z) +\uW(z)=x_\Nx h_{\Nh-1}\Pro_{\nx=1}^\Nx (z-\alp_\nx)\Pro_{\nh=1}^{\Nh-1} (z-\bet_\nh)
 + w_{N-1} \Pro_{n=1}^{N-1} (z-\gam_n)\label{eq:receivedY}
\end{align}
where first new channel zeros $\bet_\nh$ are added to the transmitted zeros $\alp_\nx$ of $\uX$, which then both will be
perturbed by a noise polynomial. We will discuss the stability of such an approach later in \secref{sec:rootstability}.

From the received signal coefficients, the zeros $\zeta_n$ of the received polynomial $\uY(z)$
can be computed using some \emph{root finding} algorithm. After assigning  the received zeros $\zet_n$ in the $K$
sectors $\Sset_k$, one can separate the data zero from its channel zeros by a
\emph{minimum distance} decision
%
\begin{align}
  \hat{m}_\nx = \begin{cases}  1,& \min_{\zeta_n \in \Sset_k}  d(\zeta_n,\alp_\nx^{(1)}) 
    < \min_{\zeta_n \in\Sset_k} d(\zeta_n,\alp_\nx^{(0)})\\ 0,& \text{else}
  \end{cases} \quad,\quad \nx=1,2,\dots,\Nx,\label{eq:mindist}
\end{align}
where $d(\cdot,\cdot)$ defines a certain metric on $\C$.  We will call this a \emph{Root-Finding Minimal Distance}
(RFMD) demodulator, see \figref{fig:BMOCZscheme}, where we used $d_{\nx,n}^{(i)}=d(\zeta_n,\alp_\nx^{(i)})$ for $i=0,1$.
For simplicity, we will in this work only consider the (unweighted) Euclidean distance $d(x,y)=|x-y|$, but other
distances might be more suitable, a point we will discuss in \secref{sec:rootstability}. The de/encoding or quantization
sets $\Dset_\nx^{(i)}$ are the Voronoi cells of the zeros $\alp_\nx^{(i)}$, leading to the best performance of the MD
decoder \eqref{eq:mindist}. If the channel is scalar, no channel zeros are added, and the receiver has only to determine
in which cells the received zeros fall to decode the data. If one cell contains multiple received zeros, the decoder
will chose the smaller distance in \eqref{eq:mindist}, see \figref{fig:BMOCZscheme} where for $\nx=2$ the zero $\zeta_n$
is closer to $\alp_2^{(0)}$ as $\zeta_{n+1}$ is to $\alp_2^{(1)}$.

\emph{Adaption for M-MOZ scheme:} For the M-MOZ scheme the transmitter will transmit one zero $\alp_\nx\in\Alp_\nx$ for each sector
$\Sset_\nx$.  If no channel zeros are present, such as scalar channels, and only one received zero $\zeta_n$ is in the
set $\Dset_\nx^{(m)}$, then the decoder will assume that $\alp_\nx^{(m)}$ was transmitted. If multiple zeros are in one
decoding set (channel zeros might be present), we will decide by minimum distance. The general decoding rule
for the $\nx$th message $m_\nx$ is therefore
\begin{align}
  \hat{m}_\nx = \argmin_{m\in\{1,\dots,M\}} \min_{\zeta_n\in\Dset_{\nx}^{(m)}} |\alp_\nx^{(m)}-\zeta_n|.
\end{align}
If the $\nx$th sector $\Sset_{\nx}$ contains no zero at all, then ${m}_{\nx}$ can not be reliable decoded and 
will be in error. Here, multiple scenarios are possible, either one can chose the closest zero from the next neighbor
sectors, as in \eqref{eq:mindist}, or one can request a retransmission for this message.
See \figref{fig:mozscheme} for the general modulation and demodulation scheme. 
\begin{remark}
  Let us note, that the RFMD decoder can also detect potential bit errors, for example if in one cell multiple zeros
  occur, but no zeros in the next-neighbor sector. 
  The encoding/decoding scheme is fundamentally different to classical coding schemes, since at the receiver we observe
  more zeros as we transmit, due to the channel. This can be seen as ISI in the zero-domain. 
\end{remark}

\section{Huffman BMOCZ}
The proposed BMOCZ scheme can be applied to any autocorrelation sequence 
$\va = \vx * \cc{\vx^-}\in\C^{2\Nx+1}$ generating a polynomial with simple zeros, i.e., all zeros are distinct.
%
%
Among all these autocorrelations, \emph{Huffman autocorrelations} \cite{Huf62}  are the most impulsive ones
given for a peak-to-side-lobe (PSL) $\eta\in (0,1/2)$ as
\begin{align}
  a_k & = \begin{cases} -\eta ,&  k=0,2\Nx\\
        1,& k=\Nx\\
        0,& \text{else}
      \end{cases}.
\end{align}
Hence, the autocorrelation generates the polynomial
\begin{align}
  \uA(z)= -\eta + z^{\Nx} -\eta z^{2\Nx}
  \quad\text{and}\quad
  \uA(e^{i2\pi\omega})=2\eta \cos(2\pi  \Nx \omega) -1. \Peter{} 
  \label{eq:autohuf}
\end{align}
Since $\uA(z)=0$ is a quadratic equation in $z^\Nx$, solving for all
zeros $\alp_\nx=R_\nx e^{\im\phil}$, yields for the magnitude and
phases
\begin{align}
  R_\nx=R^{\pm 1}= \left(\frac{1 \pm \sqrt{1-4\eta^2}}{2\eta}\right)^{1/\Nx}, 
  \quad \phil= 2\pi\frac{\nx-1}{\Nx}\ \text{ for }\ k=1,2,\dots,K. 
\end{align}
This results in $\Nx$  conjugate-reciprocal zero pairs uniformly
placed on two circles with radii $R>1$ and $R^{-1}$
\begin{align}
  \alp_\nx\in\Zero_\nx=\{R e^{2\pi\im \frac{\nx-1}{\Nx}}, R^{-1} e^{2\pi\im \frac{\nx-1}{\Nx}}\}\ \text{ for }\ \nx=1,2\dots,\Nx.
\end{align}
Since, the zeros are the vertices of two regular polygons, centered at the origin, they have the best pairwise distance
from all autocorrelations, see \figref{fig:HuffmanBMOCZscheme}.  Expressing the autocorrelation \eqref{eq:autohuf} in the
$z-$domain by its zeros, gives
\begin{align}
  \uA(z) &=-\eta\Pro_{\nx=1}^\Nx (z-\alp_\nx) 
  (z-\cc{\alp_\nx^{-1}}) 
  = \underbrace{ x_\Nx\Pro_{\nx=1}^\Nx (z-\alp_\nx) }_{\uX(z)}\cdot 
  \underbrace{\cc{x_0} \Pro_{\nx=1}^\Nx (z-\cc{\alp_\nx^{-1}})}_{\uX^*(z)}\label{eq:autohufzdomain},
\end{align}
where $\uX^*(z)=\sum_k \cc{x_{\Nx-k}} z^k$ is the \emph{conjugate-reciprocal polynomial} generated by $\cc{\vx^-}$. Each
$\uX(z)$ respectively $\uX^*(z)$ is then called a \emph{Huffman polynomial} and their coefficients a \emph{Huffman
sequence}. Since the
autocorrelation is constant for each selection $\valp\in\C^\Nx$,  the first and last coefficient of $\vx$ depend on the
chosen zeros, i.e., on the bit vector $\vm=(m_1,\dots,m_K)$:
\begin{align}
  &\cc{x_0} \cdot {x_\Nx} \overset{\eqref{eq:autohufzdomain}}{=}-\eta\text{ and } x_\Nx\cdot (-1)^\Nx\Pro_\nx \alp_\nx =x_0
  \quad \RA\quad |x_\Nx|^2 = \frac{\eta}{(-1)^{\Nx-1}\cc{\Pro_\nx\alp_\nx}} \\
  &\quad\LRA\quad x_\Nx = e^{j\phi_0}  \sqrt{\eta}
  R^{\Nx/2-\Norm{\vm}_1}, \ \ x_0 = e^{j\phi_0} \sqrt{\eta} R^{\Norm{\vm}_1-\Nx/2}
  \quad,\quad\phi_0\in[0,2\pi)
\end{align}
since we have
\begin{align}
  \Pro_{\nx=1}^\Nx \alp_\nx= \Pro_{\nx=1}^\Nx R^{2m_\nx-1}e^{j2\pi \frac{(\nx-1)}{\Nx}}
  =  R^{2\Norm{\vm}_1-\Nx} e^{j2\pi \frac{\sum_{\nx=1}^{\Nx-1} \nx}{\Nx}} = R^{2\Norm{\vm}_1-\Nx} e^{j2\pi  \frac{(\Nx-1)\Nx}{2\Nx}} =
R^{2\Norm{\vm}_1-\Nx}(-1)^{\Nx-1}\notag. 
\end{align}
By rewriting $\eta$ in terms of the radius we get
\begin{align}
  \eta = \frac{1}{R^\Nx+R^{-\Nx}}.
\end{align}
If $\phi_0=0$, the first and last coefficients of $\vx$ are given by 
\begin{align}
  x_\Nx= \sqrt{\frac{R^{\Nx-2\Norm{\vm}_1}}{R^\Nx+R^{-\Nx}}} = \sqrt{
    \frac{R^{-2\Norm{\vm}_1}}{1+R^{-2\Nx}}}\ \text{ and }\ x_0
    =-\sqrt{\frac{R^{2\Norm{\vm}_1}}{1+R^{2\Nx}}}.\label{eq:hufflastfirst}
\end{align}
This suggest, that the first and last coefficients of Huffman sequences are dominant, which might help for a
synchronisation and detection at the receiver. Furthermore, the free choice of the phase reflects the degree of freedom,
we will lose in our modulation scheme. Note, we have $2\Nx+2$ real parameters representing $\Nx$ complex zeros and $1$
complex constant (trivial polynomial). The magnitude of this constant is determined by the PSL $\eta$, the bit vector
$\vm$, and the signal power. But its phase, acting as a global phase of the Huffman sequence, can not be resolved at the
receiver at all, due to the last coefficient of the channel and the additive noise \eqref{eq:receivedY} and therefore
has to be fixed for the scheme.  Hence, the MOCZ scheme uses $2\Nx+1$ of the $2\Nx+2$ degrees of freedom. 
\begin{remark}
  Impulsive-equivalent autocorrelations are usually used to estimate the distance of objects, as used in radar, or to
  estimate the channel state, see for example \cite[Cha.12]{GG05}. By the best knowledge of the authors, the properties
  of Huffman sequences have never been used for a digital data communication. \Peter{Sure?} 
\end{remark}

\section{Maximum Likelihood Receiver for BMOCZ}

We shall derive now a much simpler and efficient demodulation
technique for the BMOCZ scheme by using the fixed autocorrelation
property of the codebook $\Code$, represented by the autocorrelation
matrix $\vA\in\C^{\Nx+1\times \Nx+1}$, which is Hermitian Toeplitz and
generated by $\va\in\C^{2\Nx+1}$.  If $\Nh=\Nx+1$ then the
autocorrelation matrix of $\vx$ is given by the $N\times \Nh$ banded
Toeplitz matrix $\vX$ generated by $\vx$ as
\begin{align}
  \vA=\vX^*\vX= 
\begin{pmatrix} 
    a_\Nx & \dots & a_0 \\
    \vdots &\diagdown & \vdots\\
    a_{2\Nx}& \dots & a_\Nx
  \end{pmatrix}
  \quad,\quad 
  \vX= \begin{pmatrix}
    x_0 & 0 &\dots &0 & 0 \\
    x_1 & x_0 &\dots & 0& 0\\
    \vdots & \vdots &\diagdown & &\vdots\\
    x_\Nx & x_{\Nx-1} &&& 0 \\
    0 & x_{\Nx} &\dots &  \diagdown & \\
    \vdots & & \diagdown & &x_0\\
    \vdots &  & &\diagdown & \vdots\\
    0 & 0 &\dots &0 & x_{\Nx}\end{pmatrix}\in\C^{N\times \Nh}. \label{eq:AKlarge}
\end{align}
Note, that we can write the convolution in \eqref{eq:receivedsignal} with $\vX$ in the vector-matrix notation as $\vx*\vh=\vX\vh$.
If $\Nh<\Nx+1$ then we cut out a $\Nh\times \Nh$ principal submatrix of $\vA$ and for $\Nh\geq \Nx+1$ we extend $\vA$ by adding
zeros to the generating vector $\va$, i.e.
\begin{align}
\vA_\Nh\!=\!
\begin{pmatrix} 
  a_\Nx & \dots & a_{\Nx-\Nh} \\
  \vdots &\diagdown & \vdots\\
  a_{\Nx+\Nh}& \dots & a_\Nx
\end{pmatrix} \text{for } \Nh\!<\!\Nx+1, \quad\vA_\Nh\!=\!
\begin{pmatrix} 
  a_\Nx & \dots & a_0 & & \text{\kern-0.5em\smash{\raisebox{-1.5ex}{\huge 0}}}\\
  \vdots &\diagdown &\vdots & \diagdown   &\\
  a_{2\Nx} & \dots & a_\Nx &\dots & a_0 \\
         & \diagdown & \vdots & \diagdown &\vdots\\
         \huge{\text{0}} &  & a_{2\Nx} & \dots & a_\Nx \\
       \end{pmatrix} \text{ for } \Nh\!\geq\! \Nx+1
\end{align}
%
In any case, the matrix $\vA_\Nh$ will be constant for any fixed Codebook $\Code$. 
%
For multipath channels the \emph{maximum likelihood (sequence)
  detector} is known to be optimal \Peter{(was meinst Du
  damit und woher kommt das, quelle)} and is given by maximizing the
conditional probability for each possible signal (codeword, sequence)
$\vx$ in the codebook \Peter{(marginals over the channel statistics ?)}
\newcommand{\Huff}{\mathscr{H}}
\begin{align}
  \hvx=  \arg\max_{\vx\in\Code} p(\vy|\vx)\label{eq:defml}.
\end{align}
By assumption \eqref{eq:channel} and
\eqref{eq:noise} the channel and noise parameters are independent
zero-mean Gaussian random variables,
hence  the received signal $\vy$ is also a Gaussian random vector with mean zero and covariance matrix
$\vR_{\vy}$, see \cite[(3.17)]{Mad08}.
The conditional probability is therefore given by 
\begin{align}
  p(\vy|\vx)  =  \frac{e^{-\vy^*\vR_{\vy}^{-1}\vy}}{\pi^{N}\det(\vR_{\vy})}\label{eq:gaussry},
\end{align}
see \cite[Lem.3.B.1]{Kai00}. The covariance matrix of $\vy$ is given by 
\begin{align}
  \vR_{\vy}&=\Expect{\vy\vy^*}= \Expect{(\vX\vh+\vn)(\vh^*\vX^*+\vn^*)}\\
  &=\Expect{\vX\vh\vh^*\vX^*} +
  \underbrace{\Expect{\vX\vh\vn^*}}_{=0}+
  \underbrace{\Expect{\vn\vh^*\vX^*}}_{=0}+ \Expect{\vn\vn^*}
   = \vX\Expect{\vh\vh^*}\vX^* +\sig^2\id_N,
\end{align}
since $\vn$ and $\vh$ are independent zero-mean random variables.  The
discrete power delay profile
\renewcommand{\vp}{\ensuremath{\underline{\mathbf p}}}
\begin{align}
  \vp=(\pd^0,\pd^1,\dots,\pd^{\Nh-1})\label{eq:pdp}
\end{align}
generates the channel covariance matrix $\Expect{\vh\vh^*}=\vD_{\vp}$, which is a $\Nh\times \Nh$
diagonal matrix with diagonal $\vp$. This gives for the covariance matrix
\begin{align}
 \vR_{\vy}= \sig^2\id_N + \vX\vD_{\vp}\vX^*. 
\end{align}
We will set $\vXp:=\vX\vD_{\vp}^{1/2}\in\C^{N\times \Nh}$ such that \eqref{eq:defml} separates with \eqref{eq:gaussry} to  
\begin{align}
  \arg\max_{\vx} p(\vy|\vx)&=  \arg\min_{\vx} -\log p(\vy|\vx) \\
                           &= \arg\min_{\vx} \Big(\underbrace{\vy^*(\sig^2\id_N +
    \vXp \vXp^*)^{-1} \vy}_{\geq 0} +   \log(\pi^N\det (\sig^2\id_N + \vXp\vXp^*))\Big)\label{eq:MLdet}
\end{align}
where the log-function is monotone increasing and negative, since $p(\vy|\vx)<1$. By using \emph{Sylvester's determinant
identity}, we get for the second summand in \eqref{eq:MLdet} by using that the autocorrelation, power delay profile
$\vp$ and noise power $\sig$ is constant: 
\begin{align}
  \det(\sig^2\id_N + \vXp\vXp^*)=\det(\sig^2\id_{\Nh} +\vXp^*\vXp)
= 
  \det(\sig^2\id_{\Nh} + \vD_{\vp}^{1/2}
  \vA_\Nh\vD_{\vp}^{1/2}) =\text{const.}
\end{align}
Hence, we can omit this term in \eqref{eq:MLdet}.
By applying the \emph{Woodbury matrix identity}\footnote{Note, that $\id_{\Nh}$ and $\id_N$ are non-singular, but not
$\vXp$.}, see for example \cite[(0.7.4.1)]{HJ13}, we get
\begin{align}
  \vy^*(\sig^{2}\id_N + \vXp \id_{\Nh} \vXp^*)^{-1} \vy 
  &= \vy^*(\sig^{-2}\id_N - \sig^{-2}\vXp(\id_{\Nh} + \sig^{-2}\vXp^*\vXp)^{-1}\vXp^*\sig^{-2})\vy \\
  &=\underbrace{\sig^{-2}\Norm{\vy}_2^2}_{=\text{const.}} -\sig^{-2}\vy^*\vXp(\sig^2\id_{\Nh}
  +\vXp^*\vXp)^{-1}\vXp^*\vy.
\end{align}
Hence, the ML estimator simplifies to
\begin{align}
  \hvx &
  = \arg\max_{\vx} \vy^*\vXp(\sig^2\id_{\Nh} + \vXp^*\vXp)^{-1}\vXp^*\vy.
\end{align}
Inserting the diagonal power delay profile matrix we get
\begin{align}
\hvx &
   = \arg\max_{\vx} \vy^*\vX\vD_{\vp}^{1/2}(\sig^2\id_{\Nh} + \vD_{\vp}^{1/2}\vX^*\vX\vD_{\vp}^{1/2})^{-1}\vD_{\vp}^{1/2}\vX^*\vy\\
  &= \arg\max_{\vx} \vy^*\vX(\underbrace{\sig^2\vD_{\vp}^{-1} + \vA_\Nh}_{=\vB\mgeq 0})^{-1}\vX^*\vy,
  \label{eq:yxbxy}
\end{align}
where $\vA_\Nh\in\C^{\Nh\times \Nh}$ is given by
\eqref{eq:AKlarge}. Since the matrix $\vB\in\C^{\Nh\times\Nh}$ is
constant and reflects the codebook, power delay profile, and noise
power it acts as a weighting for the projections of $\vy$ to the
shifted codewords.  We will call this decoder the \emph{Maximum
  Likelihood} (ML) decoder:
\begin{align}
  \hvx =  
    \arg\max_{\vx} \Norm{\vB^{-1/2}\vX^*\vy}_2^2
  \label{eq:MLdecoder}.
\end{align}
\Peter{essentially this is the {\em matched filter}, isn't
  it ?}
Note, the ML reduces for  $\Nh=1$ to the \emph{correlation receiver}
\begin{align}
  \arg\max_{\vx} |\vx^*\vy|^2 
\end{align}
see for example \cite[Sec.4.2-2]{PS08}.  Since the codebook has cardinality $2^\Nx$ and $\vx\in\C^{\Nx+1}$ the scheme is
non-orthogonal for $\Nx\geq 2$.  If $\Nh<\Nx+1$ and the codebook are the Huffman sequences, then $\vA_\Nh=\id_\Nh$ and
$\vB=\vD_\vb$ becomes a diagonal
matrix with $\vb=\sig^2\vp^{-1}+\eins_\Nh$. Hence, we end up with a Rake receiver, where the weights for the $\nh$th fingers
(correlators) are given by $b_\nh^{-1}=(p^\nh+\sig^2)/\sig^2$, which reflects the sum power of channel gain and signal to
noise ratio of the $\nh$th path.  

\subsection{Direct Zero Testing Decoder for Huffman BMOCZ}\label{sec:dizet}

Huffman sequences not only allow a simple encoding by its zeros,
\Peter{(how is this working..filterbank...)} but also a simple decoding, since the
autocorrelation are by design the most impulsive-like autocorrelations
of any sequence $\vx$.  We set
$\veta=(\underbrace{0,\dots,0}_{\Nx},\eta,\eta,\dots,\eta)^T\in\C^{\Nh}$
for $\Nh\geq \Nx+1$ and get by \eqref{eq:autohuf} the autocorrelation
matrix
\begin{align}
  \vA_\Nh=\vX^*\vX &= \begin{cases} \id_{\Nh}, & \Nh<\Nx+1\\
    \id_{\Nh} - \veta\ve_1^* -\ve_1\veta^*, & \Nh\geq \Nx+1
  \end{cases}
\end{align}
Let us consider the case $\Nh<\Nx+1$, then the matrix $\vB$ becomes 
\begin{align}
  \vB=\vD_{\vb} = \vD_{\vb}\vX^*\vX.
\end{align}
If and only if $\vD_{\vb}=b\id_\Nh$ for some $b\not=0$ we can identify in \eqref{eq:yxbxy} the orthogonal projector on $\vX$ 
\begin{align}
  \vP =b^{-1} \vX(\vX^*\vX)^{-1}\vX^*
\end{align}
and obtain with the left null space $\vV$ of $\vX$ the identity
\begin{align}
  \vP = \id_N - \vV(\vV^*\vV)^{-1}\vV^*.
\end{align}
Let us define the $\Nx\times N$ \emph{generalized Vandermonde matrix} generated by the complex-conjugated zeros
$\alp_1,\dots, \alp_\Nx$ of $\uX(z)$
\begin{align}
\vValp^*= \begin{pmatrix}
  1 & \alp_1 & \alp_1^2 & \dots &\alp_1^{N-1} \\
  1 & \alp_2 & \alp_2^2 & \dots &\alp_2^{N-1} \\
    \vdots &  & &\vdots\\
    1 & \alp_{\Nx} & \alp_\Nx^2 & \dots  &\alp_\Nx^{N-1}\\
  \end{pmatrix}.
\end{align}
Since each zero is distinct, the Vandermonde matrix has full rank $\Nx$.  Then, each complex-conjugated column is in the
left null space of the matrix $\vX$. More precisely we get 
\begin{align}
  \vValp^*\vX=\begin{pmatrix} 
    \uX(\alp_1) & \alp_1 \uX(\alp_1) & \dots & \alp_1^{\Nh-1}\uX(\alp_1)\\
    \uX(\alp_2) & \alp_2 \uX(\alp_2) & \dots & \alp_2^{\Nh-1}\uX(\alp_2)\\
                 \vdots & &\vdots \\ 
    \uX(\alp_\Nx) & \alp_\Nx \uX(\alp_\Nx) & \dots & \alp_\Nx^{\Nh-1}\uX(\alp_\Nx)\\
               \end{pmatrix}=\vzero \quad\LRA \quad \vX^*\vValp=\vzero
\end{align}
In fact, the dimension of the left null space of $\vX$ (null space of $\vX^*$) is exactly $\Nx$ for each $\vX$ generated
by $\vx\in\Code$, since it holds $N=\Nh+\Nx=\rank(\vX^*)+\nullity(\vX^*)$, where $\rank(\vX)=\rank(\vX^*)=\Nh$ and the shifts of
$\vx$ are all linear independent for any $\vx\not=\zero$. Hence, we get
\begin{align}
 \vy^*\vX(\vX^*\vX)^{-1}\vX^*\vy= \vy^*(\id_N-\vValp(\vValp^*\vValp)^{-1}\vValp^*)\vy
  =\underbrace{\Norm{\vy}^2}_{=\text{const}>0}- \Norm{(\vValp^*\vValp)^{-1/2}\vValp^*\vy}^2\label{eq:yxxy}
\end{align}
which yields  with the mixing matrix $\vM_{\valp}=(\vV_{\valp}^*\vV_{\valp})^{-1/2}\in\C^{\Nx\times \Nx}$ to
\begin{align}
  \arg\max_{\valp} p(\vy|\vx(\valp)) = \arg\min_{\valp} \Norm{(\vV_{\valp}^*\vV_{\valp})^{-1/2} \vV_{\valp}^*\vy}^2
  =
  \arg\min_{\valp}\Norm{\vM_{\valp}^{-1/2}\vV_{\valp}^*\vy}^2. 
\end{align}
For Huffman zeros we have $\alp_\nx=R_\nx e^{i2\pi (\nx-1)/\Nx}$  with $R_\nx\in\{R,R^{-1}\}$ and we get
\begin{align}
  \vV_{\valp}^*\vV_{{\valp}}
=
\begin{pmatrix}
 1 & \alp_1 & \dots & \alp_1^{N-1}\\
 \vdots & &&\\
 1 & \alp_\Nx & \dots & \alp_\Nx^{N-1}
\end{pmatrix}
  \begin{pmatrix}
  1 & \dots & 1 \\
  \cc{\alp_1} & \dots &\cc{\alp_\Nx}\\
  \vdots &  &\vdots\\
  \cc{\alp_1^{N-1}} & \dots &\cc{\alp_\Nx^{N-1}}
\end{pmatrix}
=\begin{pmatrix}
  c_{1,1} & c_{1,2} & \dots & c_{1,\Nx}\\
  c_{2,1}  & c_{2,2} &\dots & c_{2,\Nx} \\
  \vdots && \ddots & \vdots\\
  c_{\Nx,1} & c_{\Nx,2} &\dots  & c_{\Nx,\Nx}
\end{pmatrix}.\label{eq:C}
\end{align}
With the geometric series we get
\begin{align}
  c_{\nx,m} &= \sum_{n=0}^{N-1} (\alp_\nx \cc{\alp}_m)^n 
  = \sum_{n=0}^{N-1}(R_\nx R_m)^n e^{\im 2\pi n\frac{(\nx-m)}{\Nx}}\\
  c_{\nx,\nx} & =c_\nx=\sum_{n=0}^{N-1} |\alp_\nx|^{2n} = \frac{1-R_\nx^{2N}}{1-R_\nx^2}.
\end{align}
%
In expectation, for uniform bit sequences, we get
$\Expect{R_\nx R_m}\simeq 1$ for $\nx\not=m$ and hence for $N=\nh\Nx$
the off diagonals are roughly vanishing, since
$\sum_{n=0}^{\Nx-1} e^{\im 2\pi n(k-m)/\Nx}=0 $. Hence, we approximate
\eqref{eq:C} as a diagonal matrix, which leads to
\begin{align}
  \vM_{\valp}^{-1/2}\simeq \diag\left( \sqrt{\frac{1-|\alp_1|^{2}}{1-|\alp_1|^{2N}}},\dots,
  \sqrt{\frac{1-|\alp_\Nx|^{2}}{1-|\alp_\Nx|^{2N}}}\right)\label{eq:gwfactor}.
\end{align}
By observing 
\begin{align}
  \vV_{\valp}^*\vy=\begin{pmatrix}\uY(\alp_1) & \dots & \uY(\alp_\Nx)\end{pmatrix}^T,
\end{align}
the exhaustive search of the ML simplifies to independent decisions for each zero symbol
\begin{align}
  \hat{\alp_\nx}:=\argmin_{\alp_\nx\in \{ R,R^{-1}\} e^{\im 2\pi\frac{\nx-1}{\Nx}}} 
   \Big|\sqrt{\frac{1-|\alp_\nx|^{2}}{1-|\alp_\nx|^{2(N-1)}}} \uY(\alp_\nx)\Big|.
\end{align}
This gives the \emph{Direct Zero Testing} (DiZeT) decoding rule for $k\in\{1,\dots,K\}$  
\begin{align}
  b_\nx = \begin{cases} 1, & |\uY(R e^{\im 2\pi \frac{\nx-1}{\Nx}})| <
    R^{N-2}| \uY(R^{-1} e^{\im 2\pi\frac{\nx-1}{\Nx}})|\\
     0, &\text{else}
   \end{cases}\label{eq:dizetdecoder}
\end{align}
since it holds for the  \emph{geometrical weights} (GW)
\begin{align}
  \sqrt{\frac{1-R^{2(N-1)}}{1-R^{-2(N-1)}}\frac{1-R^{-2}}{1-R^2}}= \sqrt{(-R^{2N-2})\cdot(-R^{-2})}=R^{N-2}.
\end{align}
If $\Nh\geq \Nx+1$ we will approximate $\vA_\Nh\simeq \id_\Nh$. Then the same approximation yield to the same DiZeT decoder.
%

\subsection{FFT-Implementation of Huffman BMOCZ-decoding}
In fact, the DiZeT decoder for Huffman sequences allows also a simple hardware implementation at the receiver.
If we scale the received samples $y_n$ with the positive radius $R^n$ and resp. $R^{-n}$, i.e., 
\begin{align}
  \vD_R\vy:=\begin{pmatrix}
    1 & 0  & \dots &0\\
    0 & R & \dots &0\\
    \vdots & & \ddots& \vdots\\
    0 & 0&\dots & R^{N-1}\end{pmatrix} \vy
\end{align}
and apply the $N-$point DFT matrix if $\Nh=(t-1)\Nx$ for $t\in\N$, yielding to $N=tK$, we get the samples of the $z-$transform
\begin{align}
  \Fmatrixa \vD_R\vy = \begin{pmatrix}\sum_{n=0}^{N-1} y_n R^n e^{i2\pi 0\cdot n/N}\\ \vdots\\ \sum_n y_n R^n e^{i2\pi
    (N-1)\cdot k/N}\end{pmatrix}=:\uY(\valp_t^{(1)})\quad,\quad \Fmatrixa \vD_{R^{-1}} \vy = \uY(\valp_t^{(0)}).
\end{align}
Then the decoding rule  \eqref{eq:dizetdecoder} becomes %
\begin{align}
  b_{\nx} = \begin{cases} 1&, |(\Fmatrix\vD_R\vy)_{tk}|< R^{N-2}|(\Fmatrixa\vD_{R^{-1}} \vy)_{tk}|\\
    0 &, \text{ else}
  \end{cases}.
\end{align}
Hence, the decoder can be fully implemented by a simple $N-$point DFT from the delayed amplified received signal, by
using for example FPGA or even analog front-ends.



%% file: MOZblockscheme.pdf_tex
\begingroup%
  \makeatletter%
  \providecommand\color[2][]{%
    \errmessage{(Inkscape) Color is used for the text in Inkscape, but the package 'color.sty' is not loaded}%
    \renewcommand\color[2][]{}%
  }%
  \providecommand\transparent[1]{%
    \errmessage{(Inkscape) Transparency is used (non-zero) for the text in Inkscape, but the package 'transparent.sty' is not loaded}%
    \renewcommand\transparent[1]{}%
  }%
  \providecommand\rotatebox[2]{#2}%
  \ifx\svgwidth\undefined%
    \setlength{\unitlength}{844.66173716bp}%
    \ifx\svgscale\undefined%
      \relax%
    \else%
      \setlength{\unitlength}{\unitlength * \real{\svgscale}}%
    \fi%
  \else%
    \setlength{\unitlength}{\svgwidth}%
  \fi%
  \global\let\svgwidth\undefined%
  \global\let\svgscale\undefined%
  \makeatother%
  \begin{picture}(1,0.31354007)%
    \put(0,0){\includegraphics[width=\unitlength,page=1]{MOZblockscheme.pdf}}%
    \put(0.15630985,0.16378347){\color[rgb]{0,0,0}\makebox(0,0)[lb]{\smash{\alptwo}}}%
    \put(0.15630985,0.18883573){\color[rgb]{0,0,0}\makebox(0,0)[lb]{\smash{\alpone}}}%
    \put(0.15630985,0.03678623){\color[rgb]{0,0,0}\makebox(0,0)[lb]{\smash{\alpL}}}%
    \put(0.24462524,0.1196856){\color[rgb]{0,0,0}\makebox(0,0)[lb]{\smash{Vitae}}}%
    \put(0,0){\includegraphics[width=\unitlength,page=2]{MOZblockscheme.pdf}}%
    \put(0.32243799,0.1627195){\color[rgb]{0,0,0}\makebox(0,0)[lb]{\smash{\xtwo}}}%
    \put(0.32243799,0.03572226){\color[rgb]{0,0,0}\makebox(0,0)[lb]{\smash{\xL}}}%
    \put(0,0){\includegraphics[width=\unitlength,page=3]{MOZblockscheme.pdf}}%
    \put(0.32243799,0.18976392){\color[rgb]{0,0,0}\makebox(0,0)[lb]{\smash{\xone}}}%
    \put(0,0){\includegraphics[width=\unitlength,page=4]{MOZblockscheme.pdf}}%
    \put(0.32243799,0.21036381){\color[rgb]{0,0,0}\makebox(0,0)[lb]{\smash{\xzero}}}%
    \put(0,0){\includegraphics[width=\unitlength,page=5]{MOZblockscheme.pdf}}%
    \put(0.39646118,0.1196856){\color[rgb]{0,0,0}\makebox(0,0)[lb]{\smash{Channel}}}%
    \put(0,0){\includegraphics[width=\unitlength,page=6]{MOZblockscheme.pdf}}%
    \put(0.47693773,0.20865486){\color[rgb]{0,0,0}\makebox(0,0)[lb]{\smash{\yKone}}}%
    \put(0.47630352,0.03662692){\color[rgb]{0,0,0}\makebox(0,0)[lb]{\smash{\yN}}}%
    \put(0,0){\includegraphics[width=\unitlength,page=7]{MOZblockscheme.pdf}}%
    \put(0.47693773,0.23569928){\color[rgb]{0,0,0}\makebox(0,0)[lb]{\smash{\yK}}}%
    \put(0,0){\includegraphics[width=\unitlength,page=8]{MOZblockscheme.pdf}}%
    \put(0.47630352,0.29308483){\color[rgb]{0,0,0}\makebox(0,0)[lb]{\smash{\yzero}}}%
    \put(0,0){\includegraphics[width=\unitlength,page=9]{MOZblockscheme.pdf}}%
    \put(0.47666563,0.18845528){\color[rgb]{0,0,0}\makebox(0,0)[lb]{\smash{\yKtwo}}}%
    \put(0,0){\includegraphics[width=\unitlength,page=10]{MOZblockscheme.pdf}}%
    \put(0.55096089,0.12063228){\color[rgb]{0,0,0}\makebox(0,0)[lb]{\smash{Rootfinding}}}%
    \put(0,0){\includegraphics[width=\unitlength,page=11]{MOZblockscheme.pdf}}%
    \put(0.80786781,0.18613956){\color[rgb]{0,0,0}\makebox(0,0)[lb]{\smash{\halpone}}}%
    \put(0.80850207,0.03250439){\color[rgb]{0,0,0}\makebox(0,0)[lb]{\smash{\halpL}}}%
    \put(0,0){\includegraphics[width=\unitlength,page=12]{MOZblockscheme.pdf}}%
    \put(0.6505775,0.03559803){\color[rgb]{0,0,0}\makebox(0,0)[lb]{\smash{\gamN}}}%
    \put(0.64930904,0.27310643){\color[rgb]{0,0,0}\makebox(0,0)[lb]{\smash{\gamone}}}%
    \put(0,0){\includegraphics[width=\unitlength,page=13]{MOZblockscheme.pdf}}%
    \put(0.80759574,0.16593996){\color[rgb]{0,0,0}\makebox(0,0)[lb]{\smash{\halptwo}}}%
    \put(0,0){\includegraphics[width=\unitlength,page=14]{MOZblockscheme.pdf}}%
    \put(0.74247157,0.12006953){\color[rgb]{0,0,0}\makebox(0,0)[b]{\smash{Minimum }}}%
    \put(2.37670457,-0.76194972){\color[rgb]{0,0,0}\makebox(0,0)[lt]{\begin{minipage}{0.18937167\unitlength}\raggedright \end{minipage}}}%
    \put(0.64994329,0.25074571){\color[rgb]{0,0,0}\makebox(0,0)[lb]{\smash{\gamtwo}}}%
    \put(0,0){\includegraphics[width=\unitlength,page=15]{MOZblockscheme.pdf}}%
    \put(0.07972407,0.1155584){\color[rgb]{0,0,0}\makebox(0,0)[lb]{\smash{Encoding}}}%
    \put(0,0){\includegraphics[width=\unitlength,page=16]{MOZblockscheme.pdf}}%
    \put(0.02155678,0.16442418){\color[rgb]{0,0,0}\makebox(0,0)[lb]{\smash{\mtwo}}}%
    \put(0.02155678,0.18947644){\color[rgb]{0,0,0}\makebox(0,0)[lb]{\smash{\mone}}}%
    \put(0.02155678,0.03742692){\color[rgb]{0,0,0}\makebox(0,0)[lb]{\smash{\mL}}}%
    \put(0,0){\includegraphics[width=\unitlength,page=17]{MOZblockscheme.pdf}}%
    \put(0.86057184,0.11655765){\color[rgb]{0,0,0}\makebox(0,0)[lb]{\smash{Decoding}}}%
    \put(0,0){\includegraphics[width=\unitlength,page=18]{MOZblockscheme.pdf}}%
    \put(0.93673549,0.16352075){\color[rgb]{0,0,0}\makebox(0,0)[lb]{\smash{\hmtwo}}}%
    \put(0.93673549,0.18857301){\color[rgb]{0,0,0}\makebox(0,0)[lb]{\smash{\hmone}}}%
    \put(0.93673549,0.03652346){\color[rgb]{0,0,0}\makebox(0,0)[lb]{\smash{\hmL}}}%
    \put(0,0){\includegraphics[width=\unitlength,page=19]{MOZblockscheme.pdf}}%
    \put(0.92158269,0.19856032){\color[rgb]{0,0,0}\makebox(0,0)[lb]{\smash{}}}%
    \put(0.71370575,0.10185889){\color[rgb]{0,0,0}\makebox(0,0)[lb]{\smash{Distance}}}%
    \put(0,0){\includegraphics[width=\unitlength,page=20]{MOZblockscheme.pdf}}%
    \put(0.83877673,0.28299651){\color[rgb]{0,0,0}\makebox(0,0)[lb]{\smash{\textbf{Demodulation}}}}%
    \put(0,0){\includegraphics[width=\unitlength,page=21]{MOZblockscheme.pdf}}%
    \put(0.05821003,0.21896475){\color[rgb]{0,0,0}\makebox(0,0)[lb]{\smash{\textbf{Modulation}}}}%
  \end{picture}%
\endgroup%

%% file: zeroAlpBMOCZ_codingsets2.pdf_tex
\begingroup%
  \makeatletter%
  \providecommand\color[2][]{%
    \errmessage{(Inkscape) Color is used for the text in Inkscape, but the package 'color.sty' is not loaded}%
    \renewcommand\color[2][]{}%
  }%
  \providecommand\transparent[1]{%
    \errmessage{(Inkscape) Transparency is used (non-zero) for the text in Inkscape, but the package 'transparent.sty' is not loaded}%
    \renewcommand\transparent[1]{}%
  }%
  \providecommand\rotatebox[2]{#2}%
  \ifx\svgwidth\undefined%
    \setlength{\unitlength}{368.6933082bp}%
    \ifx\svgscale\undefined%
      \relax%
    \else%
      \setlength{\unitlength}{\unitlength * \real{\svgscale}}%
    \fi%
  \else%
    \setlength{\unitlength}{\svgwidth}%
  \fi%
  \global\let\svgwidth\undefined%
  \global\let\svgscale\undefined%
  \makeatother%
  \begin{picture}(1,0.90099453)%
    \put(0,0){\includegraphics[width=\unitlength,page=1]{zeroAlpBMOCZ_codingsets2.pdf}}%
    \put(0.56283996,0.56313332){\color[rgb]{0,0,0}\makebox(0,0)[lb]{\smash{}}}%
    \put(0.48350101,0.46318541){\color[rgb]{0,0,0}\makebox(0,0)[lb]{\smash{\Dsettwozero}}}%
    \put(0,0){\includegraphics[width=\unitlength,page=2]{zeroAlpBMOCZ_codingsets2.pdf}}%
    \put(0.85682722,0.58995451){\color[rgb]{0,0,0}\makebox(0,0)[lb]{\smash{\Dsetoneone}}}%
    \put(0,0){\includegraphics[width=\unitlength,page=3]{zeroAlpBMOCZ_codingsets2.pdf}}%
    \put(0.19106042,0.57820689){\color[rgb]{0,0,0}\makebox(0,0)[lb]{\smash{Bit 0}}}%
    \put(0.09130589,0.69946573){\color[rgb]{0,0,0}\makebox(0,0)[lb]{\smash{Bit 1}}}%
    \put(0,0){\includegraphics[width=\unitlength,page=4]{zeroAlpBMOCZ_codingsets2.pdf}}%
    \put(0.61885669,0.81640075){\color[rgb]{0,0,0}\makebox(0,0)[lb]{\smash{\Dsettwoone}}}%
    \put(0.60498118,0.42817653){\color[rgb]{0,0,0}\makebox(0,0)[lb]{\smash{\Dsetonezero}}}%
  \end{picture}%
\endgroup%

%% file: zeroHuffmanBMOCZ_codingsets2.pdf_tex
\begingroup%
  \makeatletter%
  \providecommand\color[2][]{%
    \errmessage{(Inkscape) Color is used for the text in Inkscape, but the package 'color.sty' is not loaded}%
    \renewcommand\color[2][]{}%
  }%
  \providecommand\transparent[1]{%
    \errmessage{(Inkscape) Transparency is used (non-zero) for the text in Inkscape, but the package 'transparent.sty' is not loaded}%
    \renewcommand\transparent[1]{}%
  }%
  \providecommand\rotatebox[2]{#2}%
  \ifx\svgwidth\undefined%
    \setlength{\unitlength}{453.80653604bp}%
    \ifx\svgscale\undefined%
      \relax%
    \else%
      \setlength{\unitlength}{\unitlength * \real{\svgscale}}%
    \fi%
  \else%
    \setlength{\unitlength}{\svgwidth}%
  \fi%
  \global\let\svgwidth\undefined%
  \global\let\svgscale\undefined%
  \makeatother%
  \begin{picture}(1,0.87992335)%
    \put(0,0){\includegraphics[width=\unitlength,page=1]{zeroHuffmanBMOCZ_codingsets2.pdf}}%
    \put(0.27314765,0.56180282){\color[rgb]{0,0,0}\makebox(0,0)[lb]{\smash{Bit 0}}}%
    \put(0.18300352,0.66876808){\color[rgb]{0,0,0}\makebox(0,0)[lb]{\smash{Bit 1}}}%
    \put(0,0){\includegraphics[width=\unitlength,page=2]{zeroHuffmanBMOCZ_codingsets2.pdf}}%
    \put(0.51391168,0.51698118){\color[rgb]{0,0,0}\makebox(0,0)[lb]{\smash{\Dsettwozero}}}%
    \put(0.80684272,0.55303202){\color[rgb]{0,0,0}\makebox(0,0)[lb]{\smash{\Dsetoneone}}}%
    \put(0.60545992,0.78151761){\color[rgb]{0,0,0}\makebox(0,0)[lb]{\smash{\Dsettwoone}}}%
    \put(0.54992661,0.40642157){\color[rgb]{0,0,0}\makebox(0,0)[lb]{\smash{\Dsetonezero}}}%
    \put(0,0){\includegraphics[width=\unitlength,page=3]{zeroHuffmanBMOCZ_codingsets2.pdf}}%
  \end{picture}%
\endgroup%

%% file: generalzerocodebooks.tex
\section{SDP decoder for BMOCZ via Channel Autocorrelation}

As already mention above, noncoherent communication of information
bearing signals having very short length, in the order of the maximum
delay spread of the multipath propagation, is indeed related to the
blind deconvolution problem. This bilinear inverse problem itself
suffers from a rich set of nontrivial ambiguities and impossible to
solve without further constraints.  Therefore, one of the
challenges in communications and the motivation for our approach, is to
develop simple, fast, and efficient methods by restricting the class of
data signals, in this case to finite codebooks. In this section we
give a brief overview on a convex method for solving this problem for
impulsive data signals, as in the case of Huffman sequences.  In
\cite{JH16} one of the authors introduced a semi-definite program to
deconvolve up to global phase almost all signals $\vy=\vx*\vh\in\C^N$
from the knowledge of the autocorrelations $\va_{x}$ and $\va_{h}$.
This program is successful if the polynomials corresponding to the
input signals $\vx$ and $\vh$ do not share a common zero.  Later, in
\cite{WJPH17}, a stable deconvolution via the SDP over the reals
has been proven. In a nutshell, using the idea of lifting one can express the
bilinear problem as a linear estimation problem, i.e., to recover a
positive semi-definite matrix $0 \mleq\vZ \in\C^{\tN\times \tN}$ where
$\tN=L+K+1$ (details see \cite{WJPH17}). In the case of circular
convolutions and with signals in random and incoherent subspaces this
has been investigated first in \cite{ARR12}.  Let
$\Alin:\C^{\tN\times \tN} \to \C^{4\tN-4}$ the linear map representing
here the (non-circular) convolution.
As discussed in \cite{WJPH17} a stable deconvolution can be performed
by minimizing the least-square error:
\begin{align}
  \hat{\vZ}=\arg\min_{\vZ\mgeq 0} \Norm{\vb- \Alin(\vZ)}_2^2 \quad\text{ where }
  \quad\vb=(\va_x,\va_{h}, \vy, \cc{\vy^-})^T
  \label{eq:sdpdenoising}
\end{align}
Performing a rank-one projection of the minimizer $\hat{\vZ}$ using
the singular value decomposition yields then the rank-one matrix
$\hvz\hvz^*$ which reveals  the vector
$\hvz=e^{\im\phi} [\hvx, \hvh]\in\C^{\tN}$ up to a global phase $\phi$.  Thus, this method can be
indeed used for blind deconvolution in a communication context when
the $\va_x$ and $\va_h$ are known. In the following we discuss how
this is possible and why impulsive-like data signals, such as Huffman sequences,
are helpful.

\subsection{Estimation of the Channel Autocorrelation, Noiseless case}
Let us start, for ease of exposition, with the noiseless case, i.e.,
$\vy=\vx*\vh$ where $\vx\in\C^{\Nx+1}$ and $\vh\in\C^{\Nh}$.
From the Wiener-Khintchine relation we get 
\begin{align}
  \va_y=\vy*\cc{\vy^-}= (\vx*\vh)*(\cc{\vx^{-}}*\cc{\vh^{-}})=\va_x*\va_h.
\end{align}
Assume that $\va_x$ is already known and the receiver computes
$\va_y=\vy*\cc{\vy^-}$ from the received signal $\vy$. If the relation
above can be solved for $\va_h$, given $\va_x$ and $\va_y$, we can
indeed use the methodology and the convex program
\eqref{eq:sdpdenoising} for estimating $(\vx,\vh)$. To this end, we
consider this in the Fourier-domain by zero-padding the sequences
$\va_x$ and $\va_h$ to dimension $M=2N-1$ giving the vectors $\vta_x$
and $\vta_h$.  Thus, if $\Fmatrix \vta_x$ has no zeros, we get: 
\begin{align}
  \Fmatrix \va_y \bullet (\sqrt{M}\Fmatrix \vta_x)^{-1}=
  \sqrt{M} \Fmatrix_{M} \vta_h \pwprod \Fmatrix_M \vta_{x} \bullet(\sqrt{M}\Fmatrix_M \vta_x)^{-1}=\Fmatrix \vta_h,
  \label{eq:Fy}
\end{align}
%
and the autocorrelation of the channel can be obtained by:
\begin{align}
  \vta_h= \Fmatrixa (
  \Fmatrix \vta_y \bullet (\sqrt{M}\Fmatrix \vta_x)^{-1}) 
\end{align}
as long as $(\Fmatrix \vta_x)_k \not=0$, which holds by design of the
Huffman sequences.  Removing from $\vta_h$ the last $M-(2\Nh-1)$ zeros
reveals finally the channel autocorrelation $\va_h$.  


\subsection{Estimation of the Channel Autocorrelation Estimation,
  Noisy case}
When computing $\va_y$ in the presence of noise, $\vy=\vx*\vh + \vw$,
we encounter additional cross-correlations and the estimate is
affected by coloured noise:
\begin{align}
  \vw_c= \vn*\cc{\vx^{-}}*\cc{\vh^-} + \cc{\vn^-}*\vx*\vh + \va_w\label{eq:colorednoise}.
\end{align}
where $\va_w=\vw*\cc{\vw^-}$. Obviously, this stage can be improved
by, e.g., LMMSE estimation (Wiener filter).
\Peter{Ok - ich glaube, das geht besser - zumindest fuer
  Comm/SigPro Journal muss man hier was machen}
Nevertheless, let us compute a scaling estimate for the method
above. Repeating the steps above gives:
\begin{align}
  \tilde{\hat{\va}}_h&= \Fmatrixa (\Fmatrix \va_y/(\sqrt{M}\Fmatrix\vta_x))
  =\Fmatrixa \left( \Big[\sqrt{M}\Fmatrix \vta_x \bullet \Fmatrix
   \vta_h+ \Fmatrix \vw_c \Big] /(\sqrt{M} \Fmatrix \vta_x)\right)\\
   & = 
   \vta_h  +\underbrace{ \Fmatrixa (\Fmatrix \vw_c/(\sqrt{M}\Fmatrix \va_x))}_{=\vw'} =\vta_h + \vw'.\label{eq:zfah}
\end{align}
A straightforward bound for the estimation error for Huffman sequences
is:
\begin{align}
  \Norm{\vw'}_2^2
  &= \Norm{\Fmatrix \vw_c/(\sqrt{M}\Fmatrix \vta_x)}_2^2
    \leq \frac{1}{M}\Norm{|\Fmatrix\vw_c|^2}_1 \Norm{|\hat{\vta}_x|^{-2}}_{\infty}\\
  &=\Norm{\vw_c}_2^2 \cdot \frac{1}{M\min_k |(\Fmatrix \vta_x)_k|^2 }
    \overset{\eqref{eq:autohuf}}{=} \Norm{\vw_c}_2^2 \cdot \frac{1}{\min_k |2\eta \cos(2\pi
  \Nx k/M) -1|^2} \leq \frac{\Norm{\vw_c}_2^2}{(1- 2\eta)^2}\notag.
\end{align}
If $\eta <1/3$, see optimal radius \eqref{eq:optimalradius}, we get
$\Norm{\vw'}_2^2 \leq 9 \Norm{\vw_c}_2^2$.
The expectation of the colored noise power in \eqref{eq:colorednoise} can be upper bounded by
\begin{align}
  \Expect{\Norm{\vw_c}_2^2} \leq 2 N\cdot N_0\cdot \Nh + N\cdot N_0^2.
\end{align}
By using $\Norm{\vx}_2^2=\Expect{|h_\nh|^2}=1$ and $ \Expect{\Norm{\vn}_2^2}=N\cdot N_0$ this leaves us with an upper MSE
of 
\begin{align}
  \Norm{\va_h-\hat{\va}_h}_2^2\leq 18 N\cdot N_0(N_0+\Nh).
\end{align}
Hence, for large noise powers this leads to a bad estimate $\hat{\va}_h$ and might therefore  result in a poor performance
of the SDP.

\ifall 
\subsection{Matlab implementation}

Note, the {\tt fft} matlab function is  $\tFmatrix:= \sqrt{M}\Fmatrix$.
\begin{itemize}
  \item  $\tt\tx=$ \tt poly(\alp);  \quad Note: $\tt x(z)=x_1 z^\Nx + x_2 z^{\Nx-1} + \ldots + x_\Nx z + x_{\Nx+1}$
  \item  $\tt x= \tx/$norm$(\tt \tx)$;
  \item \tt a\_x$=$conv($\tt x,trc(x))$;
  \item $\tFmatrix (\vy*\cc{\vy^-}) = \frac{1}{M} \Big(\tFmatrix \begin{pmatrix}\va_x \\ \zero\end{pmatrix} \pwprod
    \tFmatrix\begin{pmatrix} \va_h \\ \zero\end{pmatrix}\Big)$
\end{itemize}

\fi 


%% file: robustness.tex
\section{Continuity and Robustness of Zeros Against Small Perturbations}
\label{sec:rootstability}
Although, the SDP gives insight in the robustness of Huffman sequences, it relies on the knowledge of the channel
autocorrelation. Moreover, as we found in \secref{sec:dizet}, the performance of the DiZeT decoder depends on the
distribution of the zero-symbols. Hence, a robustness analysis for a zero-based modulation, boils down to a robustness
analysis of polynomial zeros.   

Wilkinson investigated at first in \cite{Wil84} the stability of polynomial roots under perturbation of the polynomial coefficients. 
One extreme case of instability is known today as the Wilkinson polynomial
\begin{align}
  \ux(z)=(z-1)(z-2)(z-3)\cdots(z-20)\label{eq:wilkinson}
\end{align}
given by $20$ real-valued zeros equidistant placed on the positive real line.
If only the leading coefficient is disturbed by machine precession 
\begin{align}
  y_{20}=x_{20}+10^{-23},
\end{align}
then the three largest zeros of the perturbed polynomial $\uy(z)=\sum_n y_n z^n$ are completely off, showing that  the
zeros are not stable against distortion on its coefficients.

This can be generalized to arbitrary polynomials and the question is, if we consider the Eulcidean norm, how much the
zeros will be disturbed if we perturb the coefficients with some $\vw\in\C^N$ having $\Norm{\vw}_2^2\leq \eps$.
The answer was given in \cite{FH07} in terms of \emph{root neighborhoods} or \emph{pseudozero sets}
\begin{align}
  \ZeroSet(\eps,\uX) =\set{z\in\C}{\frac{|\sum_n x_n z^n  |^2}{{\sum_n |z|^{2n}}}\leq \eps}
\end{align}
where each disturbed polynomial $\uy(z)=\sum_n (x_n + w_n)^n$ for some $\Norm{\vw}_2\leq \eps$ has all its zeros $\zeta$
in $\ZeroSet(\eps,\uX)$. However, this characterization of the root neighborhoods, does not explain at which noise level 
$\eps$ the single root neighborhoods $\ZeroSet_n(\eps,\uX)$ of the roots $\zeta_n$ will start to overlap.  The
intuition suggest, that with increasing noise power, the single root neighborhoods should monotone grow and eventually start to
overlap, at which a unique zero separation becomes impossible, see \figref{fig:receivedzerosperturbation}.  We plotted
here a fixed Huffman polynomial with $K=6$ zeros (black squares) and $3$ channel zeros (red squares)
generated by Gaussian random vectors (Kac polynomial). The additional channel zeros have only little impact of the root
neighborhoods of the Huffman zeros.  However, they will have an heavy impact on the zero separation (decoding), if they get close
to the zero-codebook. Since the distribution of the Chanel zeros is random, we will only consider the perturbation
analysis of a given polynomial $\uX(z)$. 

\begin{figure}[t]
  \centering
  \includegraphics[width=0.5\textwidth]{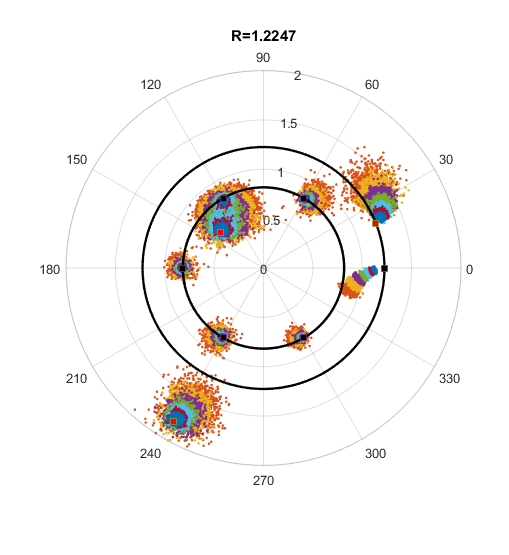}
  \caption{Root Neighborhoods for $7$ noise powers between $-22$dB and $-5$dB for $K=6$ Huffman zeros and $L-1=3$
  channel zeros.}\label{fig:receivedzerosperturbation}
\end{figure}
To derive such a quantized result we will exploit Rouch{\'e}'s Theorem to bound the single root neighborhoods by discs, see e.g.
\cite[Thm (1,3)]{Mar77}.
%
\begin{thm}[Rouch{\'e}]\label{thm:rouche}
  Let $\ux(z)$ and $\uw(z)$ be analytic functions in the interior to a simple closed Jordan curve $C$ and continuous on
  $C$. If 
  \begin{align}
    |\uw(z)|\leq |\ux(z)|,\quad z\in C,
  \end{align}
  then $\uy(z)=\ux(z)+\uw(z)$ has the same number of zeros interior to $C$ as does $\ux(z)$. 
\end{thm}
%
The Theorem allows to prove that the zeros of  polynomials are continuous functions of the coefficients, see \cite[Thm
(1,4)]{Mar77}.
%
%
However, to obtain an explicit robustness result for the zeros, we need a quantized version of the continuity, i.e., a
Lipschitz bound of the root functions with respect to the $\ell_{\infty}/\ell_2$ norm. 
As simple closed Jordan curves we will consider the Euclidean circle and the disc as its interior, which will contain
the single root neighborhoods.  Let us define for
$\alp_n\in\C$ the closed Euclidean ball (disc) of radius $\del>0$ and its boundary as 
\begin{align}
  B_n(\del)=B(\del,\alp_n) = \set{z\in\C}{|z-\alp_n|\leq \del}\quad,\quad C_n(\del)= \set {z\in\C}{|z-\alp_n|=\del}.
\end{align}
Let us consider an arbitrary polynomial (analytic function in $\C$) of order $\xord\geq 1$:
\begin{align}
  \ux(z)=\sum_{\xind=0}^{\xord} x_\xind z^\xind.
\end{align}
Then, its roots are functions of the polynomial coefficients $\vx\in\C^{\xord+1}$ given by
\begin{align}
 \alp_n= \alp_n(\vx)\quad,\quad n=1,\dots, \xord.
\end{align}
If the coefficients are disturbed by a vector $\vw\in\C^{\xord+1}$, the maximal perturbation of the zeros should be bounded
by
\begin{align}
  \max_n |\alp_n(\vx+\vw)-\alp_n(\vx)| \leq \del \cdot \Norm{\vx+\vw-\vx}_{2} = \del \cdot \Norm{\vw}_{2},
\end{align}
where the bound $\del=\del(\eps,\vx)>0$ is a \emph{local Lipschitz constant} for $\vw\in B(\eps,\vx)$, which we want to
derive. If the noise coefficient $w_\xord=-x_\xord$, i.e., the leading coefficient is vanishing, then we will set
$\alp_\xord(\vx+\vw)=0$, since the order of the perturbed polynomial would reduce to $\xord-1$.
We are now ready to prove the following local Lipschitz bound. We use here the assumption that one zero is outside the
unit circle, which is always the case for polynomials generated by autocorrelations. 
\begin{thm}\label{thm:zerodistortion}
  Let $\ux(z)\in\C[z]$ be a polynomial of order $\xord>1$  with simple zeros $\alp_1,\dots,\alp_{\xord} \subset \C$
  inside a circle of radius $R>1$ with minimal pairwise distance $\dmin>0$, i.e.
  \begin{align}
     \dmin :=\min_{n\not=k} |\alp_n-\alp_k| \quad,\quad R=\arg\max_n |\alp_n|.
  \end{align}
  Let $\vw\in\C^{\xord}$ with $\Norm{\vw}_{2}\leq \eps$ be an additive
  perturbation on the polynomial coefficients $\vx$ and $\del\in[0,\dmin/2)$. Then the $n$th zero $\zeta_n$ of the disturbed
  polynomial $\uy(z)=\ux(z)+\uw(z)$ lies in $B_n(\del)$ if
  \begin{align}
    \eps=\eps(\vx,\del)\leq
    \frac{|x_N| \del (\dmin-\del)^{N-1}}{\sqrt{1+ N}(R+\del)^\xord}\label{eq:noisebound}.
  \end{align}
\end{thm}
\begin{remark}
  The minimal pairwise distance of the zeros is also called zero separation, see for example \cite[Sec.11.4]{Zip93}.
\end{remark}
\begin{proof}
  The proof is a quantized version of the proof in \cite[Thm (1,4)]{Mar77}.
  Let us define the error polynomial
  \begin{align}
    \uw(z)= \sum_{n=0}^{\xord} w_n z^n\label{eq:wz}.
  \end{align}
  By defining $\underline{\vz}=(z^0,z^1,\dots,z^{\xord})^T$, we can upper bound the magnitude of $\uw$ with the
  Cauchy-Schwarz inequality
  \begin{align}
    |\uw(z)| = |\vw^T\underline{\vz}|\leq \Norm{\vw}_{2} \cdot \Norm{\underline{\vz}}_2 = \eps \cdot \Big(\sum_n
    |z^n|^2\Big)^{1/2} = \eps \cdot \Big(\sum_n |z|^{2n}\Big)^{1/2} = \eps\cdot  f(|z|).
  \end{align}
  Since $f(r)$ is monotone increasing\footnote{Note, $(r+\eps)^k > r^k +\eps^k+ \dots >r^k$ for $r,\eps>0$ and $k\geq 1$.} in $r>0$,
  the largest upper bound in $C_m(\del)$ is attained at $z=|\alp_m|+\del$ and hence
  \begin{align}
    f(|\alp_m|+\del)^2\leq\begin{cases} 
      1+ N\cdot (|\alp_m|+\del)^{2N}&,  |\alp_m|+\del>1\\
      1+ N\cdot(|\alp_m|+\del)^2 &,  |\alp_m|+\del\leq 1
    \end{cases}\label{eq:fbound}
  \end{align}
By assumption it holds $R=|\alpmax|>1$ which gives us the universal upper bound\footnote{This is actually Bernstein's Lemma.}
\begin{align}
  |\uw(z)| \leq  \eps \cdot \sqrt{1+N}(R+\del)^N \quad,\quad z\in\bigcup C_m(\del).
\end{align}
On the other hand, the magnitude of the original polynomial 
\begin{align}
  |\ux(z)| &= |x_\xord| \Pro_{n=1}^\xord |z-\alp_n| \quad,\quad z \in C_m(\del)\\
  &= |x_\xord| \Pro_n |\alp_m+\del e^{i\tht} -\alp_n|\quad,\quad \tht\in[0,2\pi)
\intertext{can be lower bounded by using the reverse triangle inequality \footnotemark}
   & \geq |x_\xord|\Pro_n | |\alp_m-\alp_n|-\del|
   \geq |x_\xord| \del \Pro_{n\not=m}  (\dmin-\del).\label{eq:brutallowerbound}
\end{align}
\footnotetext{Note, that $|\alp_l-\alp_n|>\dmin>\del$ for $l\not=n$.}%
Hence we get for all $z\in\bigcup C_n(\del)$:
\begin{align}
  |\ux(z)|\geq |x_\xord| \del (\dmin-\del)^{N-1}.
\end{align}
To apply Rouch{\'e}'s Theorem, we have to show $|\uw(z)|<|\ux(z)|$ for all $z\in\bigcup C_n(\del)$, which gives
us the universal bound 
\begin{align}
  \eps=\eps(\vx,\del) \leq \frac{|x_\xord| \del (\dmin-\del)^{N-1}}{\sqrt{1+ N}(R+\del)^N}.\label{eq:boundproof}
\end{align}
Since  $\del<\dmin/2$, all $B_n(\del)$ are disjoint and $\uy(z)$ has exactly one zero
in each $n$th ball $B_n(\del)$ by \thmref{thm:rouche}. Note, that $x_N=x_N(\valp)$ depends on the selected zeros and the
normalization $\Norm{\vx}=1$.
\end{proof}
%
\color{black}

\begin{remark}
  Let us note, that the bound \eqref{eq:boundproof} {\bfseries does not increases} with $\del$ for fixed $\vx,R$ and
  $\dmin$, see \figref{fig:theorem2_epsbound}. This behaviour is due to the continuity of the zeros very unlikely and
  hence caused by the worst bound in \eqref{eq:brutallowerbound}. In \secref{sec:sharperbounds} we will investiage in
  more detail the geometric structure of the zero placements, to obtain sharper stability bounds. 
  \begin{figure}
    \centering
    \includegraphics[width=0.5\linewidth]{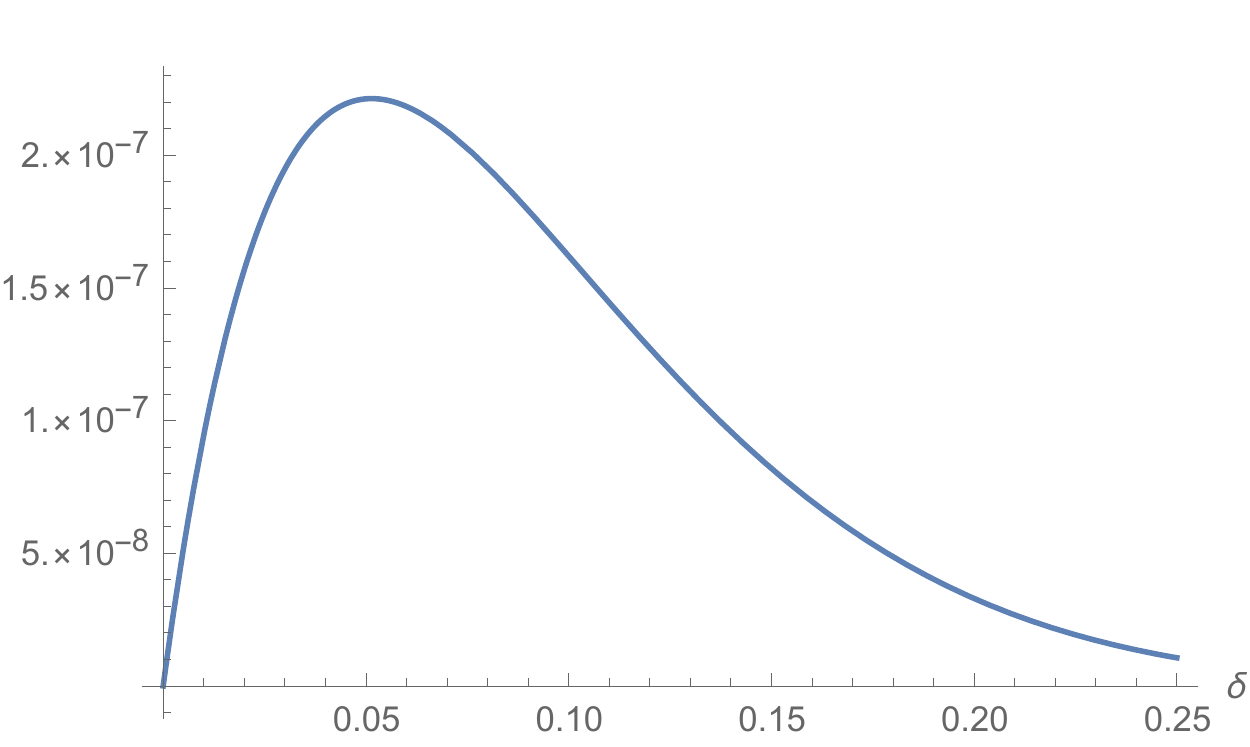}
    \caption{Noise bound \eqref{eq:boundproof}
    for fixed Radius $R=1.1$ and $\dmin=0.5$ over $\delta$.}
    \label{fig:theorem2_epsbound}
\end{figure}

\end{remark}
Furthermore, if $|\alpmax|=\const$ and $|x_N|=\const$, then a maximal separation of the zeros yields to robustness 
 against additive noise on the coefficients. Hence, if we place the zeros with maximal pairwise distance for fixed $R$, this
suggests a good BER performance for the RFMD decoder. Moreover, by setting $\del=\dmin/2$ the bound
\eqref{eq:noisebound} gives
\begin{align}
  \eps \leq \frac{|x_\xord|}{\sqrt{1+N}}\frac{\dmin^N}{2^N (|\alpmax|+\dmin/2)^N)} \label{eq:epsbad},
\end{align}
which is a upper threshold of the noise power under which no errors can occur. It can be seen that the noise bound increases if
$\dmin$ increases, which again validates a larger zero separation. 

For Huffman sequences with radius $R$ we obtain $\dmin=2R\sin(\pi/N)$ and hence \eqref{eq:epsbad} gives
\begin{align}
  \eps\leq \frac{1}{\sqrt{1+N}} \frac{1}{\sqrt{(R^{-2N}+1)}(1/\sin(\pi/N) +1)^N}\label{eq:noisebound_huff}.
\end{align}
Note, the bound becomes independent of $R$ if one zero is outside the unit circle and hence equal to $R$.  A plot for
different $N$ is given in \figref{fig:huffmaneps} for uniform radius $\Ropt$ in \eqref{eq:optimalradiusexactrho}.
\begin{figure}
  \centering
  \includegraphics[width=0.8\linewidth]{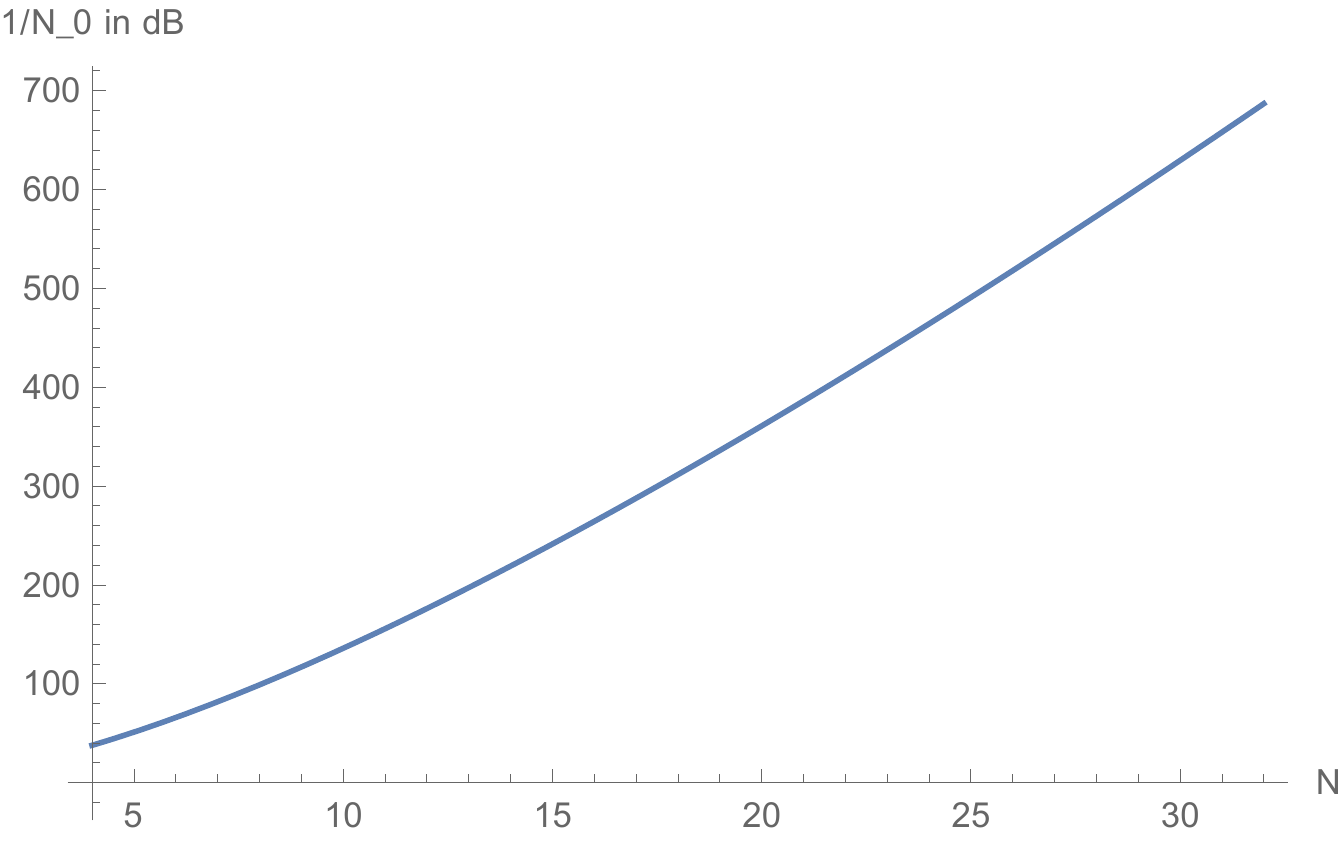}
  \caption{SNR bound \eqref{eq:noisebound_huff} with $\del=\delmax$ for Huffman sequences over various dimensions $N$ and uniform
  radius $\Ropt(N,1)$ in \eqref{eq:optimalradiusexactrho} allowing a perfect reconstruction.}
  \label{fig:huffmaneps}
\end{figure}
Moreover, if $|\alpmax|=t\dmin/2$ for some $t\geq 1$ then we get
\begin{align}
  \eps \leq \frac{|x_\xord|}{\sqrt{1+N}(t+1)^N}.
\end{align}
However, an increase of $t$ means an increase of the largest root, which is coupled by the leading coefficient, due to a
result of Cauchy
\begin{align}
  |\alpmax|\leq 1+ \max_{k<N} {\Big|\frac{x_k}{x_\xord}\Big|}
\end{align}
see for example \cite[Thm.(27,2)]{Mar77}.
If the energy of $\vx$ is normalized this gives  
\begin{align}
  |\alpmax| \leq 1 + \frac{1}{|x_\xord|}\quad \LRA \quad |x_\xord|\leq \frac{1}{|\alpmax|-1} 
\end{align}
since $|\alpmax|>1$.
Hence, if $|\alpmax|$ increases, the leading coefficient has to decrease and $\eps$ decreases rapidly, independently of
$\dmin$.


\paragraph{Zeros of Random Channels}

It is known, that a polynomial with i.i.d. Gaussian distributed coefficients has zeros concentrated around the unit
circle. If the order $\xord$ goes to infinity, all zeros will be uniformly distributed on the unit circle with probability one, see for
example \cite{PY14} In fact, this even holds for other random polynomials with non Gaussian distributions, see \cite{HN08}.
This is an important observation, since it implies for fixed $K$ and hence $R$, that an increase of $L$ will concentrate
the channel zeros on the unit circle, such that the channel zeros will not interfere with the codebook zeros, as long as
$R$ is sufficiently large. 

\begin{remark}
 The analysis of the stability radius for a certain zero-codebook and noise power, allows in principle an error detection for the RFMD
 decoder. Here, an error for the $l$th zero can only occur if the noise power is larger than the RHS of
 \eqref{eq:noisebound}. However, in the presence of the channel $\vh$, we can adopt the dimension $N$ and $x_N\ra x_\Nx
 h_{\Nh-1}$, if we assume the absolute values of the zeros of $\uH(z)$ are not larger than $R$. The minimal distance might
 be fulfilled with a certain probability. A precise analysis of the expectation might lead to upper bounds of the bit
 error probabilities of the RFMD decoder, which will be a future research topic.
 Note also, that is not clear, what the distribution of the disturbed zeros $\zeta_n$ are. 
  If they would be Gaussian known results of polar quantization might apply, see for example
 \cite{NCTH14}. Huffman sequences for $R=1$ are uniformly concentrated on a unit circle and show the best noise robustness
 \figref{fig:hufzerospread0}.
 \end{remark}

\newcommand{\dnn}{\ensuremath{d_{nn}}} \newcommand{\dcp}{\ensuremath{d_{cp}}}
\begin{figure}[t]
  \centering \def\svgwidth{0.65\textwidth} \footnotesize{ 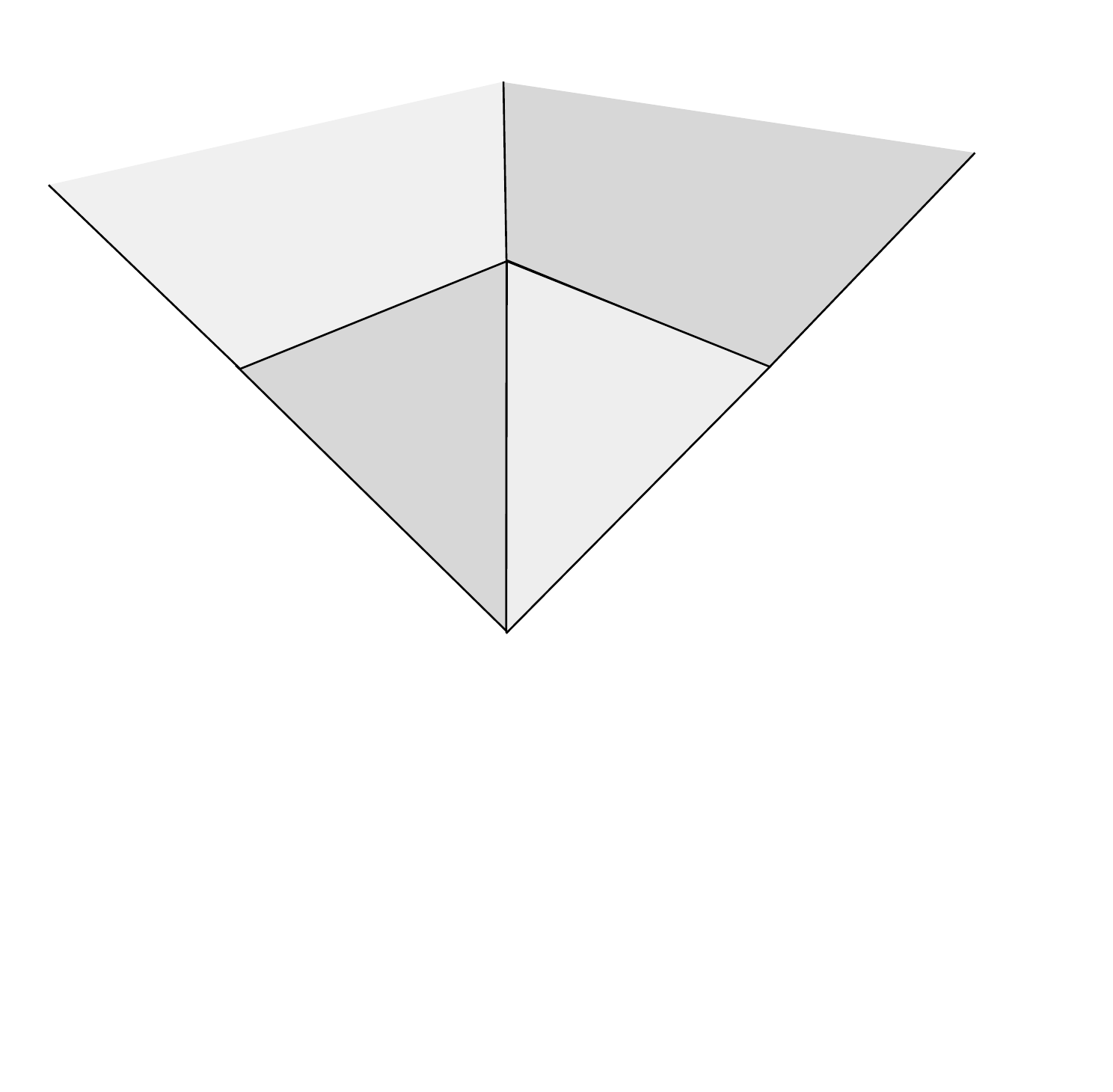 }
  \vspace{-0.28\textwidth} \caption{Zero Constellations, which allows largest non-overlapping uniform root-neighborhood
    discs for the Huffman BMOCZ scheme with largest radius $\delmax=\dmin/2$.}
  \label{fig:optimalradius}
\end{figure}

\subsection{Radius for Huffman BMOCZ Allowing Largest Uniform Root Neighborhood Discs}

To ensure robustness of a zero-based decoding against additive noise we need to place the
zero-constellations carefully, as will be pointed out in more detail in
\secref{sec:rootstability}.
%
%
\thmref{thm:zerodistortion} below suggest for  Huffman polynomials to place the zero-set $\Alp$
with maximal pairwise distance under the constraint of being uniformly
spaced conjugate-reciprocal pairs.  Here the distance between
conjugated pairs is given by
\begin{align}
  d_{cp}=R-R^{-1}
\end{align}
and the distance between next-neighbor pairs for the smaller radius $R^{-1}$ by
\begin{align}
  d_{nn}= 2R^{-1}\sin(\pi/\Nx).
\end{align}
%
%
%
Setting both distances equal, yields a zero-set $\Alp$ with maximal minimal pairwise distance $\dmin$, see
\figref{fig:optimalradius}. 

However, simulations of perturbed Huffman polynomials, see \figref{fig:zeroperturbationsim}, show a strong dependence of
the root neighborhood radius. In fact, an increasing of $R$ yields to an increasing of
the root neighborhood radius $\del$, which obtains its minimum if $R=1$. However, if $R$ gets to small, the root neighborhood
of the reciprocal-pairs will overlap. To address this problem we will introduce $\lam\geq 1$ as a scaling parameter
which yields to
\begin{align}
  \lam d_{cp}=d_{nn}
  \quad&\LRA\quad \lam(R^2-1)=Rd_{nn}=2\sin(\pi/\Nx)\label{eq:dndcc}\\
  \quad&\RA\quad\Ropt(K,\lam)= \sqrt{1+\frac{2}{\lam}\sin(\pi/\Nx)}\simeq
  \sqrt{1+\frac{2\pi}{\Nx}}\label{eq:optimalradiusexactrho},
\end{align}
which is  bounded between
\begin{align}
  1+\frac{\pi}{\lam\Nx} \leq \Ropt(K) \leq e^{\pi/2\lam\Nx}
  \label{eq:optimalradius}.
\end{align}

Therefore, we will in \secref{sec:rootstability} investigate the radius dependence in more detail.  Finding the optimal
radius for Huffman sequences yielding to the optimal Voronoi cells $\Dset_\nx$ is in fact a quantization problem, see
for example \cite{KJ17}. Note, that the zeros for Huffman BMOCZ are not the centroids of the Voronoi cells, which
suggest a much more complex metric for an optimal quantization, see \figref{fig:zeroperturbationsim}. 
From the simulation of the BER performance we observed $\lam\simeq 2$, which might be also $L$ dependent. 
%

\newcommand{\dminfour}{\ensuremath{\del_{\text{max}}^{(4)}=0.455}}
\newcommand{\dmineight}{\ensuremath{\del_{\text{max}}^{(8)}=0.288}}
\newcommand{\dminsixteen}{\ensuremath{\del_{\text{max}}^{(16)}=0.165}}
\newcommand{\done}{\ensuremath{d_{1}}}
\newcommand{\dtwo}{\ensuremath{d_{2}}}
\newcommand{\dthree}{\ensuremath{d_{3}}}
\newcommand{\dfour}{\ensuremath{d_{4}}}
\newcommand{\dfive}{\ensuremath{d_{5}}}
\newcommand{\dsix}{\ensuremath{d_{6}}}
\newcommand{\aradius}{\ensuremath{a}}
\newcommand{\phimath}{\ensuremath{\phi}}
\newcommand{\snreps}{\ensuremath{\text{SNR}=1/\eps^2}}
\begin{figure}[t] 
  \begin{subfigure}[b]{0.323\textwidth} 
\hspace{-.7cm} 
\includegraphics[width=1.1\textwidth]{%
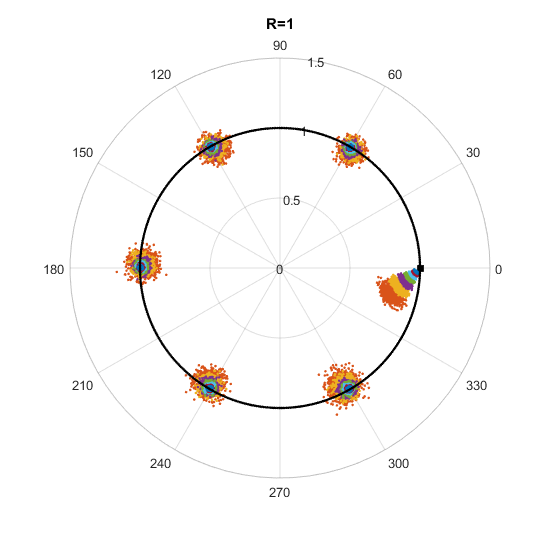}
\caption{Radius $R=1$ with $\lam=\infty$.}
      \label{fig:hufzerospread0}
\end{subfigure} 
\begin{subfigure}[b]{0.323\textwidth} 
\hspace{-3ex} 
  \includegraphics[width=1.1\textwidth]{%
  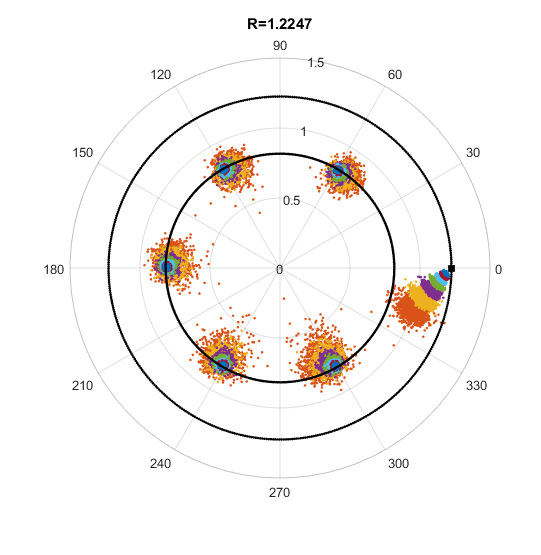}
  \caption{Radius $R=1.2247$ with $\lam=2$.}
        \label{fig:hufzerospread1}
\end{subfigure} 
\begin{subfigure}[b]{0.323\textwidth}
\includegraphics[width=1.1\textwidth]{%
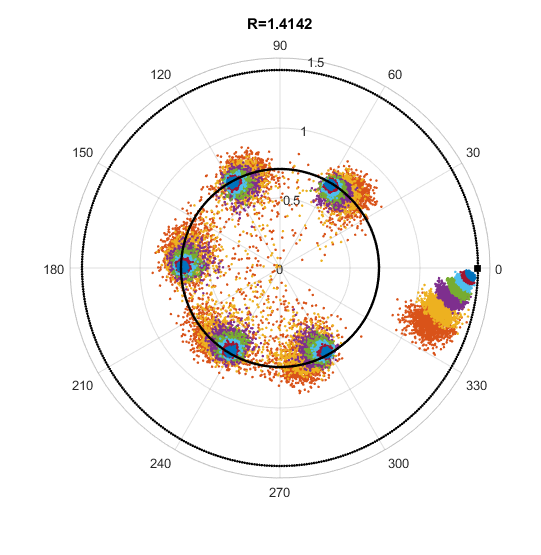}
\caption{Radius $R=1.4142$ with $\lam=1$.}
    \label{fig:hufzerospread2} 
  \end{subfigure} 
  \caption{Simulation of $7$ different SNR values from $22$dB to $5$dB for a fixed  Huffman sequence with $N=6$ over
  different radii.}\label{fig:zeroperturbationsim}
\end{figure}
\subsection{PAPR for Huffman Sequences with Uniform Radius}
From \eqref{eq:hufflastfirst}  we get for the magnitudes
\begin{align}
  |x_\Nx|^2,|x_0|^2  \in \Big[\frac{1}{{1+R^{2\Nx}}},\frac{1}{{1+R^{-2\Nx}}}\Big],
\end{align}
where the maximum is attained if $\vb=\zero$ or $\vb=\eins$, the all zero or all one bit vector.
By noting that the first and last coefficient magnitude \eqref{eq:hufflastfirst}  exploit a symmetry for $2\Norm{\vb}_1$
and $K-2\Norm{\vb}_1$, we only have to average for uniform bit distribution over $\Norm{\vb}_1\in\{0,\dots,K/2\}$
(assuming $K$ even), which
gets \Peter{was wrong}:
\begin{align}
  \Expect{\Norm{\vx}_{\infty}^2} 
  &=\frac{1}{2^{K/2}}\frac{1}{R^{2K}+1}\sum_{n=1}^{2^{K/2}} R^{-2\sum_{k=1}^{K/2} b_k^{(n)}}
  =\frac{1}{2^{K/2}}\frac{1}{R^{2K}+1} \sum_{m=0}^{{K/2}} \binom{K/2}{m}
  R^{-2m}\\
  &=\left(\frac{1+R^{-2}}{2}\right)^{\frac{K}{2}}\frac{1}{R^{2K}+1}
\end{align}
Since the Huffman sequences have all unit energy, the \emph{peak-to-average-power ratio} is for the optimal radius
$R=\Ropt(K,1)$ in \eqref{eq:optimalradiusexactrho} for large $K$ 
\begin{align}
  \PAPR& = (\Nx+1)\frac{\Expect{\Norm{\vx}_{\infty}^2}}{\Expect{\Norm{\vx}_2^2}}
  = \frac{(\Nx+1)((1+R^{-1})/2)^{K/2}}{R^{2\Nx} + 1}
  {\simeq}
  \frac{\Nx+1}{(1+2\pi/\Nx)^{\Nx}+1}\\
  &{\leq}  \frac{\Nx+1}{2+2\pi}\simeq \frac{\Nx+1}{8.28}\simeq \Nx/9,
\end{align}
which is typically for a multi-carrier system, such as OFDM \cite{BF11}.

%% file: OptimaRadiusHuffmanBMOCZ.pdf_tex
\begingroup%
  \makeatletter%
  \providecommand\color[2][]{%
    \errmessage{(Inkscape) Color is used for the text in Inkscape, but the package 'color.sty' is not loaded}%
    \renewcommand\color[2][]{}%
  }%
  \providecommand\transparent[1]{%
    \errmessage{(Inkscape) Transparency is used (non-zero) for the text in Inkscape, but the package 'transparent.sty' is not loaded}%
    \renewcommand\transparent[1]{}%
  }%
  \providecommand\rotatebox[2]{#2}%
  \ifx\svgwidth\undefined%
    \setlength{\unitlength}{428.36607983bp}%
    \ifx\svgscale\undefined%
      \relax%
    \else%
      \setlength{\unitlength}{\unitlength * \real{\svgscale}}%
    \fi%
  \else%
    \setlength{\unitlength}{\svgwidth}%
  \fi%
  \global\let\svgwidth\undefined%
  \global\let\svgscale\undefined%
  \makeatother%
  \begin{picture}(1,0.98148354)%
    \put(0,0){\includegraphics[width=\unitlength,page=1]{OptimaRadiusHuffmanBMOCZ.pdf}}%
    \put(0.50322299,0.50970804){\color[rgb]{0,0,0}\makebox(0,0)[lb]{\smash{\Ropt}}}%
    \put(0.48241083,0.63688745){\color[rgb]{0,0,0}\makebox(0,0)[lb]{\smash{\del}}}%
    \put(0.54906608,0.79939436){\color[rgb]{0,0,0}\makebox(0,0)[lb]{\smash{\del}}}%
    \put(0,0){\includegraphics[width=\unitlength,page=2]{OptimaRadiusHuffmanBMOCZ.pdf}}%
    \put(0.63254848,0.55216348){\color[rgb]{0,0,0}\makebox(0,0)[lb]{\smash{\dnn}}}%
    \put(0.59213914,0.71380079){\color[rgb]{0,0,0}\makebox(0,0)[lb]{\smash{\dcp}}}%
    \put(0,0){\includegraphics[width=\unitlength,page=3]{OptimaRadiusHuffmanBMOCZ.pdf}}%
  \end{picture}%
\endgroup%

%% file: simulations.tex
\section{Numerical Simulations}\label{sec:simulations}

We simulated with MatLab 2017a the \emph{bit-error-rate} (BER) over the rSNR \eqref{eq:rSNR} for $\Nh$ Rayleigh fading multipaths with power
delay profile exponent $\pd<1$.

In the simulation, we scaled the transmit signals $\vx$ by $\sqrt{N}$ and the channel by
$\sqrt{1/\Expect{\Norm{\vh}_2^2}}$, such that the received average
power will be normalized and equal to the transmitted average power, independent of $N, \Nh$ and $\vp$. Hence  we obtain
$\rSNR=\SNR=1/N_0$.  The energy per bit  is then
 \begin{align}
   E_b=\frac{N}{\Nh}=\etaE^{-1},
\end{align}
which is equal to the inverse of the \emph{bit rate} $\etaE$ per symbol time%
.
Hence, the $\SNR$ per bit is 
\begin{align}
  \frac{E_b}{N_0}= \frac{1}{\etaE\cdot N_0} = \frac{\SNR}{\etaE}
\end{align}
see for example \cite[pp.97]{PS08}.

\begin{figure}
\hspace{-1.4cm}
\begin{subfigure}{0.49\textwidth}
\includegraphics[width=1.24\linewidth]{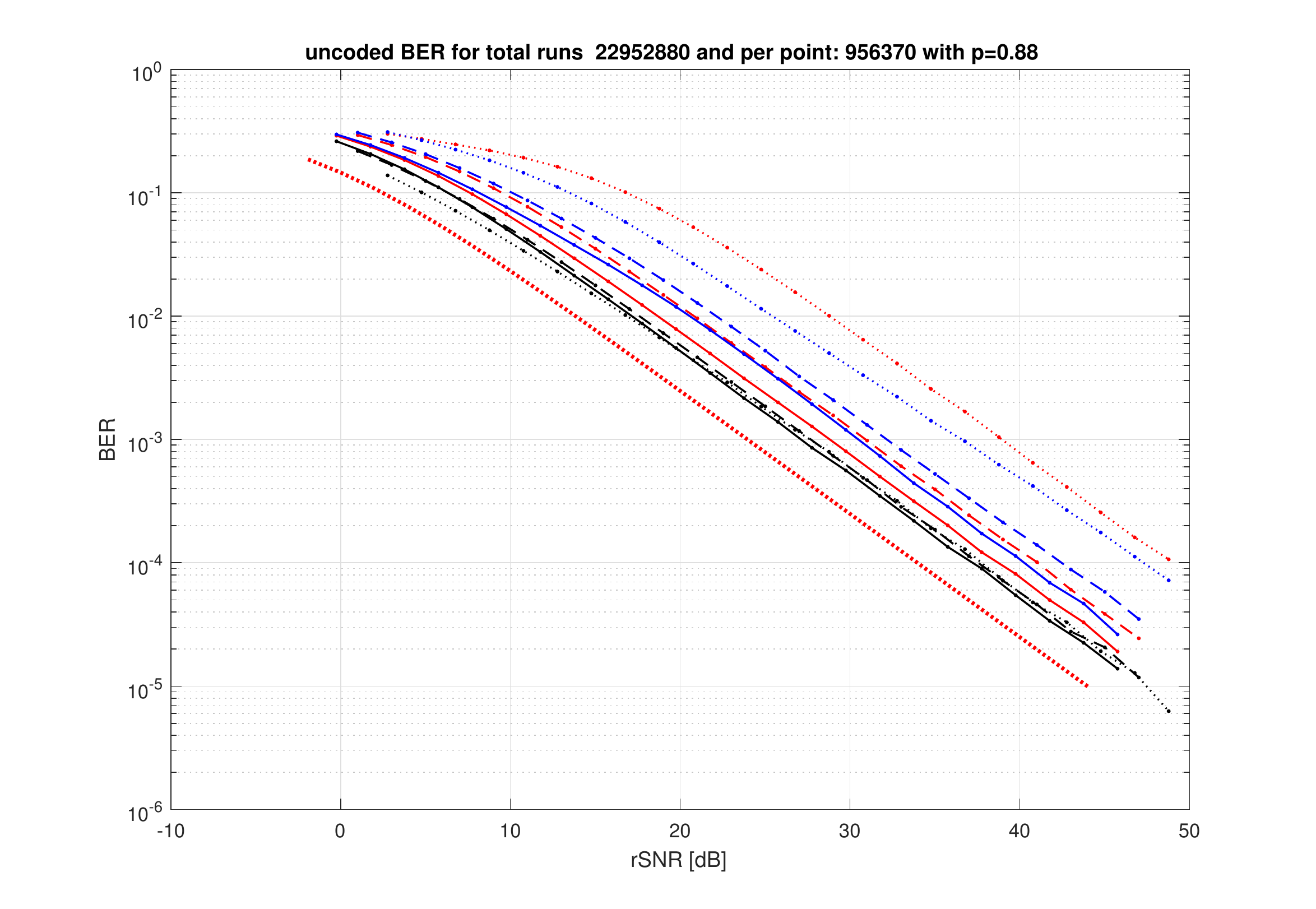}
  \caption{Bit error results over $\SNR$ for $\pd=0.88$}
  \label{fig:MLandRFMD_snr}
\end{subfigure}
\hspace{0.85cm}
\begin{subfigure}{0.49\textwidth}
  \includegraphics[width=1.24\linewidth]{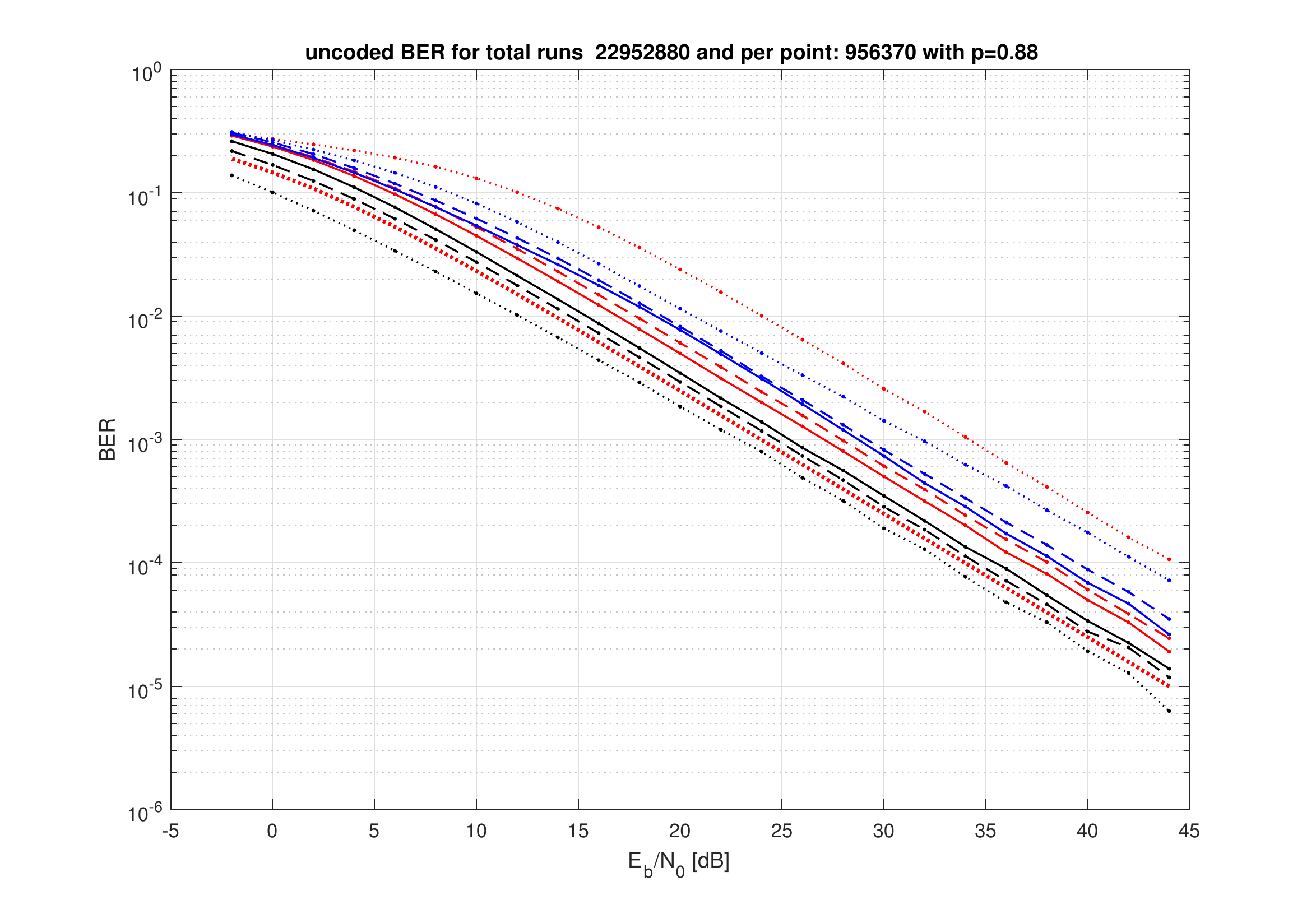}
  \caption{Bit error results over $E_b/N_0$ for $\pd=0.88$}
  \label{fig:MLandRFMD_ebno}
\end{subfigure}\\
%
\begin{subfigure}{0.49\textwidth}
\hspace{-1.3cm}
  \includegraphics[width=1.27\linewidth]{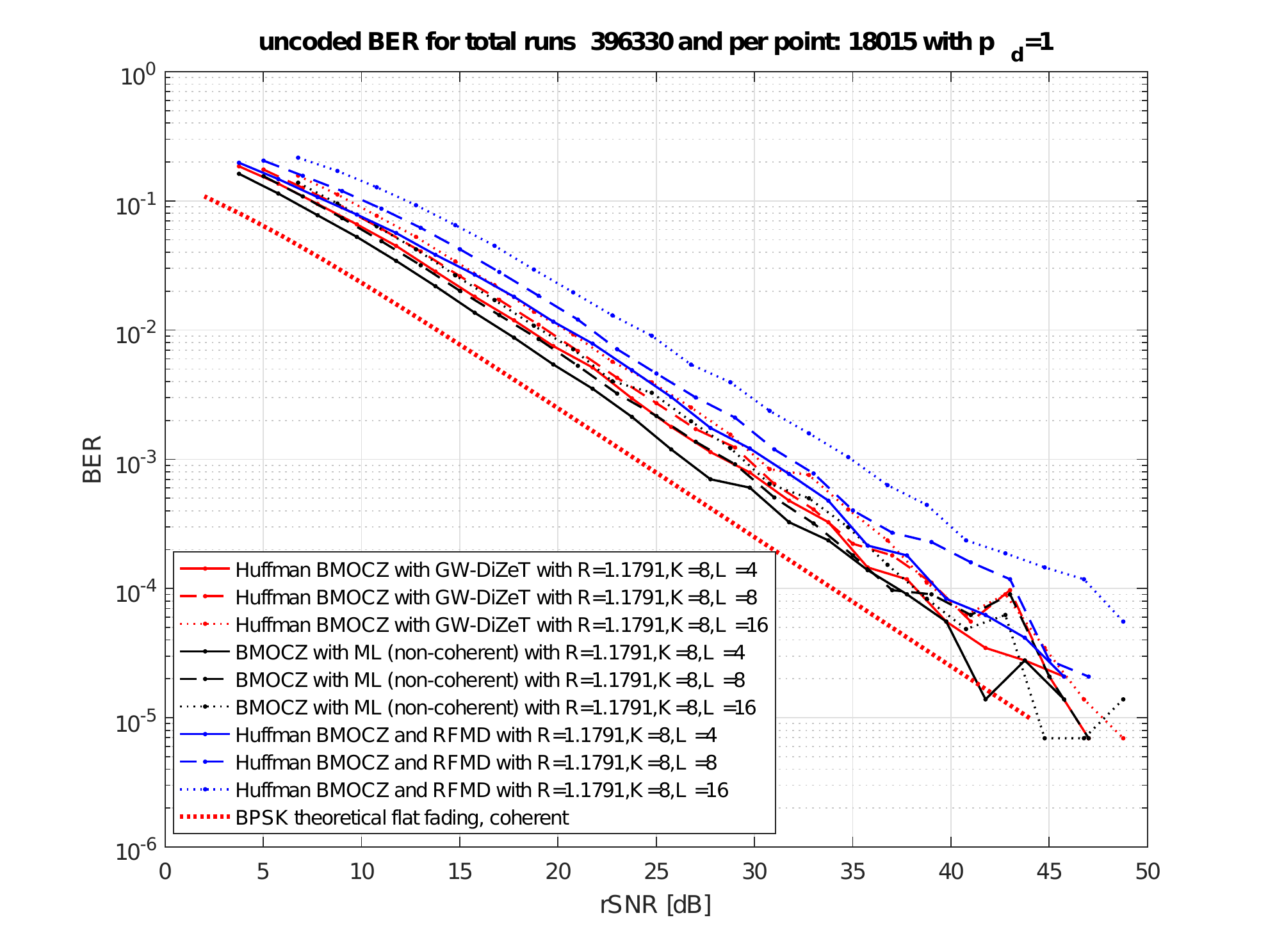}
  \caption{Bit error results over $\SNR$ for $\pd=1$.}
  \label{fig:MLandRFMD_snrpd1}
\end{subfigure}
%
\begin{subfigure}{0.49\textwidth}
  \includegraphics[width=1.2\linewidth]{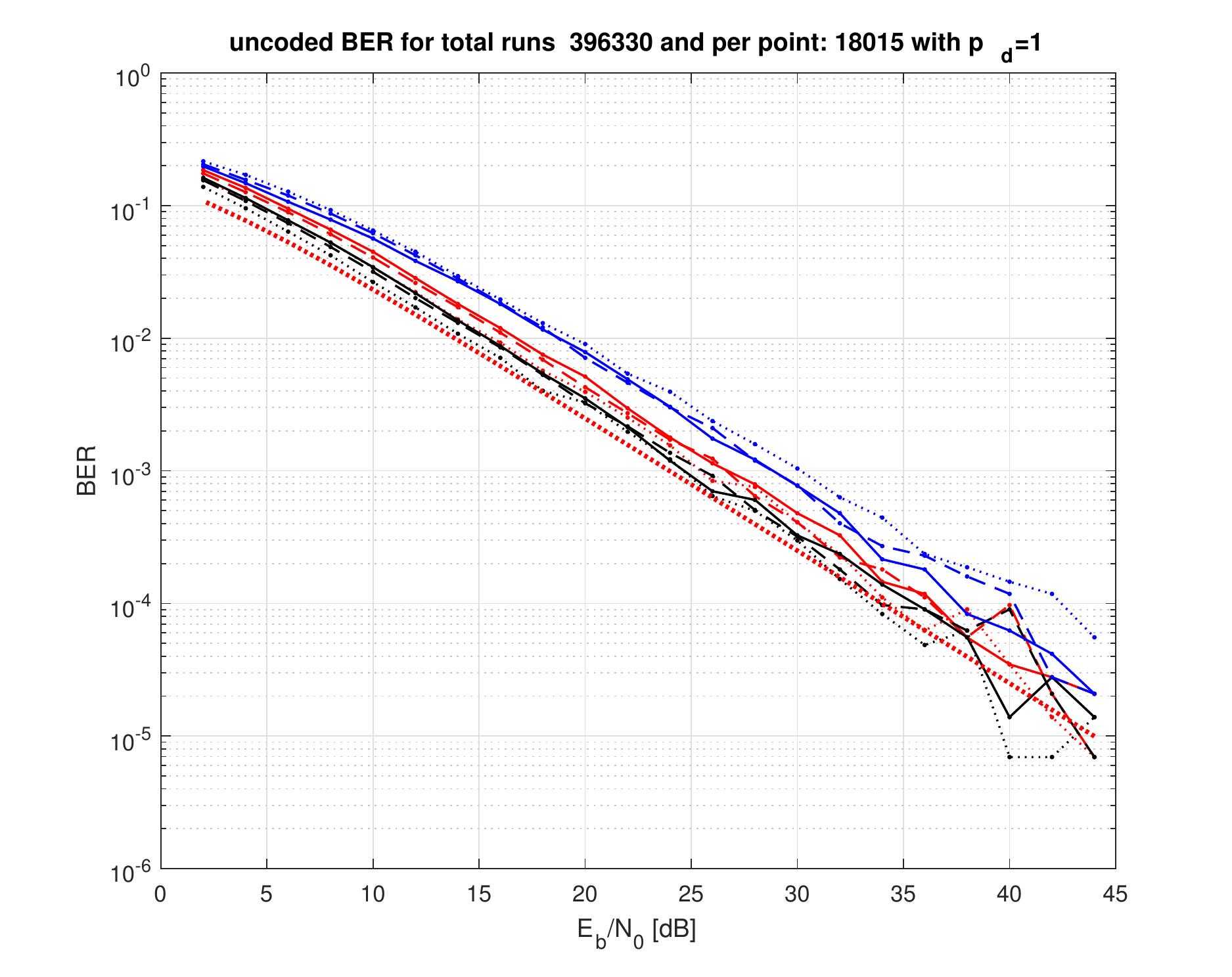}  
  \caption{Bit error results over $E_b/N_0$ for $\pd=1$.}
  \label{fig:MLandRFMD_ebnopd1}
\end{subfigure}
 \caption{Huffman BMOCZ with RFMD, GW-DiZeT and ML (with known $N_0$ and $\vp$) decoder for $\Nx=8$ and channel length
 $\Nh=4,8$ and $16$ under different power delay profiles $\pd$.}
\label{fig:MLandRFMD}
\end{figure}

As an ultimate benchmark in all simulations, we will compare to the coherent case, where the frequency selective channel
is modulated by OFDM with a binary phase shift keying (BPSK). Transforming the linear convolution for i.i.d. Gaussian
CIR in time domain to the frequency domain, yields to $N$ parallel flat fading channels. Assuming a sequential block
transmission, the cyclic prefix, allows to communicate $N$ bits per channel use and results therefore in coherent BPSK
flat fading.  The BER for BPSK over a flat fading channel $h_0=|h_0| e^{i\phi}$, with known phase
$\phi$ and $\Expect{|h_0|^2}=1$ is equivalent to the bit error probability (one bit per symbol duration) given by
\begin{align}
  P_e = \frac{1}{2} \left( 1- \sqrt{\frac{E_b \Expect{|h_0|^2}/N_0}{1+E_b\Expect{|h_0|^2}/N_0}}\right)
  = \frac{1}{2} \left(1- \sqrt{\frac{\rSNR}{1+\rSNR}}\right)
\end{align}
since $\sig_h^2=1$ we have in \figref{fig:pilotQPSKvsBMOCZ} $E_b/N_0= E_b\Expect{|h_0|^2}/N_0=\rSNR$, pictured as a
thick doted red curve. The BPSK coherent flat fading can be seen as the best binary signaling scheme performance if no
multi-path diversity is exploit (no outer codes). Note, that our scheme prevent is therefore robust against
\emph{inter-symbol-interference} (ISI), given by superposition of overlapping symbols due to the multipath delays.
It is still unclear how to exploit fully the multipath diversity gain in one-shot at the receiver without knowledge of
the CIR.  However, the DiZeT decoder performs very close to the ML decoder and coherent uncoded OFDM with BPSK, see
\figref{fig:MLandRFMD}.  Note, all simulated BER curves are for uncoded bits.  In \figref{fig:BMOCZ_sdp} the BER for the
SDP denoising \eqref{eq:sdpdenoising} with the estimate channel autocorrelation via \eqref{eq:zfah} are simulated for
$K=8$ and $L=4$ with flat power delay profile.  The denoised signal $\hvx$ from the SDP is then either decoded by the
GW-DiZeT decoder or the RFMD  decoder. The results show a $2$dB lose compared to GW-DiZeT without denoising. The reason
for the performance lose is first in the bad estimation of the channel autocorrelation and secondly due to the
simultaneously denoising of the channel and signal.  Since the SDP does not emphasize the signal reconstruction quality,
the quality in signal recovery is in sum worse as for the direct decoding approaches. However, the knowledge of the
channel might help for other purposes.

\begin{figure}[t]
    \hspace{-0.8cm}
\begin{subfigure}{0.49\textwidth}
\vspace{-2cm}
\includegraphics[width=1.15\textwidth]{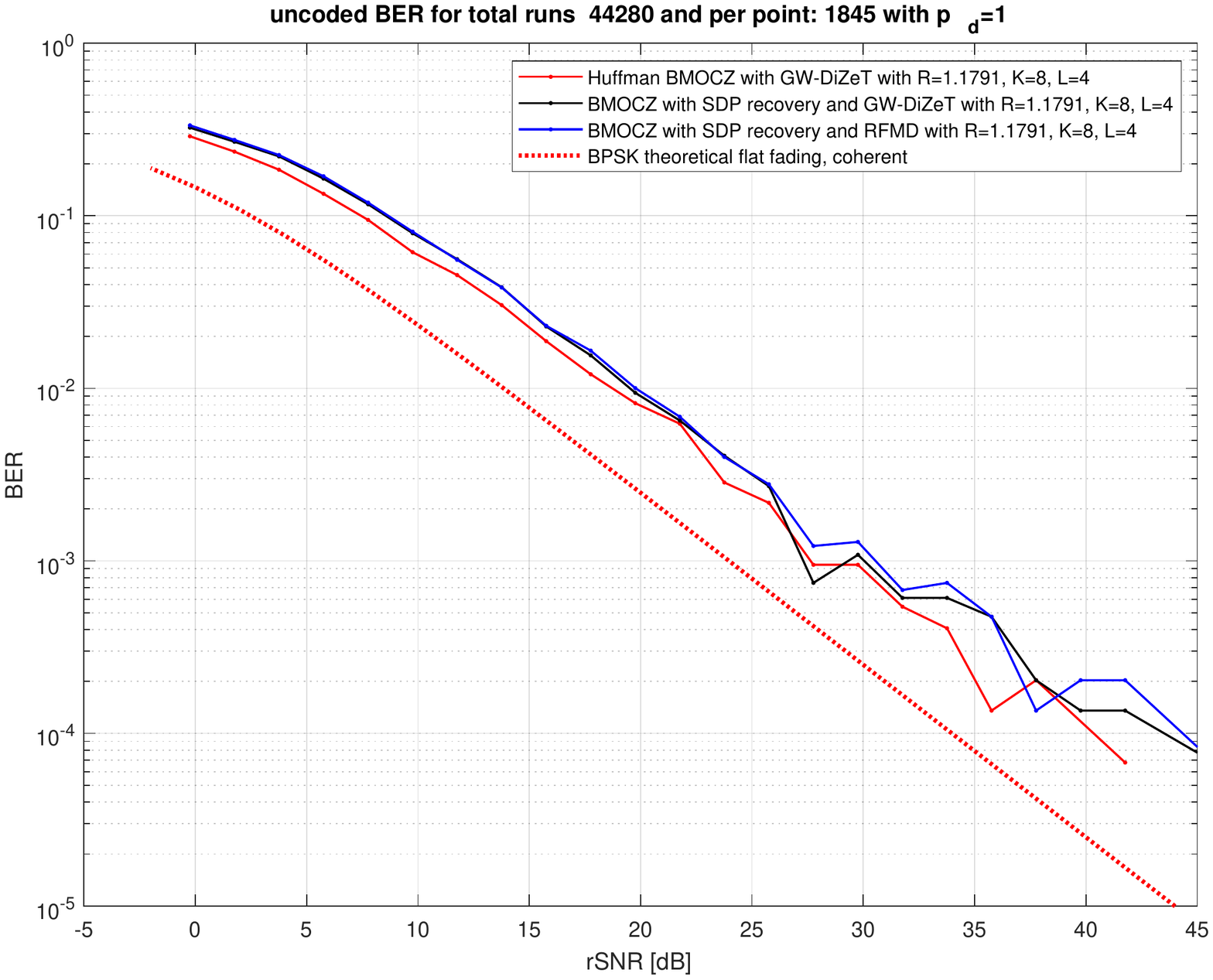}  
  \label{fig:MLandRFMD_snr2}
\end{subfigure}
\hspace{0.6cm}
\begin{subfigure}{0.485\textwidth}
\vspace{-2cm}
    \includegraphics[width=1.12\textwidth]{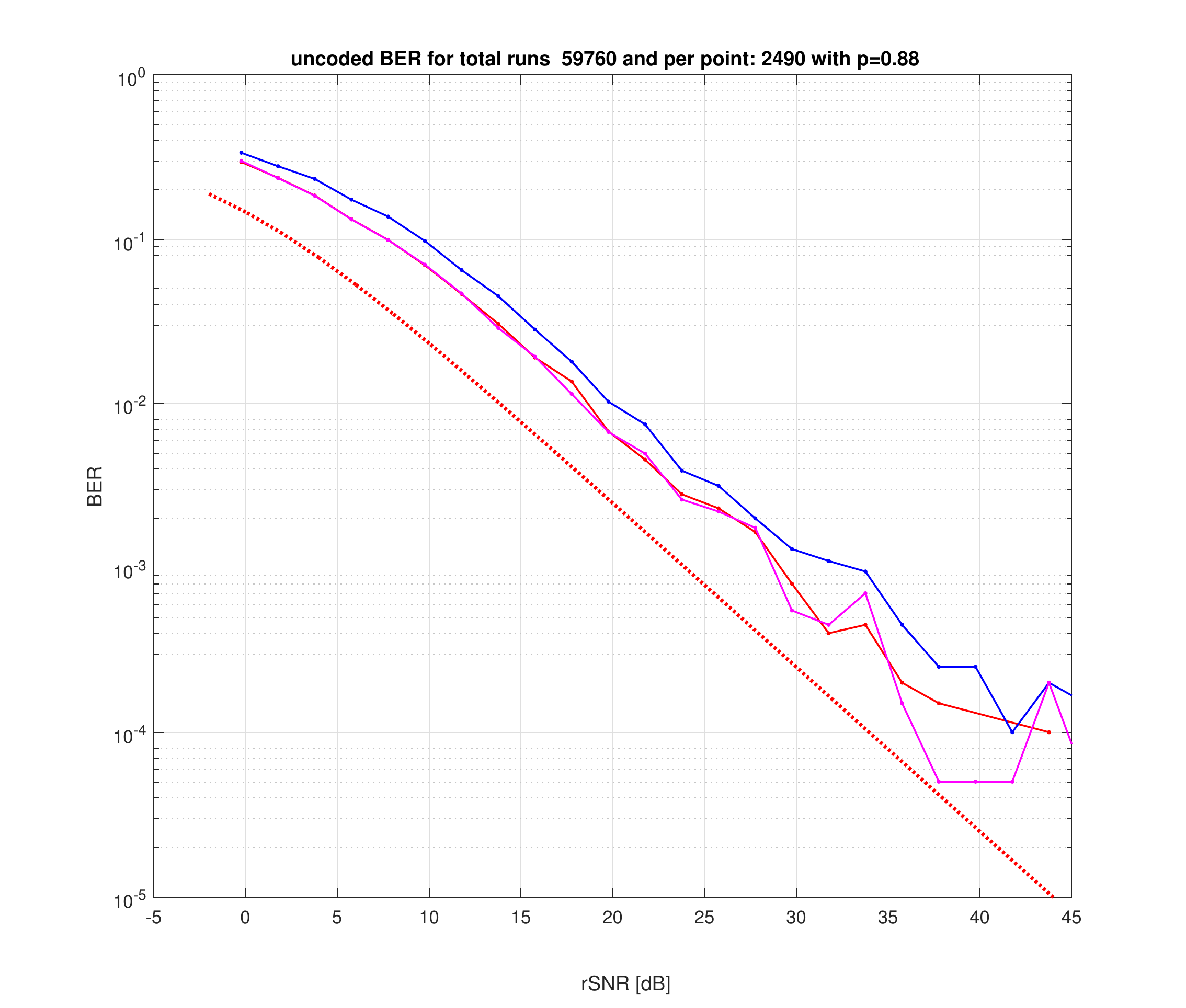}
  \label{fig:BMOCZ_sdp_pd13}
\end{subfigure}  
 \vspace{-2.9cm}
  \caption{GW-DiZeT decoder versus SDP denoising with GW-DiZeT and RFMD decoding.}
  \label{fig:BMOCZ_sdp}
\end{figure}

%% file: MaryMOCZ.tex
\section{Higher Order Modulations for MOCZ}

It is possible to use higher order modulation for MOCZ if we combine multiple autocorrelation codebooks. We will
introduce two simple $M-$ary MOCZ schemes by quantizing the radius or the phase. 
See \figref{fig:QPMOCZ} for an additional phase quantization with
$\{1,e^{i\tht}\}$ where we set $\tht=2\pi/2\Nh$, which allows $4$ positions for each zero and hence encodes two bits of
information. The decoder is for each $l$th zero
\begin{align}
  b_{\nx,1} &= \begin{cases} 1 &, |\uY(\alpl^+)| + |\uY(\alpl^+ e^{i\tht}| < |\uY(\alpl^-)|
  + |\uY(\alpl^- e^{i\tht})| \}\\
  0 &,  \text{else}\end{cases}\\
  b_{\nx,2} & = \begin{cases} 1 &, |\uY(\alpl^+)| < |\uY(\alpl^+e^{i\tht}| \\
  0 &,  \text{else}\end{cases}
\end{align}

\renewcommand{\bottomfraction}{0.9}
\renewcommand{\topfraction}{0.9}
\renewcommand{\textfraction}{0.01}
\begin{figure}[H]
\centering
  \vspace{-0.5cm}
  \includegraphics[width=0.4\textwidth]{MaryPMOCZ}  
  \vspace{-0.5cm}
  \caption{M-PMOCZ scheme with two $M/2$ Phase positions and one radius $R$ for $M/2=q_\phi$ odd.}
  \label{fig:MaryPMOCZ}
\end{figure}
\paragraph{M-PMOCZ}
Similar, we can implement an $M-ary$ PMOCZ scheme for $M=2P$, by allowing $P=2^m$  phase positions in 
\begin{align}
  \phi_l + \{-\tht_{P/2}, \dots, -\tht_2,-\tht_1, \tht_1,\tht_2,\dots, \tht_{P/2}\}
\end{align}
for $\angle(\alpl)$. This allows to encode $M$ constellations for each zeros and hence 
\begin{align}
  B= \Nx\log_2 M= \Nx\log_2(2\cdot 2^{m}) =\Nx(m+1)
\end{align}
bits per sequence yielding to a bit rate
\begin{align}
  \etaE=\frac{B}{N} = \frac{\Nx(m+1)}{\Nx+\Nx}.
\end{align}
\figref{fig:MaryPMOCZ} shows for the $\nx$th zero an encoding for $q_{\phi}=M/2$ phases. For $q_{\phi}$ even, we choose
the next-neighboor (nn) phase distance as
\begin{align}
  \phi_{nn} = \frac{\pi}{q_{\phi} \Nx}
\end{align}
This yields to a uniform placement of the zero positions on the circles with radius $R$ respectively $R^{-1}$.

\paragraph{N-RMOCZ}

If we modulate on each phase over $N$ radii we obtain an $N-$RMOCZ scheme. 
Let us choose $N$ radii as 
\begin{align}
1<  R_1<R_2<\dots<R_N.
\end{align}
Then we encode for $b_\nx^{n}\in\{0,1\}$ for $n=1\dots N$ to
\begin{align}
  \alp_\nx =e^{i\phil} 
          \begin{cases} 
              R_2      &, b_\nx^{(1)}=1, b_\nx^{(2)}=1\\
              R_1      &, b_\nx^{(1)}=1, b_\nx^{(2)}=0\\
              R_1^{-1} &, b_\nx^{(1)}=0, b_\nx^{(2)}=0\\
              R_2^{-1} &, b_\nx^{(1)}=0, b_\nx^{(2)}=1\\ 
  \end{cases}
\end{align}
where the zero-codebook is given by 
\begin{align}
  \Zero_{\Nh,N}=\set{\valp\in\C^\Nh}{\alp_\nx \in  \{\alp_{\nx,1},\dots,\alp_{\nx,N},\alp_{\nx,1}^-,\dots,\alp_{\nx,N}^-\}}.
\end{align}
For decoding we use a bisection decision for the $N$th bit by decoding
\begin{align}
  b_\nx^{(1)} &= \begin{cases}
   0 &, |\uY(\alp_{\nx,1}^+) +\uY(\alp_{\nx,2}^+| < |\uY(\alp_{\nx,1}^-) + \uY(\alp_{\nx,1}^-|\\
   1 &, \text{else}
  \end{cases}\\
  b_\nx^{(2)} &= \begin{cases}
    0 &, |\uY(e^{i\phil}R_1^{2b_l^{(1)}-1})| < |\uY(e^{i\phil}R_2^{2b_\nx^{(1)}-1})| \\
   1 &, \text{else}
  \end{cases}
\end{align}

\begin{figure}
\begin{subfigure}{0.5\textwidth}
  \hspace{-1.25cm}
  \includegraphics[width=1.25\linewidth]{2and4and8PMOCZ_runpp8k_L8K1to16_rSNR}  
\end{subfigure}
\hspace{-0.35cm}
\begin{subfigure}{0.5\textwidth}
  \includegraphics[width=1.25\linewidth]{2and4and8PMOCZ_runpp8k_L8K1to16_EbN0}  
\end{subfigure}
\caption{2,4,8-PMOCZ with $1$ $2$ resp. $4$ Phase positions with $p_d=1$ for $\Nx=16$ and $\Nh=1,9,17,24$ over \rSNR.}
\label{fig:QPMOCZ}
\end{figure}

%% file: training.tex
\section{Comparison to Training Schemes}

We will compare our noncoherent BMOCZ scheme to noncoherent QPSK with pilot signaling. 
Considering one-shot scenarios, there are not many noncoherent comparisons possible in this scenario. We refer the
reader to \cite{JHV15},\cite{Cho18} and \cite{FHV13} based on self-coherent OFDM schemes. However, our proposed Huffman BMOCZ
schemes outperforms their BER performance. We assume in the simulations the following scenario
\begin{itemize}
  \item Channel is time-invariant for $N$ time-instants
  \item Receive duration is $N= \Nx+\Nh$
  \item Block length of transmitted signal is $\Nx+1$ 
  \item We have independent Rayleigh distributed channel gains $h_\nh\in\CN(0,\pd^\nh)$ with  $\pd\leq 1$.
 \item After each block transmission the CIR can change arbitrary, e.g., caused by a fast-varying channel or due to a
   sporadic one-shot communication, where the next block message might occur much later, such that the channel state and user
   position can change drastically, resulting in fully uncorrelated CIRs.  
\end{itemize}
Note, the maximum-likelihood (ML) detection is equal to the maximum a posteriori detection (MAP), if the channel is known.
If we do joint ML over $\vx,\vh$ then the ML and LS would also be the same, but this requires blind deconvolution, see
\cite[(12-16)]{MMF98}.
We will compare to the following
scenarios
\begin{enumerate}
  \item BPSK and full channel knowledge (coherent) by using .
      Zero-forcing (ZF) and hard thresholding bit wise
        \begin{align}
          m_\nx = \begin{cases} 1, \Re(\hvx_\nx)>0\\
            0, \Re(\hvx_\nx)<0\end{cases}\quad,\quad \nx=1,2,\dots,\Nx
        \end{align}
  \item QPSK with $\Nh$ pilots. Here we assume that $\Nx=2\Nh$. We decode by separating Real and Imaginary part
    \begin{align}
      m_\nx =\begin{cases} 1, \Re(\hvx_\nx)>0\\
        0, \Re(\hvx_\nx)<0\end{cases}\quad,\quad m_{\nx+\Nh}        =\begin{cases} 1, \Im(\hvx_\nx)>0\\
        0, \Im(\hvx_\nx)<0\end{cases}\quad,\quad \nx=1,2,\dots,\Nx
    \end{align}
    This follows form the minimum distance decoder, see \cite[(7.25)]{LM93}.
\end{enumerate}

\begin{figure}[t]
  \hspace{-2.3cm}
  \includegraphics[width=1.24\linewidth]{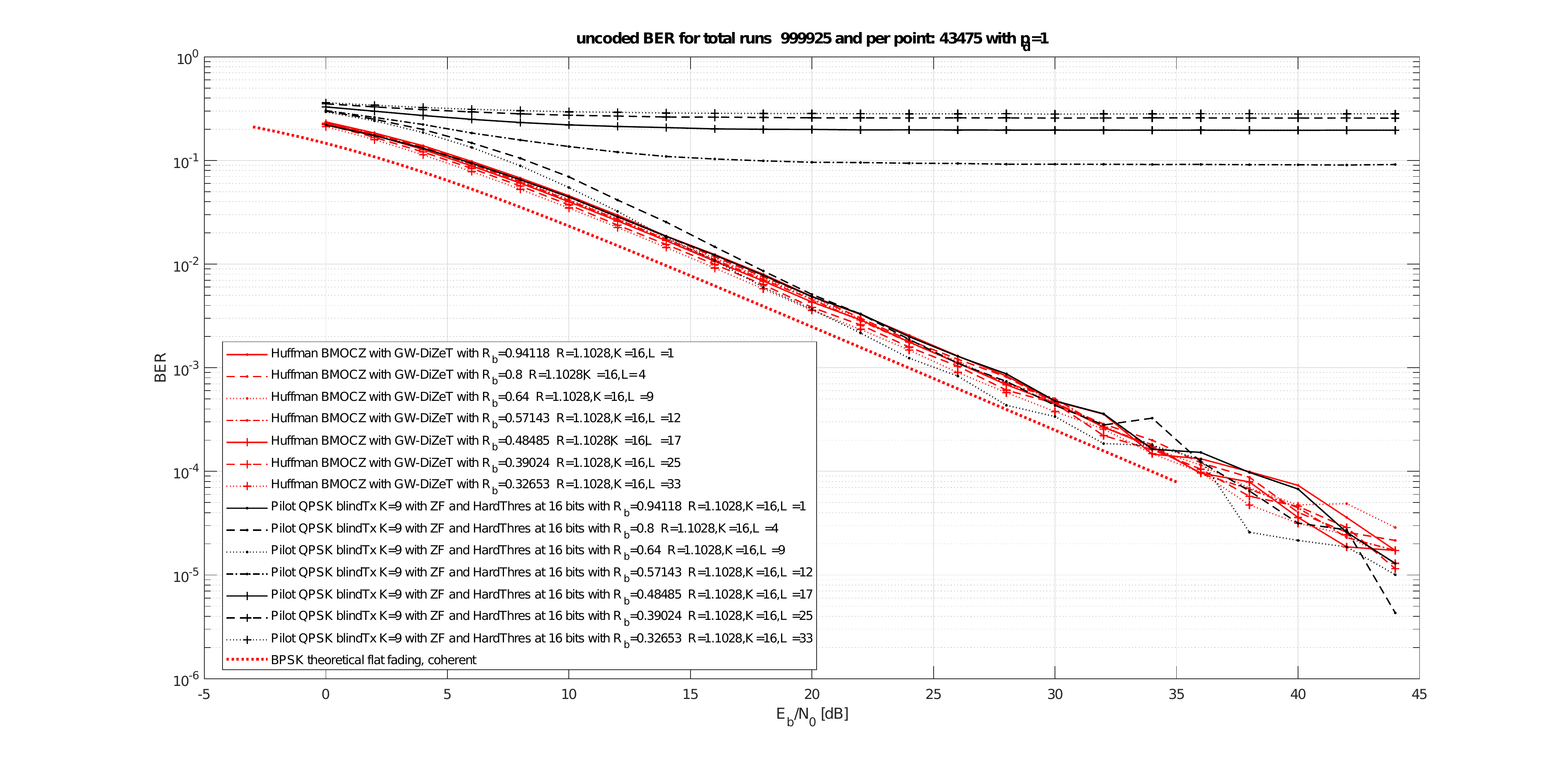}  
  \vspace{-1cm}
  \caption{Pilot based QPSK in time domain versus Huffman BMOCZ with $\pd=1$ for $\Nx=16$ and $\Nh=1,4,9,12,17,25,33$.}
  \label{fig:pilotQPSKvsBMOCZ}
\end{figure}

Interesting scenarios for the BMOCZ scheme are distributed wireless sensor networks which require
\begin{itemize}
  \item Low-Power (transmitter and receiver)
  \item Low-Latency
  \item No feedback and no channel information at transmitter
  \item Short block-length, $\Nx\simeq\Nh$
  \item One-shot communication, channel is used only once, sporadic in time
\end{itemize}
If $\Nx<\Nh$ there is no way to learn the channel with pilots in one shot. Moreover, the low-power assumption is not
suitable to use energy detector if we need to transmit more than $1$ bit. Higher order MOZ modulation might be also
considered.  This constraints, rule out
\begin{itemize}
  \item CDMA: usually requires $\Nx\gg \Nh$, \cite[p.92]{TV05} 
  \item Clasical ODFM: needs channel for decoding. Hence,
    we could assume again $\vx=\vu+\vd$ And in $\hat{\vd}$ using QAM. for example QPSK.
\end{itemize}
The BMOCZ scheme does not need channel length knowledge at the transmitter!  
On the other hand, any pilot data needs assumption
on the channel length. If the pilots are to short, it is impossible to estimate exactly the channel, even in the
noiseless case. We will investigate a scenario, where we blind transmit with $P=\lfloor \Nh/2\rfloor$ pilots and $D=P$
data and receive only $\Nh+\Nh$ taps, regardless of the true channel length $\Nh$. Hence, either we take to much sample
or to less at the receiver. This will affect the performance of pilot-based schemes, which we simulated for QPSK and
OFDM, see \figref{fig:pilotQPSKvsBMOCZ}.  Here, the BER performance suffers dramatically if the channel length at the
transmitter is underestimated, rendering a reliable communication impossible. 
However, the BMOCZ decoder depends heavily of the maximal channel length $L$. It can be seen in the simulation, that an
overestimating of $L$ (fast decreasing power profile) is not affecting the BER performance much, see
\figref{fig:MLandRFMD_snr} and \figref{fig:MLandRFMD_ebno}.

The GW-DiZeT decoder demands no complexity at all and allows with the DFT an easy and probably analog realization (using
delayed amplified circuits). The complexity at the transmitter consist of a fixed
    codebook of size $2^\Nx$ in $\C^{\Nx+1}$ dimensions.  


%% file: stabilitybound.tex
\section{Sharper Robustness Analysis for BMOCZ Codebooks by Exploiting the Geometric Zero Structure}\label{sec:sharperbounds}

We will in this section investigate  the geometric structure of the zeros to improve the robustness of normalized
polynomials against additive noise on its coefficients, sometimes also referred to the \emph{conditioning of a
polynomial}.  This is not to mistaken with the notion of stable polynomials or \emph{Hurwitz stability}, which refers to
the property that all zeros are located in the positive half-plane, see for example \cite[Cha.21]{Fis08}.
As we saw in the analysis of \thmref{thm:zerodistortion}, a large pairwise distance as well as a large leading
coefficient guarantee a robustness against additive noise.   For polynomials generated by autocorrelations, we will have
zeros in conjugate-reciprocal  pairs and if we upper bound the largest zero, we force the zeros in a ring or annulus
around the unit circle, which will exclude the extreme cases in the zero displacement, see \figref{fig:hexagonal}. It
turns out that the Euclidean metric, as used in the RFMD decoder might be reasonable for zeros near the unit circle.
Indeed, for Huffman Polynomials with uniform radius, the root neighborhoods can be bounded by disjoint uniform discs,
see \figref{fig:zeroperturbationsim}.  However, the first zero on the real line, seem to disturb non uniform.  This
might be due to discontinuity of the real valued zero on the positive half plane (winding number).  A more careful
analysis of the exact root neighborhood grow behaviour will be investigated in a follow up paper.  Since we want to keep
the root neighborhoods disjoint, the root neighborhoods should not exceed a radius $\del$ which is larger then half the
minimal pairwise distance. This in fact, leads to a circle packing problem in the plane, which is know to be most dense
if the circle centroids are placed on a hexagonal lattice, see \figref{fig:hexagonalring}.

\newcommand{\Rinv}{\ensuremath{R^{-1}}} \newcommand{\Rnorm}{\ensuremath{R}}
\begin{figure}[t] 
  \begin{subfigure}[b]{0.485\textwidth} \centering \def\svgwidth{0.9\textwidth} \footnotesize{
      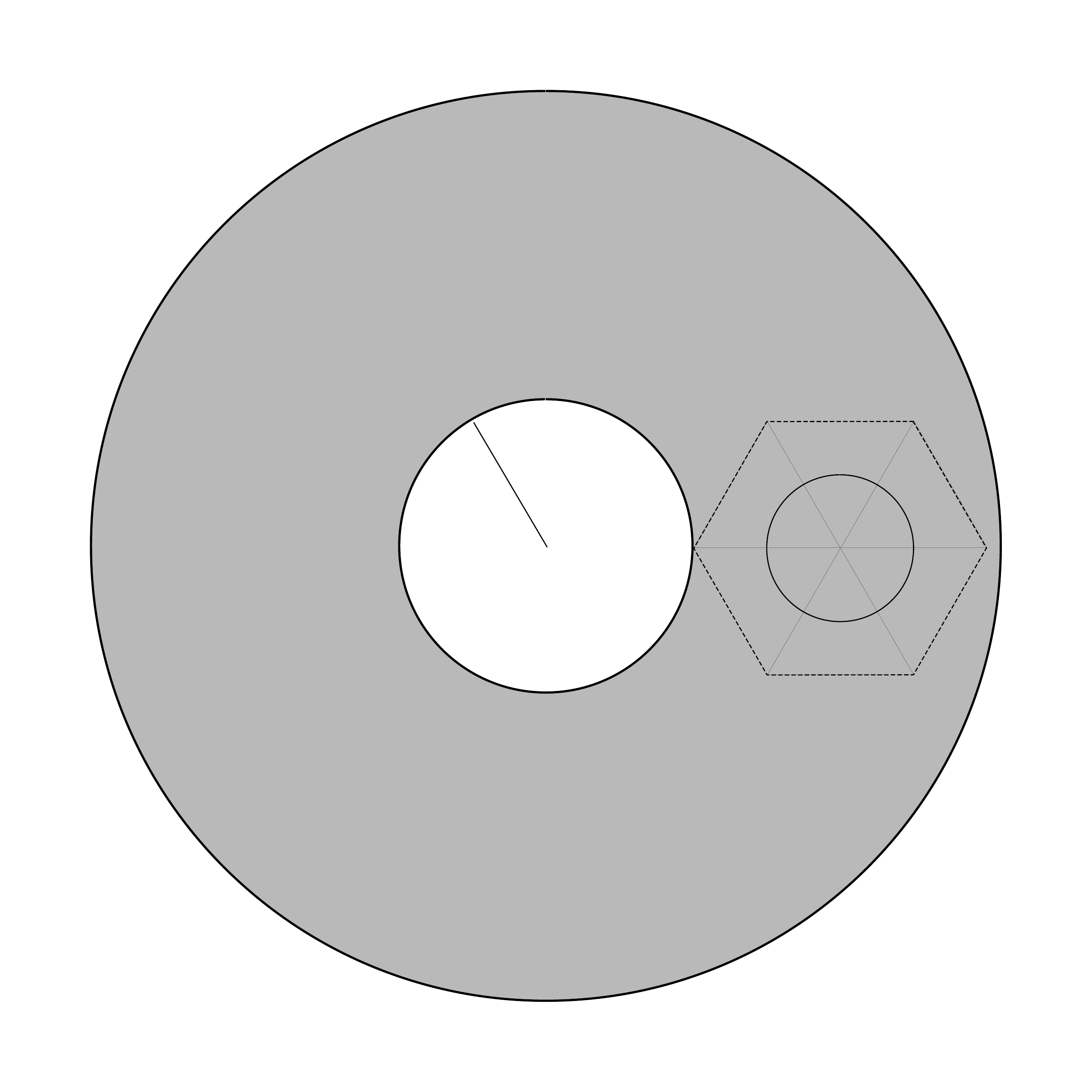 }
      \caption{Hexagonal packing in an annulus, worst constellation red, best constellation blue.}
      \label{fig:hexagonalring}
    \end{subfigure} 
    \begin{subfigure}[b]{0.485\textwidth} 
  \hspace{-1.2cm}
    \def\svgwidth{1.53\textwidth} \footnotesize{
      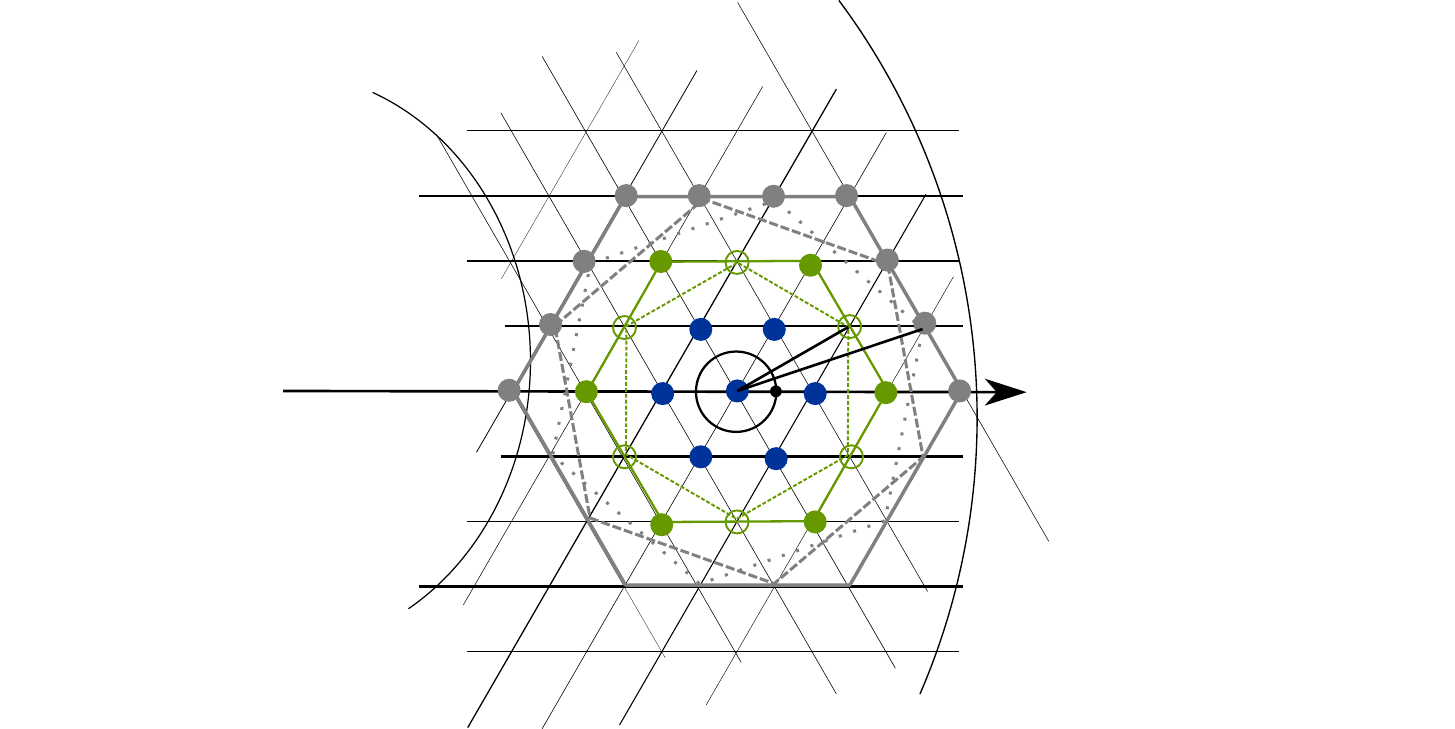 } 
      \caption{Densest packing inside the ring, worst case.\\[5ex]} \label{fig:hexagonalinring}
    \end{subfigure}     
    \caption{Homogeneous Circle Packing in an annulus.} \label{fig:hexagonal} 
\end{figure}

  %
  The idea is to place $N$ zeros in a Ring $\Ring=\Ring(R)$ of area $|\Ring|$ with minimal pairwise distance $\dmin$.
  The bound in \eqref{eq:brutallowerbound} is actually a lower bound of the geometric mean of all zeros distances. Hence, the
  robustness bound depends not only on the minimal pairwise distance but also on their geometric structure.
  In fact, the densest packing of the $N$ zeros will yield to the smallest geometric mean of the distances and therefore 
  result in the worst stability bound, see \figref{fig:hexagonal}. 
  Furthermore, the maximal amount of  zeros in the ring is bounded by the spherical circle packing problem. The exact amount is
  unsolved for arbitrary $\dmin$ even if the set is the unit disk. The problem is usually known as the density
  packing problem, where the density of placing $N$ equal circles of radius $\dmin$ in the ring $\Ring$ is given by
  \begin{align}
    D=  N\dmin^2 \pi 4|\Ring|
  \end{align}
  One bound on the maximal number $N$ of circles is given by Fejer Toth as
  \begin{align}
    N\leq \frac{\pi (R^2-R^{-2})}{\dmin^2\sqrt{12}}\label{eq:Nboundring}
  \end{align}
  see for example \cite{GT00}.
\subsection{Revised Proof of \thmref{thm:zerodistortion}}

The main idea of the proof relies on the use of Rouches \thmref{thm:rouche}, by controlling 
for each $m=1,2,\dots,N$ the modulus of the noise polynomial on the single root neighborhood circle bounds
  \begin{align} |\uW(z)|\leq |\uX(z)| \quad,\quad z\in C_m(\del).
  \end{align}
  By choosing the radius $\del$ small enough,
  such that no overlap of the single root neighborhoods occur, a separation by the RFMD decoder will be always
  successful (no channel zeros), guaranteeing  an error free decoding.
 %
  %
  Since we want to hold this for every noise polynomial generated by $\Norm{\vw}_2\leq \eps$ we have to satisfy
  \begin{align} 
    \max_{\Norm{\vw}\leq \eps} \max_{z\in C_m(\del)} \frac{|\uW(z)|^2}{|\uX(z)|^2} \leq 1.\label{eq:wmax}
  \end{align}
  Note, that $\uX(z)$ has no zeros on $\bigcup C_m(\del)$, hence we can divide.  Let us define
  $\underline{\vz}=(z^0,z^1,\dots,z^{\xord})^T$, where we set $0^0=1$. We will upper bound the magnitude of $\uw$ by
  using Cauchy-Schwarz 
  \begin{align}
    |\uw(z)|^2 = |\vw^T\underline{\vz}|^2\leq \Norm{\vw}_{2}^2\cdot \Norm{\underline{\vz}}_2^2 = \eps^2\cdot
    \sum_{n=0}^N |z|^{2n}.\label{eq:csbound} 
  \end{align}
  Since the noise $\vw$ is in the ball with radius $\eps$, all directions can be chosen and we achieve always equality
  in \eqref{eq:csbound}. Hence, \eqref{eq:wmax} is equivalent to
  \begin{align} \eps^2 \leq \frac{1}{\max_{z\in C_m(\del)}  \frac{\sum_n |z|^{2n}}{|\uX(z)|^2}}=\min_{z\in C_m(\del)}
    \frac{|\uX(z)|^2}{\sum_n |z|^{2n}} = \min_{z} f_m(z)
    \label{eq:worstcaseepsilon}
  \end{align}
  By using $z=\alp_m+ \del e^{i\tht}$ for some $\tht\in[0,2\pi)$ we need to find a tight lower bound of 
  \begin{align} f_m(\tht):=\frac{|\uX(\alp_m + \del e^{i\tht})|^2}{\sum_n |\alp_m+\del e^{i\tht}|^{2n}} 
   =|x_N|^2  \frac{\Pro_{n} |\alp_m + \del e^{i\tht} -\alp_n|^2}{\sum_n |\alp_m +\del e^{i\tht}|^{2n}}.
  \end{align}
  Since  we are searching for a uniform radius $\del$ which keeps all root neighborhoods disjoint, we search for the worst
  $\alp_m$. The only information of the zeros we have is the minimal pairwise distance $\dmin$ and the smallest and
  largest moduli $R^{-1}$ resp. $R$, we define therefore
  \begin{align}
    \Aset:=\set{\alp=\{\alp_1,\dots,\alp_N\}}{ \forall n: R\geq |\alp_n| \geq R^{-1}, \dmin\leq\min_{m\not=n}|
  \alp_m-\alp_n|}.
  \end{align}
  Note, $\Aset$ is a compact set.  The leading coefficient $x_N$ depends on all zeros and \emph{the height} of the
  polynomial, given by
  \begin{align}
    \Norm{\uX}_2:=\Norm{\vx}_2=\sqrt{\sum_{n=0}^{N} |x_n|^2}.
  \end{align}
  We chose here the Euclidean norm as the height, since we are interested in SNR performance.  Hence, we define the set
  of all allowable normalized polynomials  with zeros in $\Aset$  as
  \begin{align}
    \Pset=\Pset(\Aset) :=\set{\uX(z)=\sum_{n=0}^N x_n z^n=x_N \Pro_n (z-\alp_m)}{ \Norm{\vx}_2=1,\alp\in\Aset}\notag.
  \end{align}
  This brings us to the following optimization problem
  \begin{align}
    f(\Pset)=\min_{\uX\in\Pset} |x_N|^2\min_m \min_{\tht}  \frac{\Pro_{n} |\alp_m + \del e^{i\tht} -\alp_n|^2}{\sum_n |\alp_m +\del
      e^{i\tht}|^{2n}}.\label{eq:optimizationPsettht}
  \end{align}
The modulus of the leading coefficient can be lower bounded for normalized polynomials by
\begin{align}
2^{-N}\leq 2^{-N}\Norm{\vx}_1 \leq |x_N|\Pro_{n=1}^N\max\{1,|\alp_n|\},
\end{align}
see for example \cite{Mah60} or \cite[Prop.86]{Zip93}. Note, this bound is not very tight for simple zeros with large minimal pairwise
distance. 
If all zeros are inside the unit circle, this results in the largest lower bound and if all zeros are on the outside radius
this results in the lowest bound, i.e.,
\begin{align}
  2^{-N} \leq |x_N| \quad,\quad {(2R)}^{-N}\leq |x_N|.\label{eq:xnlowerbounds}
\end{align}
Using the worst case bound allows to eliminate the height constraint
\begin{align}
  f(\Pset)\geq (2R)^{-N} \min_{\alp\in\Aset}\min_m \min_{\tht} 
  \frac{\Pro_{n} |\alp_m + \del e^{i\tht}-\alp_n|}{\sqrt{\sum_n |\alp_m +\del e^{i\tht}|^{2n}}}
  \label{eq:optimizationalpm}
\end{align}
which is necessary to leverage the problem to a pure geometric problem.
%
Let us assume $\alp_{\hat{m}}=\rho e^{i\phi}$ is the zero selection for which there exists a $\tht$ which obtains the
minimum. Then, we can rotate all zeros by $e^{-i \phi}$, since it will not change their modulus nor their pairwise
distances and hence  be lying in $\Aset$ (rotation invariant). Since the numerating of $\alp_n$ is arbitary, we can just
chose $\alp_{\hat{m}}=\alp_1=\rho$. Hence, we can omit the minimization over $m$ since we minimize over all $N-1$ zeros
and $\alp_1\in[R^{-1},R]$. This brings us to the non-convex geometric problem 
\begin{align}
  \min_{\alp\in\Aset}  f(\alp)
  =\min_{\alp, \alp_1\in[R^{-1},R]} \min_{\tht}
  \frac{\Pro_{n} |\alp_1+\del e^{i\tht}  -\alp_n|}{\sqrt{\sum_n |\alp_1+\del e^{i\tht}|^{2n}}}.
  \label{eq:optimizationPset}
\end{align}
The nominator is independent of the other $N-1$ zeros, and obtains its maximum for $|R+\del|$.  It can be seen that the
numerator, will not yield the global minimal constellation if we place the $N-1$ zeros around $\alp_1=R$, for $N\geq 2$,
due to the restriction of the ring.  However, it is geometrical not obvious which $\alp_1$ will yield the global minimum
of $f$.  Therefore, we lower bound the nominator in $f(\alp)$, independently of the numerator, by the
geometric formula for the worst case
\begin{align}
  f(\alp)\geq \sqrt{\frac{(R+\del)^{2}-1}{(R+\del)^{2N}-1}} g(\alp)\quad\text{ with }\quad g(\alp)=\min_{\tht}
  g(\alp,\tht)=\Pro_{n=1}^N |\alp_1 +\del e^{\im
  \tht}-\alp_n|.\label{eq:lbound}
\end{align}
Now, the minimization over the zeros reduces to the densest packing of $\alp_2, \dots, \alp_N$ around $\alp_1$, since
the \emph{geometric mean} $g(\alp)^{1/n}$ of the distances decreases if each distance decreases.  If $N= 7$, the densest
packing is the hexagon inscribed in a radius $\dmin$ with one zero at its centroid, see \figref{fig:hexagonalinring}.
For arbitrary $N$ this is the well known circle packing problem, also known as the ``penny packing`` problem.  However, it is not
obvious if the optimal $z$ lies on the circle around the centroid or on the circle around the vertices.
If $\alp_1$ is the centroid, then we can show by
\thmref{thm:mgonproddistance} that extremal $z$'s will lie at crossing of the circle with the line between origin and one
vertex. Since we want to maximize the nominator, the $z$ which lies on the real axis, right from the centroid, will
therefore obtain the minimal value of $f$.  If $\alp_1$ is one of the vertices, then we need to show that $z$ does not
achieve a smaller product distance.  Unfortunately, we can not prove this analytically and formulated this as 
\conref{con:ngon} in \appref{app:ngon}.

If the conjecture holds, then the idea is to consider nested polygons (honeycombs) as the worst case zero configuration,
to derive a lower bound for \eqref{eq:lbound}.
Each $n$th  honeycomb consist then of $6n$ points, placed on $n$ hexagons rotated accordingly, see
\figref{fig:hexagonalinring}.
By \thmref{thm:mgonproddistance}, the optimal point $z=\alp_1 +\del e^{\im \tht}$ for the inner hexagon is achieved for  
$\tht=0$. This gives us a lower bound for the minimal product
distance for the $1$st hexagon inscribed in a circle with radius $r_1=\dmin$ as
\begin{align}
  \Pro_{k=1}^6 |\del -\alp_{n_k^{(1)}}|^2 \geq |r_1^6-\del^6|^2 =r_1^{2\cdot 6}
  \cdot {\underbrace{|1-(\frac{\del}{\dmin})^{6}|}_{=c^6(\del,\dmin)}}^2\geq
  r_1^{12} \cdot c^{12}
\end{align}
where we assume ${\del}/{\dmin}\leq 1/2$, since $\del<\dmin/2=r_1/2$.
The radius for the $n$th hexagon is $n\dmin$, where on the $n$th honeycomb the $6n$ zeros, are the vertices of $n$
rotated hexagons. The smallest radius $r_n$ is given hereby with the law of cosine as
\begin{align}
  r_{n} =\sqrt{ \dmin^2(1+n^2-2n\cos(\pi/3))} = \dmin\sqrt{1+n^2-n}\geq \dmin(n-1)
\end{align}
for $n\geq 2$, see \figref{fig:hexagonalinring}. If $n=1$ we set $r_1=\dmin$.
Then we get for the product of distance of the $n$th honeycomb for $n\geq 2$
\begin{align}
  p^{(n)}= \Pro_{k=1}^{6n} |\del-\alp_{n_k}|^2 =\Pro_{m=1}^n \Pro_{k=1}^6 |\del-\alp_{n_k^{(m)}}|^2 
  \geq \Pro_m (c\dmin(n-1))^{12} =(\cdmin (n-1))^{12n}
\end{align}
with $\tdmin=c \dmin$. If we have $\nmax$ nested honeycombs we have up to 
\begin{align}
  N=1+\sum_{n=1}^{\nmax} 6n =1+ 3\nmax(\nmax+1) 
\end{align}
zeros packed, which gives the lower bound
\begin{align}
  \nmax =\lceil\sqrt{\frac{4N-1}{12}}-\frac{1}{2}\rceil
\end{align}
of nested honeycombs, yielding to 
\begin{align}
 g(\alp)& \geq \del^2 \left(\cdmin \Pro_{n=2}^{\nmax} (\cdmin(n-1))^{n}\right)^{12} 
  =\del^2\left(\cdmin \Pro_{n=1}^{\nmax-1} \cdmin^n n^{n+1}\right)^{12}\\
  &= \del^2  \left(\cdmin\cdmin^{\frac{\nmax(\nmax-1)}{2}}\Pro_{n=1}^{\nmax-1} n^{n+1}\right)^{12}\notag \\
  &  \geq \del^2\cdmin^{\ 6(\nmax^2 -\nmax +2)} \cdot [(\nmax-1)!H(\nmax-1)]^{12},
\end{align}
where the last factor is the \emph{hyperfactorial} $H(\nmax-1)=\Pro_{n=1}^{\nmax-1} n^n$.
Note, that we had to add the distance square $|\alp_1+\del e^{i\tht}-\alp_1|^2=\del^2$ of the centroid zero $\alp_1$.
Combining \eqref{eq:optimizationalpm} ,\eqref{eq:optimizationPset} and \eqref{eq:lbound} would yield the final noise
bound.


\subsection{Noise Bounds for Huffman Polynomials}
Note, the noise energy bound \eqref{eq:worstcaseepsilon} is deterministic and not in mean.  First of all, for fixed $R$
and $\dmin$ the number $N$ of zeros we can place in the ring $\Ring$ is bounded by \eqref{eq:Nboundring}.  We plotted in
\figref{fig:worstcase_ring} the noise power bound for fixed $N$ over various $R$. Here, we set
$\dmin=\sqrt{\pi(R^2-R^{-2})/N\sqrt{12}}$ and assume that  all zeros are place in a circle of area $\pi(R^2-R^{-2})$
with center at $R$. This is the worst packing for $N$ zeros. However, it can be seen that the bound is not very sharp if
$N$ increases. We assume that Huffman sequences, placed on vertices of two $N-$gons are the best case. However, the
optimal radius in the sense of maximal noise robustness for a fixed $N$ is still unknown. Nevertheless, the simulation
and analysis suggest that a radius close to the optimal might be given if the uniform circle neighborhoods touching each
other, see \figref{fig:optimalradius} and \figref{fig:AnalyticlowerboundHuffconj} for a simulation of $N=6$. In fact,
the root-neighborhoods are directed, depending on the particular choice of the other zeros. Hence, in average, the
root-neighborhoods will more likely be bounded by an ellipse. Also the outside root-neighborhoods have larger radii than
the insides, which suggest also a heterogeneous neighborhoods.
\begin{figure}[h]
  \centering
      \includegraphics[width=0.8\textwidth]{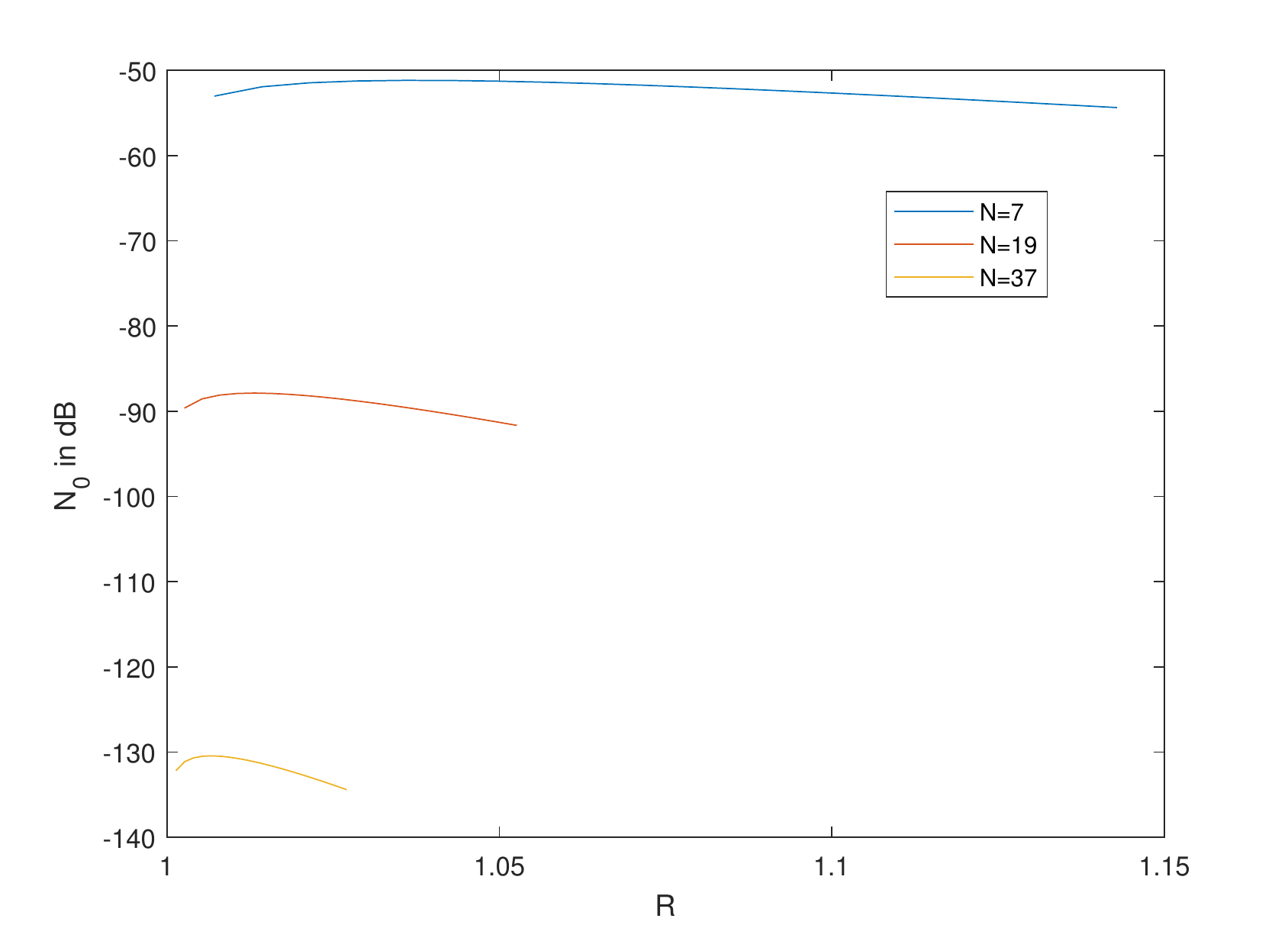}
      \caption{Worst case noise bounds for various $N$ and densest packing in circles of radius $R$.}
     \label{fig:worstcase_ring}
\end{figure}

\begin{thm}[Noise Bound for Huffman BMOCZ]
   Let $N=4M$ for some $M\geq 3$ and $\Code(N,R)$ be the set of normalized Huffman sequences in $\C^{N+1}$ with radius $R>1$. 
   Then the minimal pairwise distance is given by
   \begin{align}
     \dmin=2R^{-1} \sin(\pi/N)
   \end{align}
   and the maximal noise power which guarantees  root neighborhoods in circles of radius $\del\in[0,\dmin/2)$ is given by
   \begin{align}
     N_0 &\leq 
\frac{1}{R^{8M}+1} \frac{(R+\del)^{2}-1}{(R +\del)^{8M}-1}
\cdot  R^{2-4M} \del^4(2R^{-1}\sin(\pi/N)-\del)^4 \Pro_{m=3}^M m^4 \cdot\\
 &  \quad \quad
 \left(\frac{\sin(2\pi/N)-\sin(4\pi/N)-2\sin(\pi/N)}{2(1-\sin(2\pi/N))}\right)^{4M-12}\label{eq:AnalyticlowerboundHuffCor}
   \end{align}
   The worst case Huffman sequence is given by all zero inside the unit circle except one. 
\end{thm}
\begin{remark}
  If \conref{con:ngon} holds then we would get the noise power bound
  \begin{align}
    N_0 \leq \frac{R^{-3N}}{R^{-2N}+1} \cdot \frac{(R+R^{-1} \sin(\pi/N))^2-1}{(R+R^{-1}\sin(\pi/N))^{2N}-1}\cdot 
    (1-(1-\sin(\pi/N))^N)^2
  \end{align}
  which yields to the bounds over $N$ in \figref{fig:AnalyticlowerboundHuffconj} for
  $\del=\del_{max}=R^{-1}\sin(\pi/N)$. The red line shows the bound for a
  radius where all root neighborhoods remain disjoint.
\end{remark}

\begin{figure}[t] 
  \begin{subfigure}[b]{0.485\textwidth} 
\def\svgwidth{1.07\textwidth} \tiny{
\hspace{-0.5cm}
  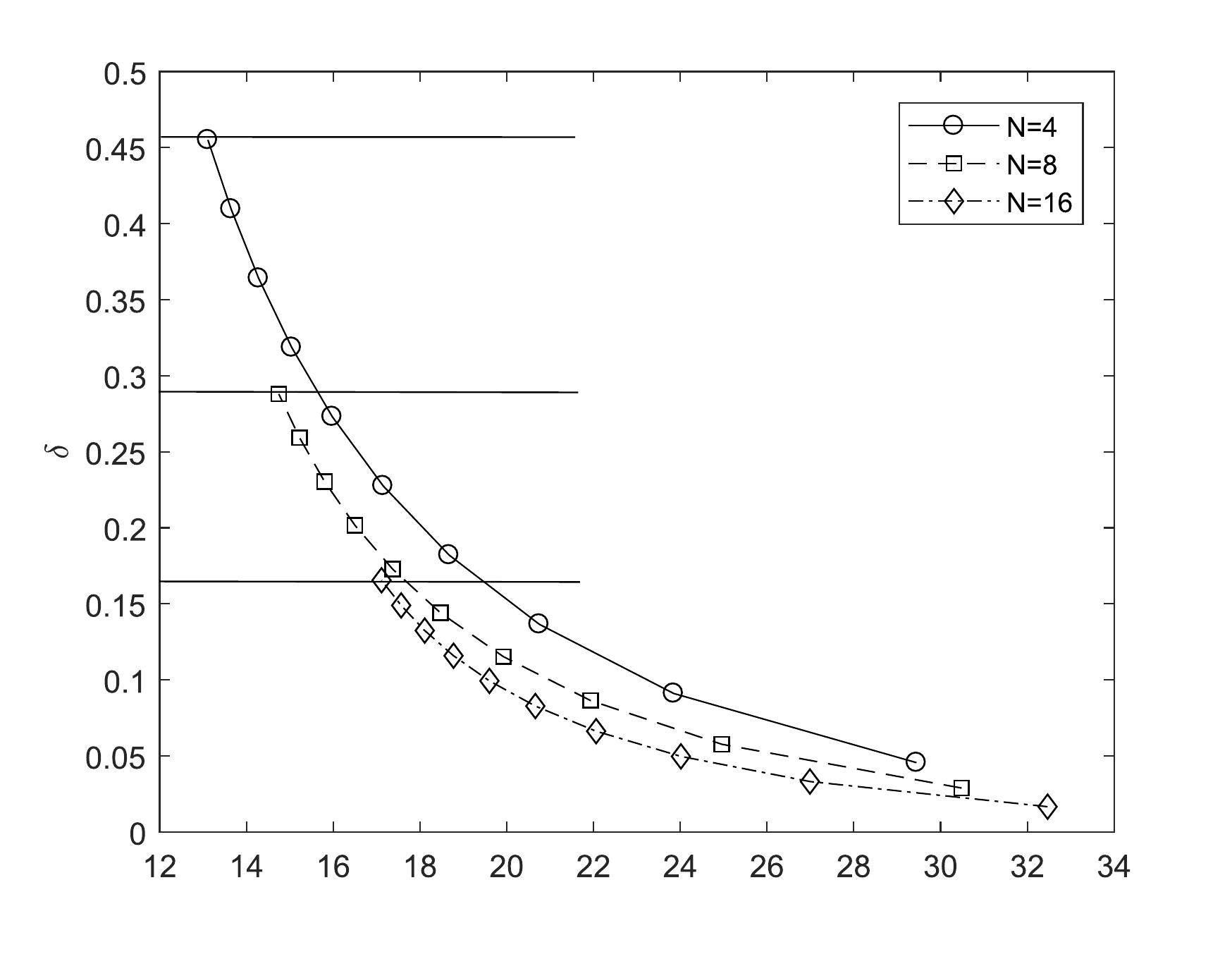}
\caption{Root neighborhood radius $\delta$ over SNR for worst Huffman sequence  
  with $\Ropt=1.5538,1.3287,1.1791$ from \eqref{eq:worstcaseepsilon} by quantizing $\tht$  with  $1000$ points.}
\label{fig:deltaepsilon}
\end{subfigure} 
\hspace{0.2cm}
\begin{subfigure}[b]{0.485\textwidth} 
  \includegraphics[width=1.25\textwidth]{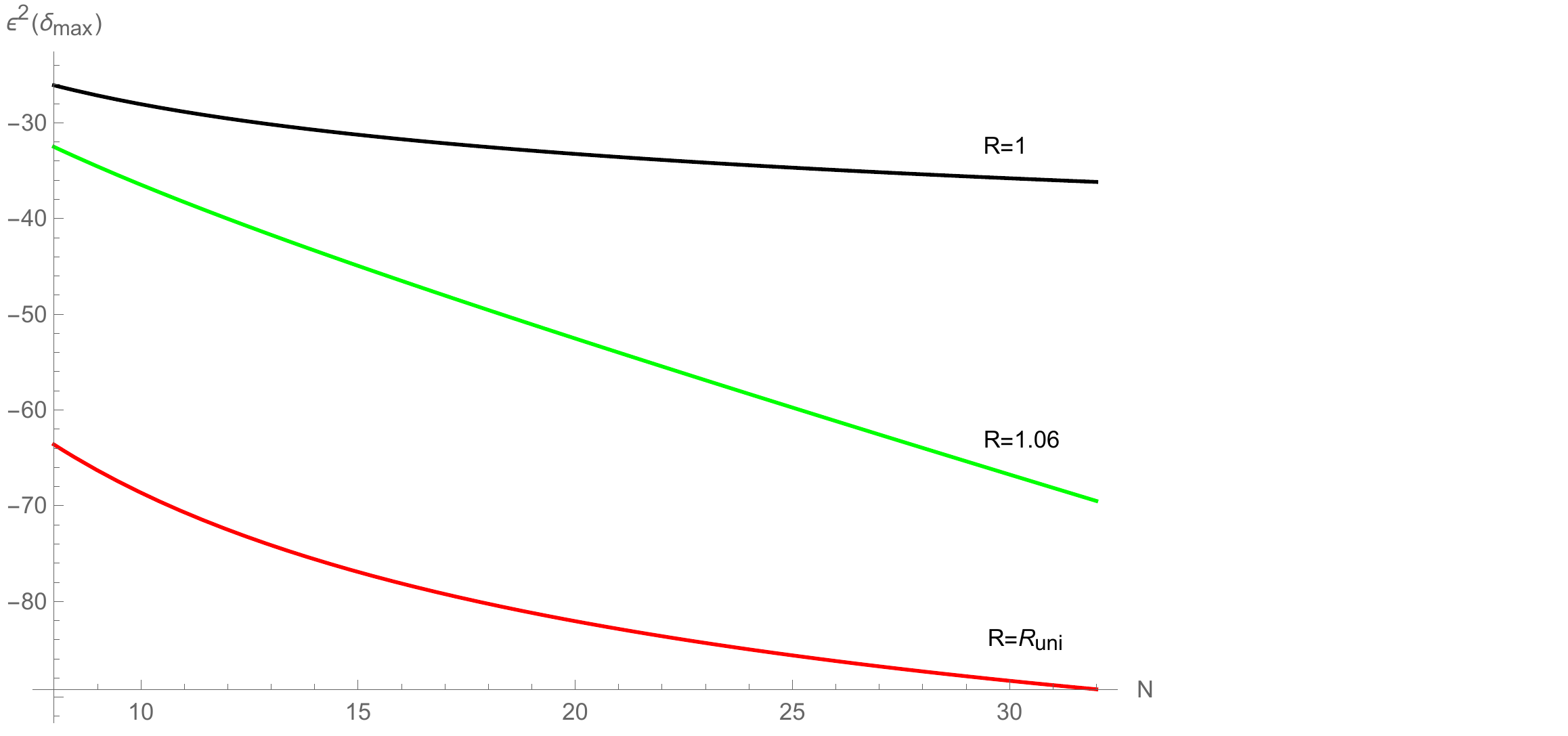}
  \caption{Analytic Noise power bounds in dB with \conref{con:ngon} over various $N$ for different radii $R$.\\}
      \label{fig:AnalyticlowerboundHuffconj}
    \end{subfigure}
    \caption{Noise power bounds  depending on radius $R$ and order $N$.}
\end{figure} 

\begin{proof}
  Since the zeros are lying on two regular $N-$gons, and the nominator only needs one outer zero to be maximize, we can
  assume that the worst case scenario is given by \figref{fig:huffmanworst}.
  From \eqref{eq:optimizationPsettht} we obtain for the Huffman Zero Codebook
  \begin{align}
    f^2(\Zero)
    &=\min_{\alp\in\Zero} |x_N(\alp)|^2\min_{m\in[N]\setminus n} \min_{\tht}  \frac{\Pro_{n} |\alp_m 
     + \del e^{i\tht} -\alp_m|^2}{\sum_n |\alp_m +\del  e^{i\tht}|^{2n}}\\
    &=\min_{\alp\in\Zero} |x_N(\alp)|^2 \min_{\tht}  \frac{\Pro_{n} |\alp_1 
      + \del e^{i\tht} -\alp_m|^2}{\sum_n |\alp_1 +\del  e^{i\tht}|^{2n}}.
  \end{align}
  Note, that $|x_N(\alp)|$ is minimized by \eqref{eq:hufflastfirst}  if all zeros lie outside the unit circle, i.e.,
  \begin{align}
    \min|x_N(\phi)|^2\geq \frac{ R^{-2K}}{1+R^{-2K}} = \frac{1}{R^{2K}+1}.
  \end{align}
  Hence, by choosing $\tht=\pi$ and all zeros inside the unit disc except one, we get with 
  \eqref{eq:huffmanMgonbound2} from \lemref{eq:gHuffman}
  \begin{align}
    f^2(\Zero)&\geq
    \frac{1}{R^{8M}+1} \frac{(R+\del)^{2}-1}{(R +\del)^{8M}-1}
    \cdot  r^{4M-2} \del^4(2a-\del)^4 \Pro_{m=3}^M m^4 \cdot\\
     &  \quad \quad
     \left(\frac{\sin(2\pi/N)-\sin(4\pi/N)-2\sin(\pi/N)}{2(1-\sin(2\pi/N))}\right)^{4M-12}\label{eq:AnalyticlowerboundHuff}
  \end{align}
\end{proof}

  In \figref{fig:deltaepsilon} we plotted the exact bound \eqref{eq:optimizationPsettht} for Huffman Codebooks with optimal radius
  \eqref{eq:optimalradiusexactrho} for $N=4,8,16$ starting with the maximal radius $\del_{max}=\dmin/2$. It can be seen
  that for increasing $N$ the
  $\del_{max}$ shifts to higher SNR, which indicates less noise robustness.

\begin{figure}[t] 
  \centering
  \includegraphics[width=0.7\textwidth]{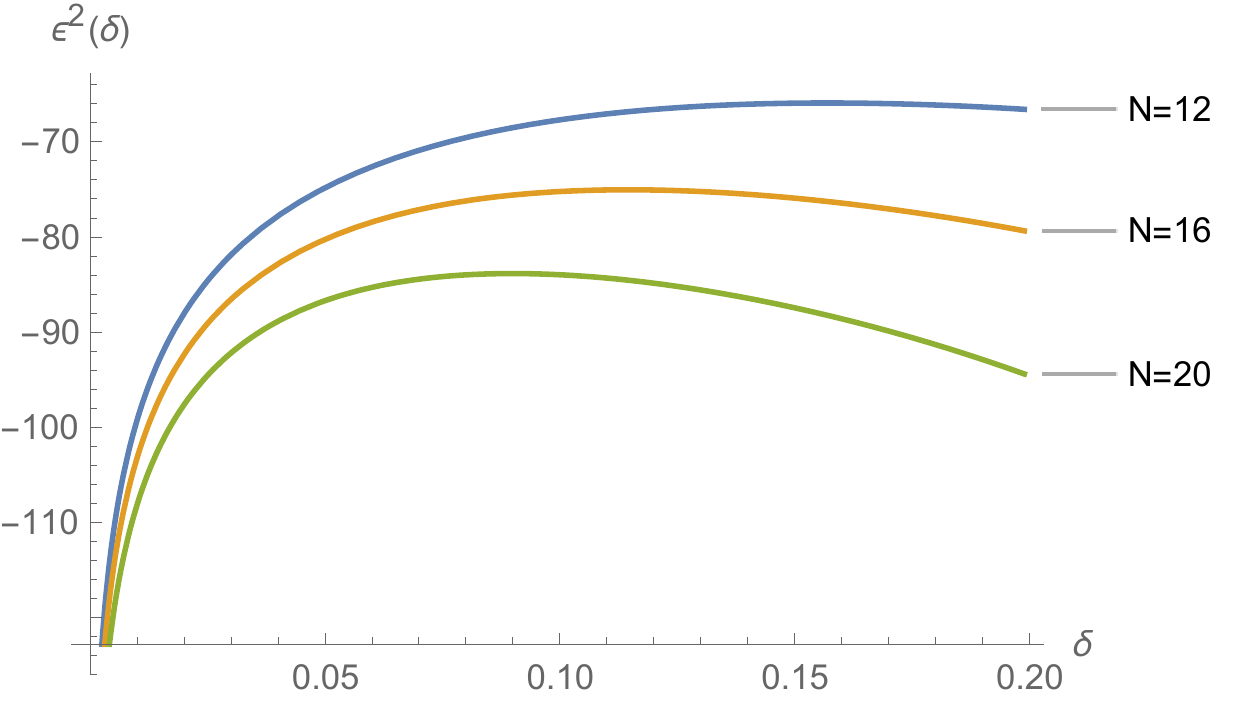}
  \caption{Analytic Noise power bounds \eqref{eq:AnalyticlowerboundHuff} in dB for Huffman sequences with $R=1.02$ and various $N$.}
      \label{fig:AnalyticlowerboundHuff}
\end{figure} 



%% file: HexagonalCirclePack_rotated.pdf_tex
\begingroup%
  \makeatletter%
  \providecommand\color[2][]{%
    \errmessage{(Inkscape) Color is used for the text in Inkscape, but the package 'color.sty' is not loaded}%
    \renewcommand\color[2][]{}%
  }%
  \providecommand\transparent[1]{%
    \errmessage{(Inkscape) Transparency is used (non-zero) for the text in Inkscape, but the package 'transparent.sty' is not loaded}%
    \renewcommand\transparent[1]{}%
  }%
  \providecommand\rotatebox[2]{#2}%
  \ifx\svgwidth\undefined%
    \setlength{\unitlength}{742.57204693bp}%
    \ifx\svgscale\undefined%
      \relax%
    \else%
      \setlength{\unitlength}{\unitlength * \real{\svgscale}}%
    \fi%
  \else%
    \setlength{\unitlength}{\svgwidth}%
  \fi%
  \global\let\svgwidth\undefined%
  \global\let\svgscale\undefined%
  \makeatother%
  \begin{picture}(1,0.99999987)%
    \put(0,0){\includegraphics[width=\unitlength,page=1]{HexagonalCirclePack_rotated.pdf}}%
    \put(0.44854847,0.59623987){\color[rgb]{0,0,0}\makebox(0,0)[lb]{\smash{\Rinv}}}%
    \put(0,0){\includegraphics[width=\unitlength,page=2]{HexagonalCirclePack_rotated.pdf}}%
    \put(0.52844967,0.50250901){\color[rgb]{0,0,0}\makebox(0,0)[lb]{\smash{\one}}}%
    \put(0.50157079,0.9277617){\color[rgb]{0,0,0}\makebox(0,0)[lb]{\smash{\Rnorm}}}%
    \put(0,0){\includegraphics[width=\unitlength,page=3]{HexagonalCirclePack_rotated.pdf}}%
  \end{picture}%
\endgroup%

%% file: hexagonallattice.pdf_tex
\begingroup%
  \makeatletter%
  \providecommand\color[2][]{%
    \errmessage{(Inkscape) Color is used for the text in Inkscape, but the package 'color.sty' is not loaded}%
    \renewcommand\color[2][]{}%
  }%
  \providecommand\transparent[1]{%
    \errmessage{(Inkscape) Transparency is used (non-zero) for the text in Inkscape, but the package 'transparent.sty' is not loaded}%
    \renewcommand\transparent[1]{}%
  }%
  \providecommand\rotatebox[2]{#2}%
  \ifx\svgwidth\undefined%
    \setlength{\unitlength}{419.00208128bp}%
    \ifx\svgscale\undefined%
      \relax%
    \else%
      \setlength{\unitlength}{\unitlength * \real{\svgscale}}%
    \fi%
  \else%
    \setlength{\unitlength}{\svgwidth}%
  \fi%
  \global\let\svgwidth\undefined%
  \global\let\svgscale\undefined%
  \makeatother%
  \begin{picture}(1,0.50100265)%
    \put(0,0){\includegraphics[width=\unitlength,page=1]{hexagonallattice.pdf}}%
    \put(0.53440403,0.20924197){\color[rgb]{0,0,0}\makebox(0,0)[lb]{\smash{\rone}}}%
    \put(0.54163653,0.26665043){\color[rgb]{0,0,0}\makebox(0,0)[lb]{\smash{\rtwo}}}%
    \put(0.60853779,0.24992512){\color[rgb]{0,0,0}\makebox(0,0)[lb]{\smash{\rthree}}}%
  \end{picture}%
\endgroup%

%% file: HuffCodebook_deltaepsilon_dmin2.pdf_tex
\begingroup%
  \makeatletter%
  \providecommand\color[2][]{%
    \errmessage{(Inkscape) Color is used for the text in Inkscape, but the package 'color.sty' is not loaded}%
    \renewcommand\color[2][]{}%
  }%
  \providecommand\transparent[1]{%
    \errmessage{(Inkscape) Transparency is used (non-zero) for the text in Inkscape, but the package 'transparent.sty' is not loaded}%
    \renewcommand\transparent[1]{}%
  }%
  \providecommand\rotatebox[2]{#2}%
  \ifx\svgwidth\undefined%
    \setlength{\unitlength}{503.25bp}%
    \ifx\svgscale\undefined%
      \relax%
    \else%
      \setlength{\unitlength}{\unitlength * \real{\svgscale}}%
    \fi%
  \else%
    \setlength{\unitlength}{\svgwidth}%
  \fi%
  \global\let\svgwidth\undefined%
  \global\let\svgscale\undefined%
  \makeatother%
  \begin{picture}(1,0.77794337)%
    \put(0,0){\includegraphics[width=\unitlength,page=1]{HuffCodebook_deltaepsilon_dmin2.pdf}}%
    \put(0.40207782,0.68061863){\color[rgb]{0,0,0}\makebox(0,0)[lb]{\smash{\dminfour}}}%
    \put(0.40207782,0.47550068){\color[rgb]{0,0,0}\makebox(0,0)[lb]{\smash{\dmineight}}}%
    \put(0.40207782,0.32515856){\color[rgb]{0,0,0}\makebox(0,0)[lb]{\smash{\dminsixteen}}}%
    \put(0.43587561,0.03263239){\color[rgb]{0,0,0}\makebox(0,0)[lb]{\smash{\snreps}}}%
  \end{picture}%
\endgroup%

%% file: mgon.tex
\section{Product distances from a circle point to the vertices of a regular polygon}\label{app:ngon}

We will adapt a result in \cite[Sec.5]{CK07b} to derive for any regular $N-$gon the extremal products of distances from a point
on a circle centered at the centroid to all its vertices, see \figref{fig:smallestproductdistance}. 
%
\begin{thm}\label{thm:mgonproddistance}
  Let $N\geq 2$. Consider the regular $N-$gon inscribed in a circle of radius $r>0$ centered at the origin.  The  product of the
  distances to any fixed point $z=\del e^{\im \tht}$ on a circle of radius $0<\del$ centered at the origin to the
  vertices is bounded by 
  \begin{align}
    |r^N+\del^N|\geq  g^c_N(z):=\Pro_{n=1}^N d_n(z) \geq |r^N-\del^N|.
  \end{align}
  The bounds are sharp and the extremal points are $\zmin=\del e^{\im n\pi/N}$ resp. $\zmax=\del e^{\im n2\pi/N}$, which lie on
  a  line between one vertex and the origin respectively on a line between the middle point of two neighbor vertices
  and the origin. If the origin is added as an $N+1$ point the minimal and maximal product distance is achieved for
the same $\zmin$ and $\zmax$ and given by
\begin{align}
  |r^N+\del^N|\del\geq \del g^c_N(z):= \Pro_{n=1}^{N+1} d_n(z) \geq\del |r^N-\del^N|.
\end{align}
\end{thm}
%
\newcommand{\zdel}{\ensuremath{z}}
\newcommand{\rtht}{\ensuremath{r(\phi)}}
\newcommand{\hoehe}{\ensuremath{h}}
\begin{figure}[t] 
  \begin{subfigure}{0.485\textwidth} 
    \hspace{-2cm}
  \def\svgwidth{1.1\textwidth} \scriptsize{
    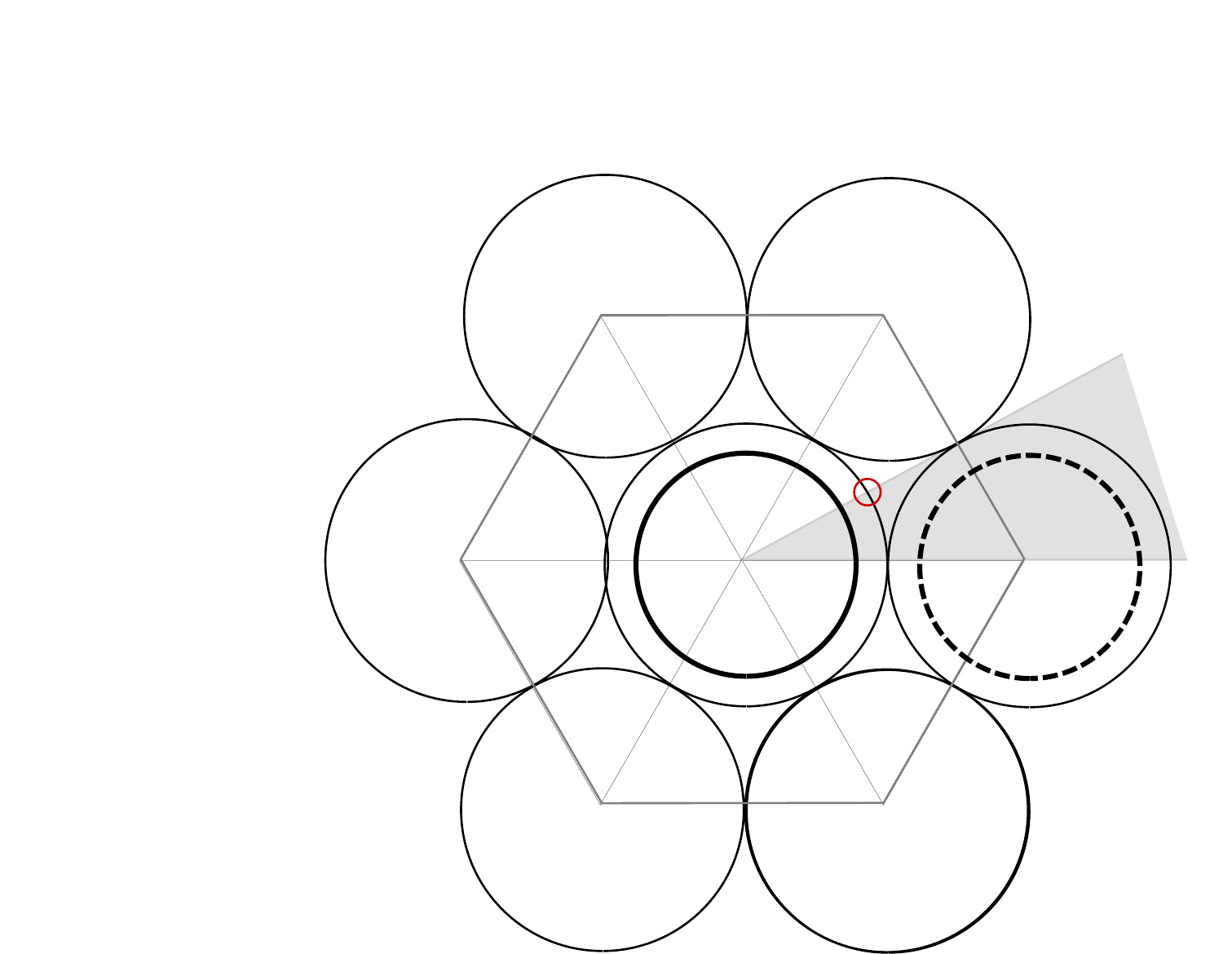}
\vspace{-0.128cm}
    \caption{Points inside the polygon on a circle of
    radius $\del$ (solid circle) resp. outside (dashed circle) with minimal (filled red) resp. maximal (non-filled
    ) distance products.}\label{fig:smallestproductdistance} 
  \end{subfigure} 
  \hspace{0.2cm}
  \begin{subfigure}{0.485\textwidth} 
    \vspace{-0.18cm}
    \centering 
  \def\svgwidth{0.85\textwidth} \scriptsize{
    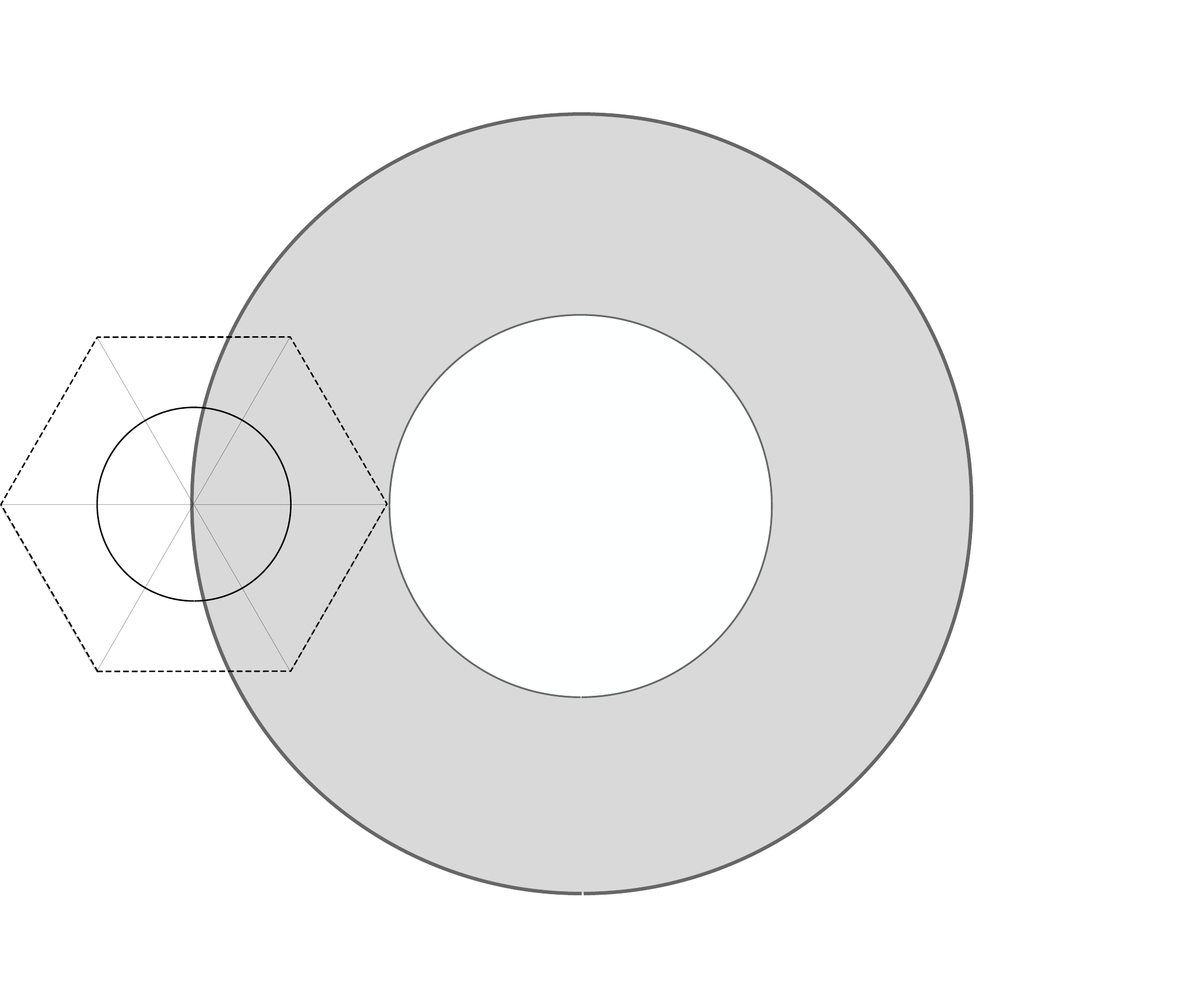}
\vspace{-0.18cm}
\caption{Worst case constellation for Huffman polynomials.}
\label{fig:huffmanworst}
 \end{subfigure} 
  \caption{Upper bounds for noise power which guarantee zero separation for $N=6$.}
\end{figure}

\begin{proof}
  Let us identify the vertices by complex numbers $r\ome^m$, where $\ome=e^{\im 2\pi /N}$ is the $N$th root of unity.
  Then the product of the distances from the vertices to any point $z$ is given by
  \begin{align}
    g^c_N(z)=\Pro_{m=1}^N d_m(z) &= |z-r|\cdot |z-r\ome^1|  \dots  |z-r\ome^{N-1}| = |z^N -r^N| 
  \end{align}
  Taking the product and inserting  $z=\del e^{i\tht}$ we get 
  \begin{align}
    g(\tht):=|\del^N e^{N\im \tht}-r^N|^2 = \del^{2N} -2{\del^N}r^N\cos(N\tht) +r^{2N}\label{eq:g}
  \end{align}
  We immediately find that $g$ is periodic in $[0,2\pi/N)$ and symmetric around $\tht_c=\pi/N$. Then, the critical
  points in the interval are given by the solutions of
  \begin{align}
    0=g'(\tht)=2N \del^N r^N \sin(N\tht) \quad \RA \quad \tht_1=0,\ \tht_2=\pi/N.
  \end{align}
  Due to the symmetry of $g$ around $\tht_1$ and $\tht_2$, one of them must be a maximum and the other a minimum.
  Inserting both in \eqref{eq:g} yields
  \begin{align}
    g(0)=\del^{2N} -2\del^Nr^N +r^{2N}=(r^N-\del^N)^2\quad,\quad g(\frac{\pi}{N})=(r^N+\del^N)^2
  \end{align}
  such that $g(0)$ is the minimum and $g(\pi/N)$ the maximum.  Due to symmetry of the hexagon we can assume that the
  extremal point is in the gray area of \figref{fig:smallestproductdistance}.  Hence, the minimal point lies on the real axis
  between the right vertex and the origin and the maximal point lies on the line crossing the midpoint of two vertices and the origin. 

  If we add the origin $0$ as the $N+1$ point, then $d_{N+1}=|0-z|=\del$ for each $\tht$ and hence we only need to scale
  the upper and lower bounds by $\del$.
\end{proof} 

We will now ask for the case, where $z$ is placed on a circle with radius $\del$ around one vertex, see the dashed
circle in \figref{fig:smallestproductdistance}.


\begin{conjecture}\label{con:ngon}
  Let $N\geq 2$ be an integer and $\ome=e^{\im 2\pi/N}$ be the $N$th root of unity. Then the points $r\ome^n$ are the
  vertices of a regular $N-$gon inscribed in the circle of radius $r>0$ centered at the origin.  Let $0<\del\leq r\sin(\pi/N)$ and consider any $z$ on a
  circle $C_k(\del)$ of radius $\del$ with center at one vertex $r\ome^k$, then the minimal product of all distances
  from $z$ to the vertices is given by
 \begin{align}
   \min_{z\in C_k(\del)} \Pro_n |z-r\ome^n| = r^N- (r-\del)^N.\label{eq:zeroonvertex}
 \end{align}
\end{conjecture}
%
\begin{remark}
If we add the origin of the $N-$gon as the $N+1$ point and demand $\del<r\sin(\pi/N)$, then we
get for the minimum of its product distances, by taking $k=0$
\begin{align}
  \Pro_{n=1}^{N+1} d_n \geq \min d_{N+1} \cdot \min \Pro_{n=1}^6 d_n \geq (r-\del) (r^N-(1-\del)^N)\label{eq:polygonmin}
\end{align}
since at any point $z=r+\del e^{i\tht}$ the distance to the centroid  is 
\begin{align}
  d_{N+1}=|z-0|=|r+\del e^{i\tht}|\geq |r-\del|=r-\del
\end{align}
where equality only holds for  $\tht=\pi$, see \figref{fig:smallestproductdistance}.  Note,
$1-\sin(\pi/N)=\versine(\pi/2-\pi/N)=\versine((N-2)\pi/(2N))=2\sin^2((N-2)\pi/(4N))$ is monotone decreasing with
increasing $N$. Hence $1-(1-\sin(\pi/N))^N$ will monotone increase with $N$. Moreover, for $N\geq 10$ we get
\begin{align}
  1-(1-\sin(\pi/N))^N\simeq 1
\end{align}
\end{remark}

For $N=2$ the conjecture holds, since for $z=1+\del e^{i\tht}$ we have
\begin{align}
  |1+\del e^{i\tht} -1|\cdot|1+\del e^{i \tht} +1| = \del\cdot|2+\del^{i\tht}|\geq \del(2-\del)=1-(1-2\del
  +\del^2)=1-(1-\del)^2,\notag
\end{align}
where equality holds if and only if $\tht=\pi$.


  \subsection{Back to the Hexagon Lattice}
  If the \conref{con:ngon} would hold, we can similar argument as in \eqref{eq:polygonmin},
  that if  the origin is included as an $N+1$ point the minimum would
  be 
  \begin{align}
    \Pro_{n=1}^{N+1} d_n(\tht) \geq (r-\del)r^N(1-(1-\del)^N)
  \end{align}
  and  be achieved for $\tht=\pi$.
  Furthermore, if we chose the circle around the centroid we get  with \thmref{thm:mgonproddistance} 
  \begin{align}
    |\del|\Pro_{n=1}^N |\del e^{\im \tht}-\alp_n | \geq \del(r^N-\del^N)
 \end{align}
 Indeed, we can then show, that the later product distance is the smallest possible.
\begin{lemi}
Let $r>0$ and  $N\geq 2$. Then it holds for any $\del\in(0,r/2)$ 
\begin{align}
  \del(r^N-\del^N) < (r^N-(r-\del)^N)\cdot(r-\del).\label{eq:conjlowerbound}
\end{align}
\end{lemi}
\begin{proof}
To show this, we only need to verify for $r=1$ and $\del\in(0,1/2)$, i.e., 
\begin{align}
   \del(1-\del^N)&<(1-(1-\del)^N)(1-\del)\\
\LRA \quad a_N=\del-\del^{N+1} &< (1-\del)-(1-\del)^{N+1} =b_N\label{eq:einsdel}
\end{align}
For $\del=1/2$ and $\del=0$ this becomes equality.
We will prove the strict inequality by induction. For $N\geq 2$ we get
\begin{align}
  b_2= 1-\del -(1-3\del+3\del^2-\del^3)=2\del -3\del^2+\del^3=b_2 + \del +2\del^3 -3\del^2\label{eq:indstart}
\end{align}
But for the last three terms it holds
\begin{align}
   \del+2\del^3-3\del^2>0 \quad\LRA \quad 1+2\del^2>3\del 
\quad \LRA \quad  \del^{-1} + 2\del>2+1=3.
\end{align}
Note, for $N=1$ the strict inequality \eqref{eq:einsdel} does not hold, since $\del(1-\del)=(1-1+\del)(1-\del)$.
We will now show that \eqref{eq:einsdel} holds for $N+1$. From \eqref{eq:einsdel} we get
\begin{align*}
  a_{N+1}& =\del-\del^{N+2} =\del(\del-\del^{N+1})-\del^2+\del=\del a_N + (1-\del) -(1-\del)^2\\
  b_{N+1}& = (1-\del)-(1-\del)^{N+2}=(1-\del)b_N+(1-\del)-(1-\del)^2>\del b_N +(1-\del)-(1-\del)^2
\end{align*} 
where we used $1-\del>\del$ for each $0\leq\del<1/2$.
Hence it holds $b_{N+1}>a_{N+1}$ if $b_N>a_N$ holds. By induction and \eqref{eq:indstart} this holds for all
$N\geq 2$.
\end{proof}

The lower bound \eqref{eq:conjlowerbound} is monotone increasing for $\del\in[0,r/2]$ and achieves its maximum at the
boundary $\del=r/2$ given by
\begin{align}
  0\leq \del(r^N-\del^N)\leq \frac{r^{N+1}}{2}\frac{(2^N-1)}{2^{N}}<\frac{r^{N+1}}{2}
\end{align}
see \figref{fig:glowerbound_centroid} for the hexagon, $N=6$ and $r=1$.

\newcommand{\radius}{\ensuremath{r}}
\newcommand{\thtthree}{\ensuremath{5\pi/6}}
\newcommand{\dseven}{\ensuremath{d_7}}
\newcommand{\bone}{\ensuremath{b_1}}
\newcommand{\btwo}{\ensuremath{b_2}}
\newcommand{\twodel}{\ensuremath{\sqrt{2}\del}}
\newcommand{\bonehalf}{\ensuremath{\frac{b_2}{2}}}
\newcommand{\bthree}{\ensuremath{b_3}}
\newcommand{\bfour}{\ensuremath{b_4}}
\newcommand{\edge}{\ensuremath{a}}
\newcommand{\edgehalf}{\ensuremath{\delmax=\frac{a}{2}}}
\newcommand{\aside}{\ensuremath{c}}
\newcommand{\piquarter}{\ensuremath{\pi/4}}
\newcommand{\alpM}{\ensuremath{\alp_{M+1}}}
\newcommand{\alpMone}{\ensuremath{\alp_{M+2}}}
\newcommand{\alptwoM}{\ensuremath{\alp_{2M+1}}}
\newcommand{\alpthreeM}{\ensuremath{\alp_{3M+1}}}

\begin{figure}[t]
  \centering
    \begin{subfigure}{0.485\textwidth} 
      \def\svgwidth{1.15\textwidth} \small{
        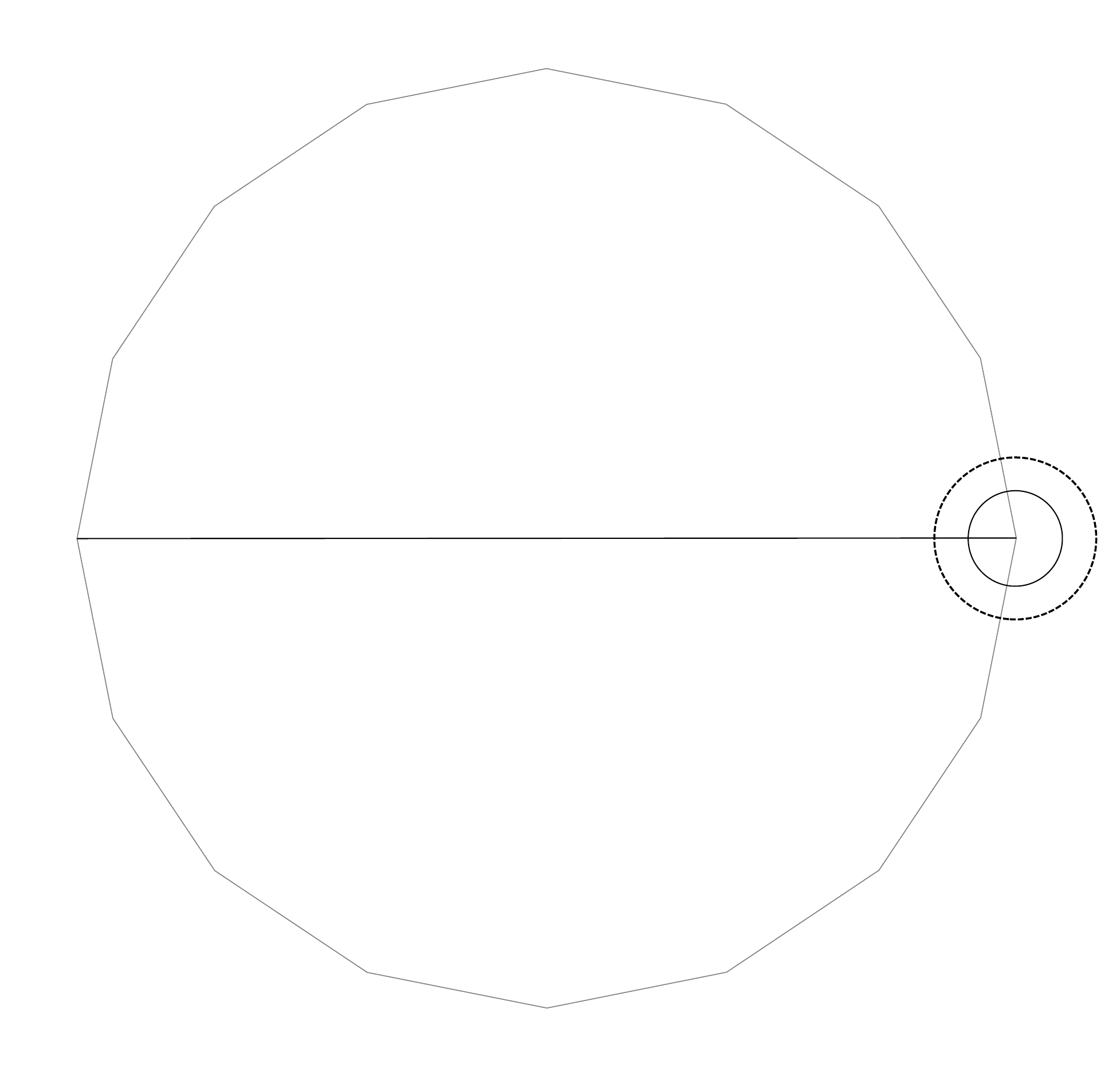 } 
        \vspace{-0.9cm}
      \end{subfigure}
      \caption{Lower bound conjecture for regular $N-$gon with circle point around one vertex.} \label{fig:ngon_lb_vertex}
\end{figure}

We will now lower bound the product distances  in \conref{con:ngon} for  $N-$gons with circle points around one vertex,
by using the geometric relaxation given in \figref{fig:ngon_lb_vertex}.

\begin{lemi}\label{eq:gHuffman}
  Consider a regular $N-$gon with  $N=4M$  for any $3\leq M\in\N$ inscribed in a circle of radius $r>0$. Consider a
  point $z$ on a circle with radius $\del<r\sin(\pi/N)$ and center $z=r+\del e^{i\tht}$. Then the minimal product distances
  for all $\tht$ is bounded by
  \begin{align}
    g_{v,N}(\tht)\geq \Pro_n d_n 
    \geq 
    \del(2r-\del) \Pro_{m=1}^{M}(1+\sin\frac{\pi}{2M}-2\tdel(\tdel+1 +\sin\frac{\pi}{2M}))\cdot
    r^4\sin^2\frac{\pi}{4M} (\Pro_{m} 2m-1)^2
\label{eq:huffmanMgonbound2}
  \end{align}
\end{lemi}
\begin{proof}
  Lets define $\ome=e^{\im 2\pi/N}$. We will use two reference points, at $r-2\del$ and at $r$, to lower bound the
  distances to the vertices. The special distance 
  \begin{align}
     d_{2M+1}=| \alp_{2M+1}-z|\geq 2r-\del\quad,\quad   d_1 =\del
  \end{align}
  For all vertices in the left half plane we will use the radius of the smallest circle given by
  \begin{align}
    b_{M+2}=\sqrt{(r+c-\sqrt{2}\del)^2+(r-h)^2}
  \end{align}
  where
  \begin{align}
    c=r\sin(\gam)=r\sin(\pi/2M) \quad,\quad h=r(1-\cos(\gam)),\quad,\quad a=2r\sin(\pi/4M). \label{eq:edge}
  \end{align}
  which gives
  \begin{align}
    b_{M+2} &=\sqrt{r^2(1+\sin(\pi/2M)-2\tdel)^2+r^2(1-(1-\cos(\pi/2M)))^2}\\
    &=r\sqrt{2}\sqrt{(1+\sin(\pi/2M) +\tdel^2-\sqrt{2}\tdel(1+\sin(\pi/2M))}
  \end{align}
  The distances in the first quadrant we will lower bound by multiples of $\bonehalf$ given as
  \begin{align}
    \bonehalf=\edge\cdot\cos(\pi/4)/2=2r\sin(\pi/4M)\frac{\sqrt{2}}{4} =r\sin(\pi/4M) \frac{1}{\sqrt{2}}
  \end{align}
  Then we get for the product of all distances the bound
  %
  \begin{align}
    g_{v,4M} &= \Pro_n d_n \geq b_{M+2}^{2M-2} \cdot d_1\cdot d_{2M+1}\cdot\Pro_{m=1}^{M} ((2m-1)\bonehalf)^2 \\
    & =\del(2r-\del) \Pro_{m=1}^{M}2r^2 (1+\sin(\pi/2M)-2\tdel(\tdel+1 +\sin(\pi/2M)))\cdot
    (2m-1)^2\frac{r^2\sin^2(\pi/4M)}{2}\notag\\
    & =\del(2r-\del) \Pro_{m=1}^{M}(1+\sin\frac{\pi}{2M}-2\tdel(\tdel+1 +\sin\frac{\pi}{2M}))\cdot
    r^4\sin^2\frac{\pi}{4M} (\Pro_{m} 2m-1)^2
  \end{align}
\end{proof}

\begin{figure}[t]
  \centering
  \begin{subfigure}[b]{0.485\textwidth} 
  \hspace{0.2cm}
  \includegraphics[width=\textwidth]{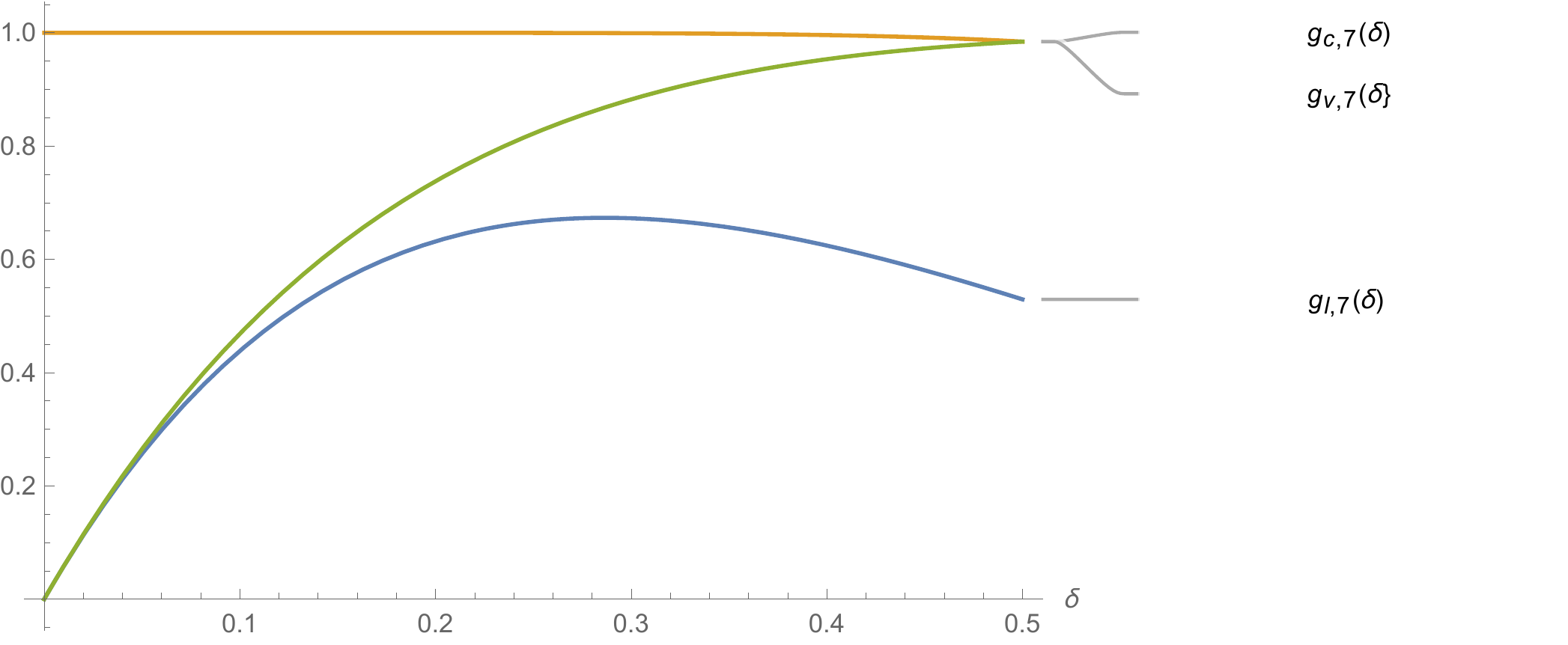}
    \caption{Bounds for $r=1$ without centroid.}\label{fig:glowerbound}
  \end{subfigure} 
  \begin{subfigure}[b]{0.485\textwidth} 
    \centering
    \includegraphics[width=0.7\textwidth]{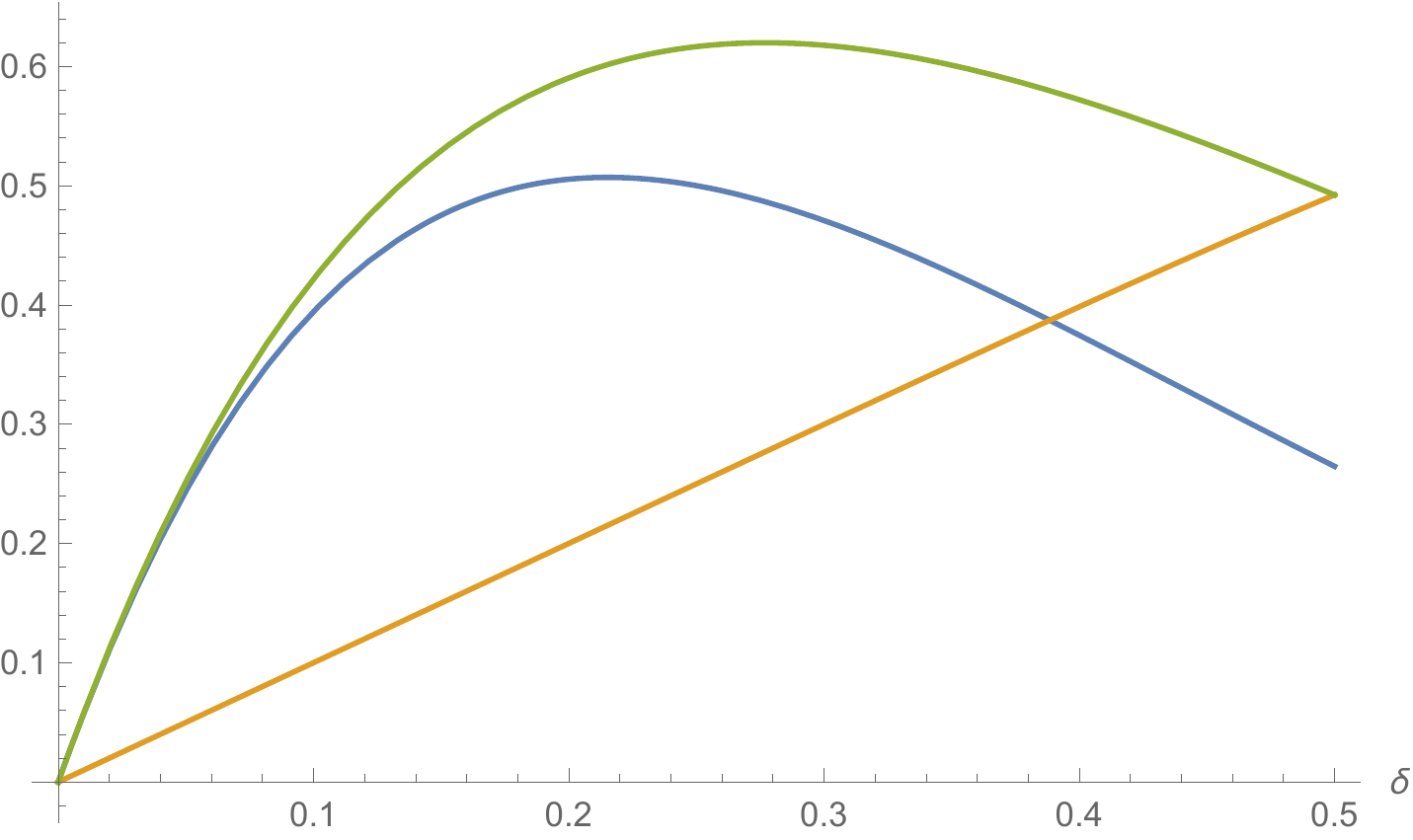}
\caption{Bounds with centroid.}
\label{fig:glowerbound_centroid}
 \end{subfigure} 
\caption{Lower Bound for Hexagon Conjecture (green line). 
  Yellow line for bound the $g_{c}$ at centroid circle.}
\end{figure}


%% file: HexagonalCirclePack_proddist2.pdf_tex
\begingroup%
  \makeatletter%
  \providecommand\color[2][]{%
    \errmessage{(Inkscape) Color is used for the text in Inkscape, but the package 'color.sty' is not loaded}%
    \renewcommand\color[2][]{}%
  }%
  \providecommand\transparent[1]{%
    \errmessage{(Inkscape) Transparency is used (non-zero) for the text in Inkscape, but the package 'transparent.sty' is not loaded}%
    \renewcommand\transparent[1]{}%
  }%
  \providecommand\rotatebox[2]{#2}%
  \ifx\svgwidth\undefined%
    \setlength{\unitlength}{435.0419031bp}%
    \ifx\svgscale\undefined%
      \relax%
    \else%
      \setlength{\unitlength}{\unitlength * \real{\svgscale}}%
    \fi%
  \else%
    \setlength{\unitlength}{\svgwidth}%
  \fi%
  \global\let\svgwidth\undefined%
  \global\let\svgscale\undefined%
  \makeatother%
  \begin{picture}(1,0.7743401)%
    \put(0,0){\includegraphics[width=\unitlength,page=1]{HexagonalCirclePack_proddist2.pdf}}%
    \put(0.63906926,0.09529774){\color[rgb]{0,0,0}\makebox(0,0)[lb]{\smash{\dmin}}}%
    \put(0.8662786,0.33092089){\color[rgb]{0,0,0}\makebox(0,0)[lb]{\smash{\del}}}%
    \put(0,0){\includegraphics[width=\unitlength,page=2]{HexagonalCirclePack_proddist2.pdf}}%
    \put(0.76528598,0.34025069){\color[rgb]{0,0,0}\makebox(0,0)[lb]{\smash{\done}}}%
    \put(0.70898066,0.42961965){\color[rgb]{0,0,0}\makebox(0,0)[lb]{\smash{\dtwo}}}%
    \put(0.53595895,0.48150937){\color[rgb]{0,0,0}\makebox(0,0)[lb]{\smash{\dthree}}}%
    \put(0.42022566,0.33651933){\color[rgb]{0,0,0}\makebox(0,0)[lb]{\smash{\dfour}}}%
    \put(0.53389211,0.15274202){\color[rgb]{0,0,0}\makebox(0,0)[lb]{\smash{\dfive}}}%
    \put(0.7096996,0.19447384){\color[rgb]{0,0,0}\makebox(0,0)[lb]{\smash{\dsix}}}%
    \put(0,0){\includegraphics[width=\unitlength,page=3]{HexagonalCirclePack_proddist2.pdf}}%
    \put(0.62497494,0.30452362){\color[rgb]{0,0,0}\makebox(0,0)[lb]{\smash{\tht}}}%
    \put(0.6385685,0.26960867){\color[rgb]{0,0,0}\makebox(0,0)[lb]{\smash{\aradius}}}%
    \put(0,0){\includegraphics[width=\unitlength,page=4]{HexagonalCirclePack_proddist2.pdf}}%
    \put(0.57843075,0.2995732){\color[rgb]{0,0,0}\makebox(0,0)[lb]{\smash{\null}}}%
    \put(0,0){\includegraphics[width=\unitlength,page=5]{HexagonalCirclePack_proddist2.pdf}}%
  \end{picture}%
\endgroup%

%% file: HexagonalCirclePack_worstHuffman.pdf_tex
\begingroup%
  \makeatletter%
  \providecommand\color[2][]{%
    \errmessage{(Inkscape) Color is used for the text in Inkscape, but the package 'color.sty' is not loaded}%
    \renewcommand\color[2][]{}%
  }%
  \providecommand\transparent[1]{%
    \errmessage{(Inkscape) Transparency is used (non-zero) for the text in Inkscape, but the package 'transparent.sty' is not loaded}%
    \renewcommand\transparent[1]{}%
  }%
  \providecommand\rotatebox[2]{#2}%
  \ifx\svgwidth\undefined%
    \setlength{\unitlength}{609.1590627bp}%
    \ifx\svgscale\undefined%
      \relax%
    \else%
      \setlength{\unitlength}{\unitlength * \real{\svgscale}}%
    \fi%
  \else%
    \setlength{\unitlength}{\svgwidth}%
  \fi%
  \global\let\svgwidth\undefined%
  \global\let\svgscale\undefined%
  \makeatother%
  \begin{picture}(1,0.85322456)%
    \put(0,0){\includegraphics[width=\unitlength,page=1]{HexagonalCirclePack_worstHuffman.pdf}}%
    \put(0.6004827,0.40242474){\color[rgb]{0,0,0}\makebox(0,0)[lb]{\smash{\Rinv}}}%
    \put(0.48464463,0.40304688){\color[rgb]{0,0,0}\makebox(0,0)[lb]{\smash{\null}}}%
    \put(0,0){\includegraphics[width=\unitlength,page=2]{HexagonalCirclePack_worstHuffman.pdf}}%
    \put(0.58314102,0.4405986){\color[rgb]{0,0,0}\makebox(0,0)[lb]{\smash{\zdel}}}%
    \put(0,0){\includegraphics[width=\unitlength,page=3]{HexagonalCirclePack_worstHuffman.pdf}}%
  \end{picture}%
\endgroup%

%% file: Ngon_vertex_lowerbound_new3.pdf_tex
\begingroup%
  \makeatletter%
  \providecommand\color[2][]{%
    \errmessage{(Inkscape) Color is used for the text in Inkscape, but the package 'color.sty' is not loaded}%
    \renewcommand\color[2][]{}%
  }%
  \providecommand\transparent[1]{%
    \errmessage{(Inkscape) Transparency is used (non-zero) for the text in Inkscape, but the package 'transparent.sty' is not loaded}%
    \renewcommand\transparent[1]{}%
  }%
  \providecommand\rotatebox[2]{#2}%
  \ifx\svgwidth\undefined%
    \setlength{\unitlength}{562.50894851bp}%
    \ifx\svgscale\undefined%
      \relax%
    \else%
      \setlength{\unitlength}{\unitlength * \real{\svgscale}}%
    \fi%
  \else%
    \setlength{\unitlength}{\svgwidth}%
  \fi%
  \global\let\svgwidth\undefined%
  \global\let\svgscale\undefined%
  \makeatother%
  \begin{picture}(1,0.97417185)%
    \put(0,0){\includegraphics[width=\unitlength,page=1]{Ngon_vertex_lowerbound_new3.pdf}}%
    \put(0.90145961,0.4652397){\color[rgb]{0,0,0}\rotatebox{-0.04004112}{\makebox(0,0)[lb]{\smash{\alpone}}}}%
    \put(0.91091037,0.50793996){\color[rgb]{0,0,0}\rotatebox{-0.40123029}{\makebox(0,0)[lb]{\smash{\del}}}}%
    \put(0,0){\includegraphics[width=\unitlength,page=2]{Ngon_vertex_lowerbound_new3.pdf}}%
    \put(0.77607803,0.5331503){\color[rgb]{0.2,0.4,0.8}\makebox(0,0)[lb]{\smash{\bonehalf}}}%
    \put(0.69988584,0.64901787){\color[rgb]{0,0,0}\makebox(0,0)[lb]{\smash{\bthree}}}%
    \put(0,0){\includegraphics[width=\unitlength,page=3]{Ngon_vertex_lowerbound_new3.pdf}}%
    \put(0.87084738,0.6655936){\color[rgb]{0,0,0}\makebox(0,0)[lb]{\smash{\alptwo}}}%
    \put(0.49169967,0.94598076){\color[rgb]{0,0,0}\makebox(0,0)[lb]{\smash{\alpM}}}%
    \put(0,0){\includegraphics[width=\unitlength,page=4]{Ngon_vertex_lowerbound_new3.pdf}}%
    \put(0.48195842,0.02734303){\color[rgb]{0,0,0}\makebox(0,0)[lb]{\smash{\alpthreeM}}}%
    \put(-0,0.45561232){\color[rgb]{0,0,0}\makebox(0,0)[lb]{\smash{\alptwoM}}}%
    \put(0,0){\includegraphics[width=\unitlength,page=5]{Ngon_vertex_lowerbound_new3.pdf}}%
    \put(0.79412327,0.45231447){\color[rgb]{0,0,0}\rotatebox{-0.40123029}{\makebox(0,0)[lb]{\smash{\twodel}}}}%
    \put(0,0){\includegraphics[width=\unitlength,page=6]{Ngon_vertex_lowerbound_new3.pdf}}%
    \put(0.59287717,0.75643538){\color[rgb]{0,0,0}\makebox(0,0)[lb]{\smash{\bfour}}}%
    \put(0,0){\includegraphics[width=\unitlength,page=7]{Ngon_vertex_lowerbound_new3.pdf}}%
    \put(0.45472276,0.88756461){\color[rgb]{0,0,0}\makebox(0,0)[lb]{\smash{\hoehe}}}%
    \put(0.4107026,0.85249046){\color[rgb]{0,0,0}\makebox(0,0)[lb]{\smash{\aside}}}%
    \put(0,0){\includegraphics[width=\unitlength,page=8]{Ngon_vertex_lowerbound_new3.pdf}}%
    \put(0.45025959,0.60571228){\color[rgb]{0,0,0}\makebox(0,0)[lb]{\smash{\gam}}}%
    \put(0,0){\includegraphics[width=\unitlength,page=9]{Ngon_vertex_lowerbound_new3.pdf}}%
    \put(0.40147004,0.90864727){\color[rgb]{0,0,0}\makebox(0,0)[lb]{\smash{\edge}}}%
    \put(0.24523558,0.45509598){\color[rgb]{0,0,0}\makebox(0,0)[lb]{\smash{\radius}}}%
    \put(0.85131736,0.58641371){\color[rgb]{0,0,0}\makebox(0,0)[lb]{\smash{\piquarter}}}%
    \put(0.93056361,0.56418512){\color[rgb]{0,0,0}\makebox(0,0)[lb]{\smash{\edgehalf}}}%
    \put(0.23080884,0.91056154){\color[rgb]{0,0,0}\makebox(0,0)[lb]{\smash{\alpMone}}}%
    \put(0.49720966,0.69263757){\color[rgb]{0,0,0}\makebox(0,0)[lb]{\smash{\radius}}}%
    \put(0.91968172,0.62688962){\color[rgb]{0,0,0}\makebox(0,0)[lb]{\smash{\edge}}}%
    \put(0,0){\includegraphics[width=\unitlength,page=10]{Ngon_vertex_lowerbound_new3.pdf}}%
  \end{picture}%
\endgroup%